\documentclass[runningheads,a4paper]{llncs}

\usepackage{amsmath,amssymb}
\usepackage{mathrsfs}
\usepackage{pifont}
\usepackage[noend]{algpseudocode}
\usepackage{threeparttable}
\usepackage{booktabs}
\usepackage{multirow}
\usepackage{xcolor}
\usepackage{bbding}
\usepackage{stmaryrd}
\usepackage{soul}

\usepackage{algorithm}

\usepackage{verbatim}
\usepackage{bm}
\usepackage{moreverb}
\usepackage{stfloats}
\usepackage{enumerate}
\usepackage{fancybox}
\usepackage{dashbox}
\usepackage[colorlinks, linkcolor=blue, citecolor=blue]{hyperref}
\usepackage{epsfig,graphics,url,color}
\usepackage{booktabs}
\usepackage{url}
\usepackage{multicol}
\usepackage{empheq}
\newcounter{rowno}
\allowdisplaybreaks[4]
\usepackage[font=footnotesize,labelfont=bf,justification=centering]{caption}
\usepackage{paralist}
\usepackage{enumitem}
\usepackage{float}
\usepackage{graphicx}
\usepackage[normalem]{ulem}

\usepackage{subfig}

\usepackage{color}

\newcommand{\setup}{\textsf{\textup{Setup}}}

\newcommand{\kg}{\textsf{\textup{KG}}}

\newcommand{\sks}{sk_\textup{s}}
\newcommand{\pks}{pk_\textup{s}}

\newcommand{\Z}{\mathbb{Z}_p^*}
 \newcommand{\G}{\mathbb{G}}

\newcommand{\RS}{\mathcal{RS}}

\newcommand{\A}{\mathcal{A}}
\newcommand{\B}{\mathcal{B}}

\newcommand{\M}{\mathcal{M}}
\newcommand{\negl}{\textup{\textsf{negl}}(\lambda)}

\newcommand{\pp}{\textup{\textsf{pp}}}

\newcommand{\spk}{\textup{\textsf{spk}}}
\newcommand{\ssk}{\textup{\textsf{ssk}}}
\newcommand{\vpk}{\textup{\textsf{vpk}}}
\newcommand{\vsk}{\textup{\textsf{vsk}}}
\newcommand{\skg}{\textup{\textsf{SignKG}}}
\newcommand{\vkg}{\textup{\textsf{VerKG}}}
\newcommand{\sign}{\textup{\textsf{Sign}}}
\newcommand{\verify}{\textup{\textsf{Verify}}}

\newcommand{\forgeds}{\textup{\textsf{FgeDS}}}
\newcommand{\forgeas}{\textup{\textsf{FgeAS}}}
\newcommand{\scor}{S_{\textup{cor}}}
\newcommand{\mdvs}{\textup{\textsf{MDVS}}}

\newcommand{\orawk}{\mathcal{O}_{\text{\emph{WK}}}}
\newcommand{\oradk}{\mathcal{O}_{\text{\emph{DK}}}}
\newcommand{\orawpk}{\mathcal{O}_{\text{\emph{WPK}}}}

\newcommand{\oradpk}{\mathcal{O}_{\text{\emph{DPK}}}}
\newcommand{\orawatm}{\mathcal{O}_{\text{\emph{W}}}}
\newcommand{\oradet}{\mathcal{O}_{\text{\emph{D}}}}
\newcommand{\gamecons}{\textup{\textbf{G}}^{\textup{cons}}_{\mddw,\A}(\lambda)}
\newcommand{\advcons}{\textup{\textbf{Adv}}^{\textup{cons}}_{\mddw,\A}(\lambda)}
\newcommand{\gameunf}{\textup{\textbf{G}}^{\textup{sound}}_{\mddw,\A}(\lambda)}
\newcommand{\advunf}{\textup{\textbf{Adv}}^{\textup{sound}}_{\mddw,\A}(\lambda)}

\newcommand{\poly}{\textsf{poly}(\lambda)}

\newcommand{\tg}{\mathcal{T}}
\newcommand{\model}{\textsf{\upshape{Model}}}
\newcommand{\genmodel}{\textsf{\upshape{GenModel}}}
\newcommand{\decode}{\textsf{\upshape{Decode}}}
\newtheorem{assumption}{Assumption}

\newcommand{\mddw}{\textup{\textsf{MDDW}}}
\newcommand{\ddw}{\textup{\textsf{DDW}}}
\newcommand{\wpk}{\textup{\textsf{wpk}}}
\newcommand{\wsk}{\textup{\textsf{wsk}}}
\newcommand{\dpk}{\textup{\textsf{dpk}}}
\newcommand{\dsk}{\textup{\textsf{dsk}}}
\newcommand{\wkg}{\textup{\textsf{WatKG}}}
\newcommand{\dkg}{\textup{\textsf{DetKG}}}
\newcommand{\watm}{\textup{\textsf{WatMar}}}
\newcommand{\detect}{\textup{\textsf{Detect}}}
\newcommand{\D}{\mathcal{D}}

\newcommand{\gamedistor}{\textup{\textbf{G}}^{\textup{dist-fr}}_{\mddw,\A}(\lambda)}
\newcommand{\advdistor}{\textup{\textbf{Adv}}^{\textup{dist-fr}}_{\mddw,\A}(\lambda)}
\newcommand{\gamerob}{\textup{\textbf{G}}^{\textup{rob}}_{\mddw,\A}(\lambda)}
\newcommand{\advrob}{\textup{\textbf{Adv}}^{\textup{rob}}_{\mddw,\A}(\lambda)}

\newcommand{\gameotrdsmddw}{\textup{\textbf{G}}^{\textup{otr-ds}}_{\mddw,\A,\forgeds}(\lambda)}
\newcommand{\advotrdsmddw}{\textup{\textbf{Adv}}^{\textup{otr-ds}}_{\mddw,\A,\forgeds}(\lambda)}
\newcommand{\gameotrasmddw}{\textup{\textbf{G}}^{\textup{otr-as}}_{\mddw,\A,\forgeas}(\lambda)}
\newcommand{\advotrasmddw}{\textup{\textbf{Adv}}^{\textup{otr-as}}_{\mddw,\A,\forgeas}(\lambda)}
\newcommand{\gameclaimunfmddw}{\textup{\textbf{G}}^{\textup{clm-unf}}_{\mddw,\A}(\lambda)}
\newcommand{\advclaimunfmddw}{\textup{\textbf{Adv}}^{\textup{clm-unf}}_{\mddw,\A}(\lambda)}
\newcommand{\gameunframemddw}{\textup{\textbf{G}}^{\textup{non-fram}}_{\mddw,\A}(\lambda)}
\newcommand{\advunframemddw}{\textup{\textbf{Adv}}^{\textup{non-fram}}_{\mddw,\A}(\lambda)}

\newcommand{\gameconsmdvs}{\textup{\textbf{G}}^{\textup{cons}}_{\mdvs,\A}(\lambda)}
\newcommand{\advconsmdvs}{\textup{\textbf{Adv}}^{\textup{cons}}_{\mdvs,\A}(\lambda)}
\newcommand{\gameunfmdvs}{\textup{\textbf{G}}^{\textup{unforg}}_{\mdvs,\A}(\lambda)}
\newcommand{\advunfmdvs}{\textup{\textbf{Adv}}^{\textup{unforg}}_{\mdvs,\A}(\lambda)}

\newcommand{\gameotrdsmdvs}{\textup{\textbf{G}}^{\textup{otr-ds}}_{\mdvs,\A,\forgeds}(\lambda)}
\newcommand{\advotrdsmdvs}{\textup{\textbf{Adv}}^{\textup{otr-ds}}_{\mdvs,\A,\forgeds}(\lambda)}
\newcommand{\gameotrasmdvs}{\textup{\textbf{G}}^{\textup{otr-as}}_{\mdvs,\A,\forgeas}(\lambda)}
\newcommand{\advotrasmdvs}{\textup{\textbf{Adv}}^{\textup{otr-as}}_{\mdvs,\A,\forgeas}(\lambda)}

\newcommand{\gameclaimunfmdvs}{\textup{\textbf{G}}^{\textup{clm-unf}}_{\mdvs,\A}(\lambda)}
\newcommand{\advclaimunfmdvs}{\textup{\textbf{Adv}}^{\textup{clm-unf}}_{\mdvs,\A}(\lambda)}
\newcommand{\gameunframemdvs}{\textup{\textbf{G}}^{\textup{non-fram}}_{\mdvs,\A}(\lambda)}
\newcommand{\advunframemdvs}{\textup{\textbf{Adv}}^{\textup{non-fram}}_{\mdvs,\A}(\lambda)}

\newcommand{\gamepseudodvs}{\textup{\textbf{G}}^{\textup{ps-rand}}_{\dvs,\A}(\lambda)}
\newcommand{\advpseudodvs}{\textup{\textbf{Adv}}^{\textup{ps-rand}}_{\dvs,\A}(\lambda)}

\newcommand{\orask}{\mathcal{O}_{\text{\emph{SK}}}}
\newcommand{\oravk}{\mathcal{O}_{\text{\emph{VK}}}}
\newcommand{\oraspk}{\mathcal{O}_{\text{\emph{SPK}}}}
\newcommand{\oravpk}{\mathcal{O}_{\text{\emph{VPK}}}}
\newcommand{\orasign}{\mathcal{O}_{\text{\emph{S}}}}
\newcommand{\oraver}{\mathcal{O}_{\text{\emph{V}}}}
\newcommand{\oraclm}{\mathcal{O}_{\text{\emph{Clm}}}}

\newcommand{\hashf}{\textsf{\textup{H}}}
\newcommand{\genG}{\textsf{\textup{GenG}}}

\newcommand{\sig}{\textsf{\textup{Sig}}}

\newcommand{\pkvi}{pk_{\textup{v}_i}}
\newcommand{\Sett}{\mathsf{S}}
\newcommand{\dvs}{\textsf{\textup{DVS}}}

\newcommand{\cmdvs}{\textsf{\textup{CMDVS}}}
\newcommand{\commit}{\textsf{\textup{Commit}}}
\newcommand{\Com}{\textsf{\textup{Com}}}
\newcommand{\Dec}{\textsf{\textup{Decom}}}
\newcommand{\rcom}{r_\textsf{\textup{com}}}
\newcommand{\com}{\textsf{\textup{com}}}

\newcommand{\prf}{\textsf{\textup{PRF}}}
\newcommand{\eval}{\textsf{\textup{Eval}}}
\newcommand{\cla}{\textsf{\textup{Claim}}}
\newcommand{\claver}{\textsf{\textup{ClmVer}}}
\newcommand{\sgsk}{\ssk_{\sig}}
\newcommand{\sgpk}{\spk_{\sig}}
\newcommand{\rssk}{\ssk_{\mdvs}}
\newcommand{\rspk}{\spk_{\mdvs}}
\newcommand{\rvsk}{\vsk_{\mdvs}}
\newcommand{\rvpk}{\vpk_{\mdvs}}
\newcommand{\sgsks}{\ssk_{\sig,i}}
\newcommand{\sgpks}{\spk_{\sig,i}}
\newcommand{\rssks}{\ssk_{\mdvs,i}}
\newcommand{\rspks}{\spk_{\mdvs,i}}

\newcommand{\ks}{k_{i}}

\newcommand{\dec}{\textup{\textsf{Decom}}}

\newenvironment{myitemize}{
  \begin{list}{$\bullet$}{
    \setlength{\leftmargin}{1em}
    \setlength{\itemindent}{0em}
    \setlength{\parsep}{0pt}
    \setlength{\topsep}{0pt}
    \setlength{\partopsep}{0pt}
    \setlength{\parskip}{0pt}
  }
}{
  \end{list}
}

\newcommand{\Evt}[1]{\textsf{\textup{Evt}}_{#1}}

\newcommand{\autom}{\textsf{\textup{Model}}}

\newcommand{\pdws}{\textsf{\textup{PDW}}}

\begin{document}

\mainmatter  

\title{Multi-Designated Detector Watermarking for Language Models}


\author{
	Zhengan Huang 
	\and
	Gongxian Zeng
	\and
	Xin Mu
	\and 
	Yu Wang
	\and
	Yue Yu
}

\authorrunning{Z. Huang et al.}
%
\institute{
	Pengcheng Laboratory, 
	Shenzhen, China \\
	\email{zhahuang.sjtu@gmail.com, gxzeng@cs.hku.hk, mux@pcl.ac.cn, wangy12@pcl.ac.cn, yuy@pcl.ac.cn 
	}
}

%
%




%
%

\maketitle

\begin{abstract}
	In this paper, we initiate the study of  \emph{multi-designated detector watermarking (MDDW)} for large language models (LLMs). This technique allows model providers to generate watermarked outputs from LLMs with two key properties: (i) only specific, possibly multiple, designated detectors can identify the watermarks, and (ii) there is no perceptible degradation in the output quality for ordinary users. We formalize the security definitions for MDDW and present a framework for constructing MDDW for any LLM using multi-designated verifier signatures (MDVS). Recognizing the significant economic value of LLM outputs, we introduce claimability  as an optional security feature for MDDW, enabling model providers to assert ownership of LLM outputs within designated-detector settings. To support claimable MDDW, we  propose a generic transformation converting any MDVS to a claimable MDVS. Our implementation of the MDDW scheme highlights its advanced functionalities and flexibility over existing methods, with satisfactory performance metrics.
\end{abstract}


\section{Introduction}
\label{sec:intro}
Generative artificial intelligence (AI) technique, e.g., large language models (LLMs), has been widely adopted in the field of language generation and has achieved excellent performance in a variety of downstream tasks. These tasks span from machine translation \cite{hendy2023good}, dialogue system \cite{hudevcek2023llms} to code generation \cite{ni2023lever} and medicine \cite{thirunavukarasu2023large}.

However, the abuse of LLMs may lead to several potential harms, including the generation of fake news \cite{huang2024fakegpt} and instances of academic dishonesty, such as cheating on writing and coding assignments \cite{KGW+23}. Another potential risk is that the proliferation of data fabricated by LLMs complicates the acquisition of superior models, as this data is not sourced from the real world and has to be excluded before training \cite{radford2023robust}.

Thus, a crucial challenge lies in \emph{distinguishing between texts generated by LLMs and those written by humans.}

Currently, the main approach to address the above issue is to training another AI  model for detection, such as GPTZero\footnote{GPTZero, \url{https://gptzero.me/}}. This approach hinges on the critical assumption that texts generated by LLMs exhibit unique characteristics that can be identified by AI. However, a significant flaw in this assumption lies in the fact that LLMs are deliberately engineered to produce content that is indistinguishable from human-created works. As a result, any ``black-box'' detection method is prone to \emph{high false positive and/or false negative rates} as LLMs become more realistic. Existing detectors (e.g., GPTZero, Detectgpt \cite{mitchell2023detectgpt}) do not offer guarantees of accuracy. In fact, there have already been instances where students were falsely accused, making headlines in the news \cite{fowler2023we,jimenez2023professors}.

Recently, some schemes with formal guarantees of negligible error bounds have been proposed.


For example, \cite{KGW+23} introduces watermarking for LLMs to achieve such formal guarantees. They demonstrate that a watermark can be embedded in the outputs of LLMs with large enough entropy. However, their watermarking scheme significantly alters the distribution of the generated texts and the watermark detection relies on this alteration. 
Thus, 
the methodology of the watermark detection in \cite{KGW+23} 
has led to \emph{degradation in the quality of the watermarked texts}.

To achieve the perfect quality (i.e., there is no degradation in the quality of watermarked outputs), \cite{CGZ23} proposes a method for embedding a watermark if the outputs of LLMs are sufficiently random. The main idea is that it generates the watermark with a    pseudorandom function (PRF) and utilizes the watermark (i.e., a pseudorandom number) to sample the outputs of LLMs. The perfect quality of the output is ensured by the \emph{undetectability}, the notion of which is formalized in \cite{CGZ23} and requires that without the knowledge of the secret key of the underlying PRF, the watermarked output is computationally indistinguishable from that output by the original LLM. 
However, the solution proposed in \cite{CGZ23}  has a  notable limitation. Both the watermark generation and detection require a secret key, rendering the scheme \emph{a privately detectable watermarking or a symmetric watermarking}. If the detection is desired to be outsourced, the secret key has to be shared to others, which could compromise its unforgeability. 


To address these challenges, 
\cite{FGJ+23} introduces an asymmetric solution called \emph{publicly detectable watermarking (PDW)},  
 where the generation of a watermark requires a secret key of the signature, and public detectability 
is achieved  through the public verification of the  signature. In \cite{FGJ+23}, a security  notion called \emph{distortion-freeness} is proposed for PDW, serving as the asymmetric version of undetectability from \cite{CGZ23}, to ensure that the watermarking scheme maintains the quality of the LLM output.  More specifically, distortion-freeness in \cite{FGJ+23} is defined as follows: ``without the secret watermarking key, no PPT machine can distinguish
plain LLM output from watermarked LLM output''. 
This definition aligns with the way undetectability is defined in \cite{CGZ23}. However, in the asymmetric setting, the public key is usually known to the public, and a distinguisher with the public key can easily find out whether the LLM output is watermarked or not, since the watermarking solution in \cite{FGJ+23} is publicly detectable. 
In other words, in the public-key setting,  the distortion-freeness defined in  \cite{FGJ+23} and the completeness  are contradictory. 
Furthermore, for PDW,  any third party can detect watermarks, which 
may compromise privacy and other interests  (e.g., economic interests). This universal detectability is often undesirable, particularly when restricted detectability is critical. 
\vspace{1mm}



\noindent\underline{\textbf{Contributions.}} In this paper, we initiate the study of \emph{multi-designated detector watermarking (MDDW)} for LLMs. Generally speaking, for the watermarked LLM outputs generated via  MDDW, multi-designated detectors are allowed to detect the watermarks, and all the other parties cannot distinguish the  watermarked outputs from the original LLM outputs.  The contributions of this paper is  summarized as follows. 

\begin{itemize}
	\item We introduce a new primitive called multi-designated detector watermarking (MDDW), and formalize its security notions. 
	\item We offer a framework for  constructing MDDW from any LLM and multi-designated verifier signature (MDVS), and show that it achieves the required security properties. 
	\item We provide a general method to transfer any MDVS to a claimable MDVS. Then, applying the above framework, we   obtain a claimable MDDW, which allows the model provider to provably claim that some candidate texts are indeed generated by the LLMs. 
	\item When considering only a single designated verifier, we present a more efficient concrete designated detector watermarking (DDW) in Appendix \ref{sec:efficientddw}, compared with the above solution from MDVS. We also provide a detailed experimental evaluation in Sec. \ref{sec:evaluation} to show that our schemes are practical. 
\end{itemize}

\noindent \underline{\textbf{MDDW primitive.}} In   MDDW, there are three kinds of roles involved: \emph{model providers}, \emph{designated detectors} and \emph{users}. \emph{Model providers} are the ones who  deliver   the LLM service to   users and   execute the watermarking scheme during the text generation phase.   \emph{Designated detectors} are the ones  responsible for detecting whether some text was  output by the model (by extracting and validating the watermark). \emph{Users} are the ones who use the model.

An MDDW scheme consists of five algorithms: a setup algorithm $\setup$, a key generation algorithm $\wkg$ for  model providers, a key generation algorithm $\dkg$ for designated detectors, a watermarking algorithm $\watm$ and a detection algorithm $\detect$.


Upon receiving the public parameter $\pp$ output by the setup algorithm $\setup$, the model provider can invoke $\wkg$ to generate its key pair $(\wpk, \wsk)$ for watermarking, and a designated detector can call $\dkg$ to generate its key pair $(\dpk, \dsk)$. 

The watermarking algorithm $\watm$, invoked by a model provider, takes the public parameter $\pp$, the model provider's secret watermarking key $\wsk_i$, the public keys of all the designated detectors $(\dpk_j)_{j \in S}$ (where $S$ is the index set of the designated detectors, and the same below), and a prompt $\bm{p}$ as input, and outputs a text $\bm{t}$ embedded with a watermark.

The detection algorithm $\detect$, invoked by a designated detector, takes  the public parameter $\pp$, a public watermarking key $\wpk_i$, the designated detector's  secret key $\dsk_{j'}$, the public keys of all the designated detectors $(\dpk_j)_{j \in S}$ (where $j'\in S$), and a candidate watermarked text $\bm{t}$ as input, and outputs a bit $b$, indicating whether $\bm{t}$ was watermarked. 
 

Similar to \cite{CGZ23,FGJ+23}, the watermarking scheme just uses the LLM as a black box. Thus, it works without adopting a specific LLM or using specific configuration (e.g., the specific decoder algorithm, model parameters, etc). In addition, the detection also works without access to the LLMs or corresponding configurations. 

\vspace{1mm}
\noindent\underline{\textbf{MDDW security.}} It is important to note that,  compared with  the security notions of symmetric watermarking and publicly detectable watermarking, the security requirements in the multi-designated detector setting are much more complex. 

The first challenge is ensuring consistency. If a designated detector successfully detects a watermark in a given text, how can they be assured that other designated detectors will also be able to detect it successfully?  
Another challenge is preventing designated detectors from provably convincing ordinary users that a text output by the LLMs is watermarked. If they can do so, the multi-designated detector functionality would be rendered useless. Moreover, in certain scenarios, even if designated detectors are unable to convince ordinary users of the watermark, it is crucial that the model provider  can assert the truth.  This mechanism is very helpful in MDDW, as it allows model providers to claim copyrights on content such as novels, pictures, or videos created by LLMs, which can yield considerable economic gains.

In this paper, we require that MDDW should satisfy the following   properties.
\begin{itemize}
	\item \textbf{\emph{Completeness.}}  Completeness requires for any designated  detector set $S$ and any candidate text $\bm{t}$ generated for the designated  detectors in  $S$ with MDDW, as long as $\bm{t}$ is long enough (to embed a watermark), each detector in $S$ should be able to extract and validate the watermark from $\bm{t}$ successfully   using their individual secret detection keys. 
	\item \textbf{\emph{Consistency.}} Consistency aims to ensure that for any text $\bm{t}$, including those created maliciously,  two designated detectors with uncompromised secret keys should not yield different detection results. Furthermore, if one designated detector accepts that the text $\bm{t}$ is generated by the LLM, all other designated detectors in $S$ should obtain the same outcome.
	\item \textbf{\emph{Soundness.}} Soundness is formalized to ensure that no PPT adversary can forge a watermarked text in the name of some model provider, such that some designated detector, whose secret key is not compromised, would accept the text is output by the corresponding model owned by this model provider.
	\item \textbf{\emph{Distortion-freeness.}} As discussed before, the purpose of distortion-freeness is to guarantee that the watermarking scheme does not degrade the quality of the LMM output. Specifically, distortion-freeness for MDDW  formalized to require that  for any PPT distinguisher without the secret detection keys of the designated detectors in $S$, a watermarked text  generated for the designated detectors in $S$ by MDDW should be indistinguishable from   a text output by the original LLM.
	\item \textbf{\emph{Robustness.}} It is likely that the text obtained from MDDW is modified before publication. The watermark detector should be robust enough to  detect a watermark even if the text has been artificially altered, provided the text semantics are preserved.  However, in extreme cases where a significant portion of the text is modified, the watermark should normally become undetectable, as the text effectively becomes the adversary's creation. Thus, we use a ``soft'' definition of robustness. Informally, if a continuous segment of text with length $\delta$ remains unaltered, then it is highly unlikely that any designated detector will fail to detect the watermark. We term this property $\delta$-robust. 
	\item \textbf{\emph{Off-the-record property for designated set.}} This property is formalized  to ensure that a  text $\bm{t}$ can be simulated by all the designated detectors in $S$, such that  it is  indistinguishable from a watermarked text generated for them by MDDW, from the perspective of any third party, even if they possess all the secret detection keys of the designated detectors in $S$. Intuitively, this prevents   the designated detectors from  convincing a third party that the text $\bm{t}$ is watermarked by the model provider, as it can be simulated by the designated detectors themselves. Moreover,   this property allows   users to deny using the LLMs when confronted with third-party suspicions, even if some designated detectors have exposed their secret keys. 
\end{itemize}

In addition, we introduce two \emph{optional} security properties for MDDW: \emph{off-the-record property for any subset} and \emph{claimability}.
\begin{itemize}
	\item \textbf{\emph{Off-the-record property for any subset.}} This property  is similar to the  off-the-record property for designated set, with  key differences:    the plausible-looking text $\bm{t}$ can be simulated  by \emph{any subset} of the designated detectors in $S$, and it is indistinguishable from   a watermarked text generated for the designated detectors in $S$ by MDDW, from the perspective of any third party, even if they possess  the secret detection keys of \emph{the subset that produced $\bm{t}$}. It is crucial to note that the requirement of ``any subset'' enhances the meaningfulness of the off-the-record property. For example, if two designated detectors in $S$ are unable to communicate, the off-the-record property for  designated set loses its significance.   Thus, off-the-record for any subset is a stronger security notion. 
	\item \textbf{\emph{Claimability.}} 
	Claimability is designed to give the model provider a means to convincingly demonstrate  the public that specific watermarked texts were generated by the provider's LLM. 
	This requires two algorithms: $\cla$ and $\claver$. The model provider runs $\cla$ with their  secret watermarking key to create proof $\pi$ for some  watermarked text $\bm{t}$. The public then uses $\claver$ to verify this proof $\pi$. Claimability requires that (i) proofs created with $\cla$ must be correctly validated by $\claver$, 	(ii) only the model provider can create   proofs to claim the LLM's output, and (iii) no one can falsely accuse another provider of generating the watermarked text $\bm{t}$.
\end{itemize}


\noindent\underline{\textbf{MDDW construction.}} The following outlines a technical overview of our MDDW construction.  

\vspace{1mm}
\noindent \underline{\emph{Framework.}} We propose a framework for building MDDW, which applies to any LLM, based on a  multi-designated verifier signature (MDVS) scheme. Our method is inspired by \cite{CGZ23,FGJ+23}. 


Before providing the high-level description of our framework, we firstly abstract the LLMs as in  \cite{CGZ23,FGJ+23,CHS24}, disregarding the specifics of their implementations. LLMs have a ``vocabulary'' $\mathcal{T}$ consisting of words or word fragments called ``tokens''. These models employ neural networks to process and generate tokens, trained on varied datasets to learn and predict language patterns.  During text generation, an LLM takes a prompt $\bm{p}$ as input and produces a sequence of tokens $\bm{t}$, with each subsequent token predicted based on the preceding context. For notation, let $\autom$ denote an auto-regressive language model, which accepts a prompt $\bm{p} \in \mathcal{T}$ and previously generated tokens $\bm{t}$ as input, outputting subsequent tokens over $\mathcal{T}$. For any polynomial $n$, $\textsf{GenModel}_n$ iteratively invokes $\autom$ to generate $n$ tokens. 

Next, we briefly recall another building block of our framework,  MDVS. An MDVS scheme consists of five algorithms: a setup algorithm $\setup$, a signing key generation algorithm $\skg$ for signers, a verification key generation algorithm $\vkg$ for verifiers, a signing algorithm $\sign$ and a verification algorithm $\verify$.
In a nutshell, given a MDVS signature $\sigma$ generated for the designated verifiers in a set $S$ by $\sign$, only the designated verifiers in $S$ can check the validity of the signature with $\verify$.

Now, we show a high-level description of our framework as follows. For a visual depiction, please see Fig. \ref{fig:MDDW_framework}.

The setup algorithm of MDDW is the same is the setup algorithm of MDVS. The public/secret watermarking (resp., detection)  key pairs are generated by invoking the signing (resp., verification) key generation algorithm of the MDVS.

Given the  watermarking secret  key $\wsk$ (which belongs to a model provider), a designated detector set $S$ (whose public keys are $\{\dpk_j\}_{j\in S}$), and a prompt $\bm{p}$,  the  watermarking algorithm $\watm$, run by the model provider,   proceeds as follows.

\begin{figure}[h]
	\centering
\includegraphics[width=0.78\textwidth]{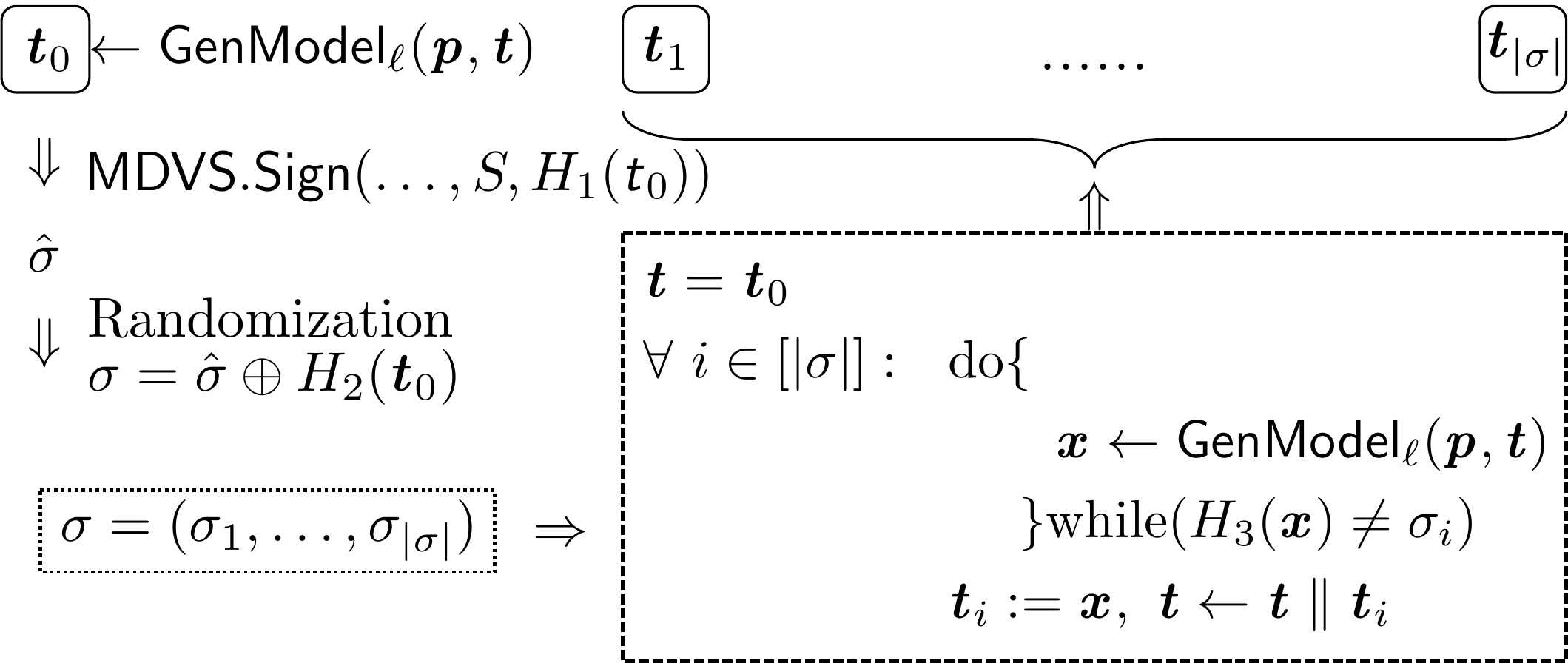}	
 \captionsetup{justification=centering}
	\caption{The MDDW framework based on MDVS}
	\label{fig:MDDW_framework}
\end{figure}

\begin{enumerate}
\item[(i)] \emph{Prompt Processing.} The model provider begins by using $\textsf{GenModel}_{\ell}$ (given a prompt $\bm{p}$ and the previously generated  tokens $\bm{t}$) to  generate  text in blocks of $\ell$ tokens.

\item[(ii)] \emph{Generating MDVS signatures.}  
Take the last $\ell$ tokens of the  current token sequence $\bm{t}$ as message $m$. The model provider generates a signature $\hat{\sigma}$ on $m$ for the designated detectors in $S$, with the signing algorithm of the underlying MDVS. It is crucial to note that the model provider's secret key $\wsk$ is the signing secret key, and $\{\dpk_j\}_{j\in S}$ are actually  the corresponding verification public keys of the designated detectors in $S$.

\item[(iii)] \emph{Randomizing MDVS signatures\footnote{Very recently, \cite{FGJ+23} also independently proposes  a similar randomization step in their watermarking framework in their latest version, in order to eliminate the need for pseudorandom signatures. However, as discussed in the aforementioned section, it still fails to achieve meaningful  distortion-freeness in the public-key setting.}.} This process involves performing the XOR operation on the MDVS signature $\hat{\sigma}$ and the hash value of the last $\ell$ tokens of the current token  sequence $\bm{t}$. For simplicity, we denote the XOR result as $\sigma$.

\item[(iv)] \emph{Sampling tokens with rejection.} Let $\sigma_i$ represent the $i^{th}$ bit of $\sigma$, for $i\in[|\sigma|]$. Firstly, generate a candidate $\ell$-bit token sequence $\bm{x}$ using  $\textsf{GenModel}_{\ell}$ with given a prompt $\bm{p}$ and the current  token sequence $\bm{t}$. Compute the hash value of   $\bm{x}$. If the hash value matches $\sigma_1$, then  accept $\bm{x}$ and append it to the end of $\bm{t}$.  Otherwise, reject $\bm{x}$ and generate a new candidate sequence. For $\sigma_2$, $\sigma_3$, and continuing through $\sigma_{|\sigma|}$, repeat the above steps by calculating the hash value of the last $\ell$ bits of the current $\bm{t}$ and comparing it to $\sigma_i$ ($i\in\{2,\cdots,|\sigma|\}$), appending the sequence only if they match. 
\end{enumerate}
That is the main idea for generating  the MDDW watermarked texts. 

The detection process is similar to the watermarking process, so we will not delve into many details here. Given the  watermarking public  key $\wpk$ (which belongs to a model provider), 
a designated detector set $S$ (whose public keys are $\{\dpk_j\}_{j\in S}$), a secret detection key $\dsk_{j'}$ belonging to a detector $j'\in S$, and a token sequence $\bm{t}$,  the detection algorithm $\detect$, run by the detector $j'$, proceeds   as follows. First, divide $\bm{t}$   into contiguous blocks of $\ell$ tokens each.  
Then, using the predetermined method (which involves hashing and comparing operations) and the verification algorithm of the underlying MDVS, extract a watermark or signature   from the token block (if one is embedded) with the help of the verification secret key of the designated detector (i.e., $\dsk_{j'}$). 


\vspace{1mm}
\noindent \underline{\emph{Security analysis.}} Our generic MDDW construction is based on an MDVS scheme, ensuring that many of its security properties are directly derived from those of the MDVS. For example,  the consistency of our  MDDW   is inherited directly from the consistency of the underlying MDVS, the soundness of our MDDW is also guaranteed by the existential unforgeability of the underlying MDVS,  and so on. Moreover, following \cite{CGZ23,FGJ+23}, we also assume that any contiguous block of $\ell$ tokens contains at least $\alpha$ bits of min-entropy. With this assumption and the properties of random oracle, we can show that our MDDW achieves distortion-freeness.

Regarding  the optional security requirements, the off-the-record property for any subset of our MDDW is implied by the off-the-record property for any subset of the underlying MDVS, and  some MDVS schemes provide this security  \cite{damgaard2020stronger,maurer2022multi}. However, to the best of our knowledge, it seems that currently no known security property of MDVS can ensure the claimability of our MDDW.

To ensure claimability in our generic MDDW, we introduce the security notion of claimability for MDVS,  an extension of the claimability notions for  designated verifier signatures \cite{yamashita2023designated} and ring signatures \cite{PS19}. 
	
Similar to the claimability notion for MDDW, claimability for MDVS requires two algorithms: $\cla_{\text{mdvs}}$ and $\claver_{\text{mdvs}}$. A genuine signer uses $\cla_{\text{mdvs}}$ with their secret signing key to generate proof $\pi$ for a signature $\sigma$ they created. The public then uses $\claver_{\text{mdvs}}$ to verify this proof $\pi$. This security also requires that (i) proofs created with $\cla_{\text{mdvs}}$ must be correctly validated by $\claver_{\text{mdvs}}$,	(ii) no one can successfully claim a signature that they did not generate, and (iii) no one can falsely accuse others of being the signer.
	
By applying this claimability to the underlying MDVS, we can show that our generic MDDW achieves claimability.

\vspace{1mm}
\noindent \underline{\emph{Construction of  MDVS with claimability.}} We provide a generic method to transform any MDVS into an MDVS with claimability, abbreviated as CMDVS. The main idea, with some details omitted,   is as follows. To generate a CMDVS signature, first  produce an MDVS signature $\sigma_{\text{mdvs}}$.  Then, use a standard signature scheme to sign the signer's public key and   the  MDVS signature $\sigma_{\text{mdvs}}$, producing  a standard signature. Next, commit the signer's public key, the designated  verifiers' public keys and the standard signature with a commitment scheme  to produce  a commitment $\com$. The CMDVS signature comprises $(\sigma_{\text{mdvs}},\com)$. To make a claim, the signer reveals the opening of the commitment $\com$, including the standard signature and the randomness used for the commitment. The claim  verification process involves checking the correctness of the opening and the validity of the standard signature. 

Due to the correctness of the standard signature and the commitment scheme, the generated claim will be verified successfully. Regarding the second requirement of claimability, note that the CMDVS signature contains $\com$, produced by committing the signer's public key, the designated  verifiers' public keys and the standard signature. If another party tries to claim that they generated the CMDVS signature, they would need to provide an opening of $\com$ with their own public key, distinct from the actual signer's public key. This would imply two different openings for the same commitment, contradicting the binding property of the commitment scheme. As for the third requirement of claimability, it primarily relies on    the existential unforgeability of the standard signature. This is because the generated claim contains a standard signature, which is inherently difficult to forge without access to the signer's secret key.

\vspace{1mm}
\noindent \underline{\textbf{Discussions.}} Here, we provide some discussions, which contains some further contributions and future works. 

\vspace{1mm}
\noindent \underline{\emph{More efficient concrete DDW.}} Our generic framework of MDDW implies a generic designated detector watermarking (DDW), when there is only one designated detector (i.e., there is only one designated verifier in the corresponding MDVS scheme). However, when plugging  the concrete MDVS (e.g., \cite{au2014strong}) with one designated verifier into our framework, the resulting DDW is not very efficient. In Appendix \ref{sec:efficientddw}, we consider a more efficient concrete DDW. The concrete DDW is based on the designated verifier signature (DVS) \cite{LV04,steinfeld2003universal}. Compared with the DDW scheme from the MDVS \cite{au2014strong}, the DDW from DVS \cite{LV04,steinfeld2003universal} has shorter length of watermark (appropriately only $1/4$ of the length of the DDW from \cite{au2014strong}). Thus, given a fixed length of the text output by the LLM, we can embed more watermarks, which also implies that the DDW from DVS \cite{LV04,steinfeld2003universal} is more robust. The experimental results in Sec. \ref{sec:evaluation} also show that DDW in Appendix \ref{sec:efficientddw} is practical.

\vspace{1mm}
\noindent \underline{\emph{Bit Length of  watermarks for any subset.}} When  considering the off-the-record property for any subset, we have the following theorem for the bit length of   watermarks, with a proof  inspired by  \cite[Theorem 1]{damgaard2020stronger}.
\begin{theorem}
	If an MDDW scheme supports the off-the-record property for any subset and soundness, then the size of the generated watermarks must be $\Omega(n)$, where $n$ is the number of the designated detectors (i.e., $|S|=n$). 
\end{theorem}

\begin{proof}
	Suppose that there exists a distinguisher $D$ who knows all detectors' secret keys. Formally, the  off-the-record property for any subset requires the existence of a PPT algorithm $\forgeas$ (please refer to Def. \ref{def:mddw_otr_as}), that takes the public parameter $\pp$, the public key of the signer $\wpk_i$, the public keys of the designated detectors $(\dpk_j)_{j \in S}$, the secret keys of the designated detectors in the corrupted set $(\dsk_j)_{j \in \scor}$ and a prompt $\bm{p}$ as input, and outputs a text $\bm{t}$, which is indistinguishable from a real one output by $\watm$.
	
	Given a text generated by $\forgeas$ with a corrupted detector set $\scor \subseteq S$, $D$ can distinguish whether some designated detector is in $\scor$ or not, by verifying the validity of the watermark (i.e., the MDVS signature). Note that the  off-the-record property for any subset  and soundness guarantee that only the detectors in $\scor$ would accept the watermark. Thus, $D$ can determine the set $\scor$ when given a text generated by $\forgeas$. Therefore, the watermark must contain enough bits to indicate $\scor$. Since $\scor \subseteq S$ and it can be an arbitrary subset, there are $2^{|S|} = 2^n$ kinds of different subsets. Thus, the bit length of watermark is at least $\log_2 2^{|S|} = \log_2 2^n = n$. Considering that the text output by $\watm$ is indistinguishable from that output by $\forgeas$, the bit length of watermark embedded in the text output by $\watm$ is also $\Omega(n)$.
	\qed
\end{proof}

\noindent\underline{\emph{Strengthening the security models.}} It is important to note that in this paper,   claimability  is defined as \emph{optional} for MDDW. If it is considered a mandatory security requirement, then an MDDW scheme must include the algorithms $\cla$ and $\claver$. In this case, several security models (e.g., the game defining soundness) can be strengthened by granting the adversary access to the claiming oracle.  Similar enhancements can be applied to the security models of MDVS. We will not elaborate on the details and leave this for future work.

\vspace{1mm}
\noindent \underline{\emph{Error-correcting code.}} 
We note that several works \cite{FGJ+23,qu2024provably,christ2024pseudorandom} have considered the use of error-correction codes to enhance watermarking schemes. It improves robustness and relaxes the entropy assumption by allowing the watermarking scheme to tolerate errors caused by adversarial modifications or low entropy token distribution in rejection sampling. We emphasize that error-correction codes can also be applied to our work to achieve the same effects. Since error-correction codes are essentially specialized coding schemes that alter the form of signatures, they would not affect our MDDW framework.

\vspace{1mm}
\noindent \underline{\emph{Empirical attacks.}} We stress that, similar to existing works like  \cite{KGW+23,CGZ23,FGJ+23}, our scheme is also vulnerable to certain empirical attacks, like the emoji attack and translation attack.   For a more detailed introduction to these attacks, please refer to \cite{KGW+23,CGZ23}.  Actually,  \cite{CGZ23} concludes that no undetectable watermarking schemes can be completely unremovable. 


\vspace{1mm}

\noindent \underline{\textbf{Other watermarking schemes.}} 
Recently, there is a line of works using   cryptographic tools to construct watermarking schemes for LLMs, besides the aforementioned works \cite{KGW+23,CGZ23,FGJ+23}. Here, we summarized some works as follows.


Piet et al. \cite{DBLP:journals/corr/abs-2312-00273} systematically analyze watermarking schemes in the secret key setting. Their study focuses on assessing generation quality and robustness to minor edits for practical protocol parameters. They conclude that distortion-freeness is too strong a property for practice. However, as pointed  out in \cite{FGJ+23}, this conclusion was drawn from quality assessment performed by the chat version of Llama 2 7B \cite{touvron2023llama}, and they remark that Llama 2 7B's quality assessment likely does not generalize, since higher fidelity models may reveal weaknesses in distortion-inducing watermarking schemes. 

Zhang et al. \cite{zhang2023watermarks} formally prove that ``strong'' robustness is impossible in watermarking schemes. They further demonstrate that their attack works in practice against a range of secret-key watermarking schemes (including the \cite{KGW+23} scheme). On the other hand, \cite{zhang2023watermarks} also points out that ``weak watermarking schemes'' do exist and can be useful. To be noticed, our security notion of robustness is a kind of weak notions.

Qu et al. \cite{qu2024provably} present a watermarking scheme for LLMs that uses error correction codes to gain robustness. Similarly,    Christ  and Gunn  \cite{christ2024pseudorandom} construct  pseudorandom error-correction codes and applies it to show an undetectable watermarking scheme. 

\vspace{1mm}
\noindent\underline{\textbf{Roadmap.}} We recall some preliminaries in Sec. \ref{sec:preliminary}. Then, we introduce the primitive of MDDW and formalize its security notions in Sec. \ref{sec:MDDW_def}. In Sec. \ref{sec:mddw_const}, we provide a framework of constructing MDDW for any LLM from  MDVS, and show that it achieves all   required security properties.  Finally, we present the experimental results in Sec. \ref{sec:evaluation}.

\section{Preliminaries}\label{sec:preliminary}
\underline{\textbf{Notations.}} We assume that the security parameter  $\lambda$ is an (implicit) input to all algorithms. For any $n\in\mathbb{N}$, let $[n]:=\{1,2,\cdots,n\}$. For a finite set $S$, let $|S|$ denote the number of elements in $S$, and  $x\leftarrow S$ denote the process of sampling $x$ from $S$ uniformly at random. For a distribution $\textsf{Dist}$, we use $x\leftarrow\textsf{Dist}$ to denote the process of sampling $x$ from $\textsf{Dist}$. A function $f$ of $\lambda$ is \emph{negligible} if $f\in O(\frac{1}{\textsf{poly}(\lambda)})$ for any polynomial $\textsf{poly}(\cdot)$. For simplicity, we write $f(\lambda)\leq\negl$ to mean that $f$ is negligible. Throughout this paper, let $\epsilon$ denote the empty list or empty string.   

For a random variable $X$, the \emph{min-entropy} $\mathbf{H}_{\infty}(X)$ is $-\log(\max_x\Pr[X=x])$.  

Let $a||b$ denote the tail-to-head concatenation of $a$ and $b$. In a bit string $a$, $a_i$ represents the $i^{th}$ bit of $a$, unless indicated otherwise. 

Throughout this paper, bold lower-case letters denote vectors or sequence of tokens, e.g., $\bm{t}=(t_1,\ldots,t_{|\bm{t}|})$ is a $|\bm{t}|$-dimension vector, and usually the number of dimensions can be inferred from the context.
For a vector or sequence of tokens  $\bm{t}$, let $\bm{t}[i]$ denote the $i^{th}$ element of $(t_1,\cdots,t_{|\bm{t}|})$, and  $\bm{t}[-i]$ denote the  $i^{th}$  last element  of $(t_1,\cdots,t_{|\bm{t}|})$. For $i\leq j$,  $\bm{t}[i:j]$ denotes $(t_i,t_{i+1},\cdots,t_{j})$, and  $\bm{t}[-i:]$ denotes the last $i$ elements  of $(t_1,\cdots,t_{|\bm{t}|})$. We sometimes adapt these notations to bit strings, such as $a[1]$ and $a[-i]$.




\subsection{Language models}\label{sec:language_model}
We follow \cite{KGW+23,CGZ23,FGJ+23} in our definition of a language model, and refer to language models simply as \emph{models} in this paper. 

\begin{definition}[Auto-regressive model]  An \emph{auto-regressive model} $\model$ over token vocabulary $\tg$ is a 
deterministic algorithm that takes as input a prompt $\bm{p}\in\tg$ and tokens previously output by the model $\bm{t}\in\tg$
and outputs a probability distribution $p=\model(\bm{p},\bm{t})$ over $\tg$.
\end{definition}

For all $\ell\in\mathbb{N}$, $\genmodel_{\ell}$ wraps around $\model$ to implement a generative model. $\genmodel_{\ell}$  iteratively generates $\ell$  tokens.
$\decode$ is the specific decoding method, e.g., $\mathsf{argmax}$.

\vspace{-3mm}
\begin{algorithm}[h]
\caption{Generative  model  $\genmodel_{\ell}(\bm{p}\in\tg, \bm{t}\in\tg)$ }
\begin{algorithmic}[1]
\For {$i=1,\cdots,\ell$} 
\State $\bm{t}\leftarrow\bm{t}||\decode(\model(\bm{p},\bm{t}))$
\EndFor
\State \Return $\bm{t}$
\end{algorithmic}
\end{algorithm}

Following \cite{FGJ+23}, we assume that any contiguous block of $\ell\in\mathbb{N}$ tokens contains at least $\alpha$ bits of min-entropy, i.e., no particular
sample is more than $2^{-\alpha}$ likely to happen. 
\begin{assumption}\label{assumption} For any prompt $\bm{p}$ and tokens $\bm{t}$, the new tokens $\bm{t}'\leftarrow\genmodel_\ell(\bm{p}, \bm{t})\in\tg^{\ell}$ were sampled
from distributions with min-entropy at least $\alpha$.
\end{assumption}

\noindent\underline{\textbf{Entities.}}  
In the scenarios discussed in this paper, three types of entities are involved.
\begin{itemize}
	\item\emph{Model provider.} The model provider delivers the large language model (LLM) service, generating LLM outputs for specified configurations based on received prompts.  An honest model provider will execute the watermarking scheme using their watermarking secret key  during the text generation phase.
	\item\emph{Designated detector.} The designated detector, using their own detection secret key, checks if a text is watermarked with respect to a specific model provider, without access to the model weights or the model provider's secret watermarking key. 
	\item\emph{User.} Users create prompts that are sent to the LLM to receive the model output in return.
\end{itemize}

\subsection{Multi-designated verifier signature}\label{sec:MDVS_definition}
We recall the definition of multi-designated verifier signature (MDVS) from  \cite{damgaard2020stronger,maurer2022multi}, with some adjustments. 


An MDVS scheme, associated with message space $\M$, is a tuple of algorithms $\mdvs=(\setup,\skg,\vkg,\sign, \verify)$, where 
\begin{myitemize}
	\item
	$\setup(1^\lambda) \rightarrow \pp$: On input the security parameter $\lambda$, the setup algorithm  outputs a public parameter  $\pp$.
	\item
	$\skg(\pp) \rightarrow (\spk,\ssk)$: On input  $\pp$, the signing key generation algorithm  outputs a public key $\spk$ and a secret key $\ssk$ for a signer. 
	\item
	$\vkg(\pp) \rightarrow (\vpk,\vsk)$: On input  $\pp$, the verification key generation algorithm outputs a public key $\vpk$ and a secret key $\vsk$ for a verifier.
	\item
	$\sign(\pp,\ssk_i,\{\vpk_j\}_{j \in S},m) \rightarrow \sigma$: On input  $\pp$, a  signing secret key $\ssk_i$,  public keys of the designated verifiers $\{\vpk_j\}_{j\in S}$, and a message $m$, the signing algorithm outputs a signature $\sigma$.
	\item
	$\verify(\pp,\spk_i,\vsk_{j'},\{\vpk_j\}_{j \in S},m,\sigma) \rightarrow b$: On input  $\pp$, a signing public   key $\spk_i$, a secret key $\vsk_{j'}$ of a   verifier such that $j' \in S$,  public keys of the designated verifiers $\{\vpk_j\}_{j\in S}$, a message $m$, and a signature $\sigma$, the verification algorithm outputs a bit $b\in\{0,1\}$.
\end{myitemize}

\begin{figure}[t!]\small
	\centering
	\fbox{\parbox{1.05\linewidth}{\shortstack[l]{
				\underline{$\gameconsmdvs$:}\\
				$\pp \leftarrow \setup(1^{\lambda})$\\
				$(i^*,S^*,m^*,\sigma^*) \leftarrow \A^{\orask,\oravk,\oraspk,\oravpk,\orasign,\oraver}(\pp)$\\
				If $\exists~j^*_0,j^*_1 \in S^*$, such that\\
				$\qquad$$\verify(\pp,\spk_{i^*},\vsk_{j^*_0},\{\vpk_j\}_{j\in S^*},m^*,\sigma^*) \neq \verify(\pp,\spk_{i^*},\vsk_{j^*_1},\{\vpk_j\}_{j\in S^*},m^*,\sigma^*)$,\\ $\qquad$where all keys are the honestly generated outputs of  the   key generation  oracles, and \\
				$\qquad$$\oravk$ is never queried on  $j^*_0$ or $j^*_1$\\ 
				$\quad$then return $1$\\
				Return $0$
			}
	}}
	\fbox{\parbox{1.05\linewidth}{\shortstack[l]{
				\underline{$\gameunfmdvs$:}\\
				$\pp \leftarrow \setup(1^{\lambda})$\\
				$(i^*,S^*,m^*,\sigma^*) \leftarrow \A^{\orask,\oravk,\oraspk,\oravpk,\orasign,\oraver}(\pp)$\\
				$\qquad$where $\orask$ is never queried on $i^*$, and $\orasign$ is never queried on $(i^*,S^*,m^*)$\\
				If $\exists~j^* \in S^*$, such that\\
				$\qquad$$\verify(\pp,\spk_{i^*},\vsk_{j^*},\{\vpk_j\}_{j\in S^*},m^*,\sigma^*)=1$, where all keys are the honestly\\
			$\qquad$generated  outputs  of the   key generation oracles, and  $\oravk$ is never queried on $j^*$ \\
				$\quad$then return $1$\\
				Return $0$
			}
	}}
	\captionsetup{justification=centering}
	\caption{\small Games $\gameconsmdvs$ and   $\gameunfmdvs$  for MDVS, and the oracles are given in Fig.  \ref{fig:MDVS_oracle}}
	\label{fig:mdvs_cons}
\end{figure}

\begin{figure}[t!]\scriptsize
	\centering
	\fbox{\parbox{\linewidth}{
			\textbf{Signer Key Generation Oracle} $\orask(i)$:
			\begin{enumerate}[itemsep=2pt,topsep=0pt,parsep=0pt]
				\item On the first call to $\orask$ on  $i$, compute $(\spk_i,\ssk_i)\leftarrow\skg(\pp)$, output and store $(\spk_i,\ssk_i)$.
				\item On subsequent calls, simply output $(\spk_i,\ssk_i)$. 
			\end{enumerate}
			\vspace{0.4\baselineskip}
			\textbf{Verifier Key Generation Oracle} $\oravk(j)$: 
			\begin{enumerate}[itemsep=2pt,topsep=0pt,parsep=0pt]
				\item On the first call to $\oravk$ on   $j$, compute $(\vpk_j,\vsk_j)\leftarrow\vkg(\pp)$, output and store $(\vpk_j,\vsk_j)$.
				\item On subsequent calls, simply output $(\vpk_j,\vsk_j)$. 
			\end{enumerate}
			\vspace{0.4\baselineskip}
			\textbf{Public Signer Key Generation Oracle}  $\oraspk(i)$: 
			\begin{enumerate}[itemsep=2pt,topsep=0pt,parsep=0pt]
				\item $(\spk_i,\ssk_i)\leftarrow\orask(i)$; output $\spk_i$.
			\end{enumerate}
			\vspace{0.4\baselineskip}
			\textbf{Public Verifier Key Generation Oracle}  $\oravpk(j)$: 
			\begin{enumerate}[itemsep=2pt,topsep=0pt,parsep=0pt]
				\item $(\vpk_j,\vsk_j)\leftarrow\oravk(j)$; output $\vpk_j$.
			\end{enumerate}
			\vspace{0.4\baselineskip}
			\textbf{Signing Oracle} $\orasign(i,S,m)$:
			\begin{enumerate}[itemsep=2pt,topsep=0pt,parsep=0pt]
				\item $(\spk_i,\ssk_i)\leftarrow\orask(i)$, $\{\vpk_j\leftarrow \oravpk(j)\}_{j\in S}$. 
				\item Output $\sigma\leftarrow \sign(\pp,\ssk_i,\{\vpk_j\}_{j \in S},m)$.
			\end{enumerate}
			\vspace{0.4\baselineskip}
			\textbf{Verification Oracle}  $\oraver(i,j'\in S,S,m,\sigma)$: 
			\begin{enumerate}[itemsep=2pt,topsep=0pt,parsep=0pt]
				\item $\spk_i\leftarrow\oraspk(i)$, $(\vpk_{j'},\vsk_{j'})\leftarrow\oravk(j')$
				\item $\{\vpk_j\leftarrow \oravpk(j)\}_{j\in S}$. 
				\item Output $b\leftarrow\verify(\pp,\spk_i,\vsk_{j'},\{\vpk_j\}_{j \in S},m,\sigma)$.
			\end{enumerate}
	}		}
	\captionsetup{justification=centering}
	\caption{\small The oracles for the games defining security notions for MDVS
	}
	\label{fig:MDVS_oracle}
\end{figure}  

\begin{figure}[h!]\small
	\centering
		\fbox{\parbox{1.05\linewidth}{\shortstack[l]{
				\underline{$\gameotrdsmdvs$:}\\
				$\pp \leftarrow \setup(1^{\lambda})$, $b\leftarrow \{0,1\}$\\
				$b' \leftarrow \A^{\orask,\oravk,\oraspk,\oravpk,\mathcal{O}_{otr\text{-}chl}^{(b)},\oraver}(\pp)$\\
				$\qquad$where  $\mathcal{O}_{\text{\emph{otr-chl}}}^{(0)}(i,S,m)$ outputs $\orasign(i,S,m)$, $\mathcal{O}_{\text{\emph{otr-chl}}}^{(1)}(i,S,m)$ outputs $\sigma\leftarrow\forgeds(\pp,$ \\  
				$\qquad$$\spk_i,\{\vsk_j\}_{j\in S},m)$ (where all keys are the honestly  generated outputs of  the 
				key \\  
				$\qquad$generation  oracles),\\
				$\qquad$let $Q$ denote the set of query-response 
				of  $\mathcal{O}_{\text{\emph{otr-chl}}}^{(b)}$, \\
				$\qquad$all queries $(i,S,m)$ to $\mathcal{O}_{\text{\emph{otr-chl}}}^{(b)}$ should satisfy that  $\orask$ is never queried on $i$,  \\ 
				$\qquad$and  all queries $(i',j',S',m',\sigma')$ to $\oraver$ should satisfy   ``$\not\exists~(*,*,*,\sigma)\in Q$ s.t. $\sigma=\sigma'$'' \\
				If $b'=b$, then return $1$\\
				Return $0$
			}		
	}}
	\fbox{\parbox{1.05\linewidth}{\shortstack[l]{
				\underline{$\gameotrasmdvs$:}\\
				$\pp \leftarrow \setup(1^{\lambda})$, $b\leftarrow \{0,1\}$\\
				$b' \leftarrow \A^{\orask,\oravk,\oraspk,\oravpk,\mathcal{O}_{otr\text{-}chl}^{(b)},\oraver}(\pp)$\\
				$\qquad$where  $\mathcal{O}_{\text{\emph{otr-chl}}}^{(0)}(i,S,\scor,m)$ outputs $\orasign(i,S,m)$, $\mathcal{O}_{\text{\emph{otr-chl}}}^{(1)}(i,S,\scor,m)$ outputs \\  $\qquad$$\sigma\leftarrow\forgeas(\pp,\spk_i,\{\vpk_j\}_{j\in S},\{\vsk_j\}_{j\in \scor},m)$ (where all keys are the honestly  \\  $\qquad$$\qquad$generated outputs of  the 
				 key generation oracles),\\
				$\qquad$let $Q$ denote the set of query-response 
				of  $\mathcal{O}_{\text{\emph{otr-chl}}}^{(b)}$, \\
				$\qquad$all queries $(i,S,\scor,m)$ to $\mathcal{O}_{\text{\emph{otr-chl}}}^{(b)}$ should satisfy that (1) $\scor\subset S$, (2) $\orask$ is never \\
				$\qquad$$\qquad$queried on $i$, and (3) for all $j\in S\setminus\scor$, $\oravk$ is never queried on $j$,  \\ 
				$\qquad$all queries $j$ to   $\oravk$ should satisfy  ``$\not\exists~(*,S,\scor,*,*)\in Q$ s.t. $j\in S\setminus\scor$'',  \\ 
				$\qquad$and  all queries $(i',j',S',m',\sigma')$ to $\oraver$ should satisfy   ``$\not\exists~(*,*,*,*,\sigma)\in Q$ s.t. $\sigma=\sigma'$'' \\
				If $b'=b$, then return $1$\\
				Return $0$
			}		
	}}
	\captionsetup{justification=centering}
	\caption{\small Games  $\gameotrdsmdvs$ and  $\gameotrasmdvs$   for MDVS, and the oracles are given in Fig.   \ref{fig:MDVS_oracle}}
	\label{fig:mdvs_otr}
\end{figure}

In this paper, we require that an MDVS scheme should satisfy  correctness, consistency,  existential unforgeability, and off-the-record property for designated set. 

\begin{definition}[Correctness] 
	We say that $\mdvs$ is \emph{correct}, if for any signer $i$, any message $m \in \M$, any verifier identity set $S$ and any $j' \in S$, it holds that
	\begin{equation*}\small
		\Pr \left[
		\begin{aligned}
			&\pp\leftarrow \setup(1^{\lambda})\\
			&(\spk_i,\ssk_i) \leftarrow \skg(\pp)\\
			&\{(\vpk_j,\vsk_j) \leftarrow \vkg(\pp)\}_{j\in S}\\
			&\sigma \leftarrow \sign(\pp,\ssk_i,\{\vpk_j\}_{j \in S}, m)
		\end{aligned}
		~\textup{\Large{:}}~
		\begin{aligned}
			&\verify(\pp,\spk_i,\vsk_{j'},\\
			&\quad\quad\{\vpk_j\}_{j \in S},m,\sigma) \neq 1
		\end{aligned}
		\right] = 0.
	\end{equation*}
	\label{def:mdvs_correct}
\end{definition}

\begin{definition}[Consistency] 
	We say that  $\mdvs$ is \emph{consistent}, if for any PPT adversary  $\A$,  \[\advconsmdvs=\Pr[\gameconsmdvs=1] \leq \negl \] where $\gameconsmdvs$ is shown in Fig. \ref{fig:mdvs_cons}.
	\label{def:mdvs_cons}
\end{definition}

\begin{definition}[Unforgeability] 
	We say that $\mdvs$ is \emph{unforgeable},  if for any PPT adversary $\A$, 
	\[\advunfmdvs=\Pr[\gameunfmdvs=1] \leq \negl \] 
	where $\gameunfmdvs$ is shown in Fig. \ref{fig:mdvs_cons}.
	\label{def:mdvs_unf}
\end{definition}

\begin{definition}[Off-the-record for designated set]
	We say that $\mdvs$ is \emph{off-the-record for designated  set}, if there is a PPT algorithm $\forgeds$, taking $(\pp,\spk_i,\{\vsk_j\}_{j\in S},m)$   as input and outputting $\sigma$, such that  for any PPT adversary $\A$,  \[\advotrdsmdvs=|\Pr[\gameotrdsmdvs=1]-\frac{1}{2}| \leq \negl \] where $\gameotrdsmdvs$ is shown in Fig. \ref{fig:mdvs_otr}.
	\label{def:mdvs_otr_ds}
\end{definition}

Now, we recall the off-the-record property  for any subset  for MDVS \cite{damgaard2020stronger}. Note that unlike \cite{damgaard2020stronger}, in this paper, the off-the-record property is defined as \emph{optional} for MDVS. 
Note that off-the-record property for any subset implies off-the-record property for designated set. 

\begin{definition}[Off-the-record for any subset] 
	 	We say that $\mdvs$ is \emph{off-the-record for any subset}, if there is a PPT algorithm $\forgeas$, taking $(\pp,\spk_i,$ $\{\vpk_j\}_{j\in S},\{\vsk_j\}_{j\in \scor},m)$ (where $\scor\subset S$) as input and outputting $\sigma$, such that  for any PPT adversary $\A$,  \[\advotrasmdvs=|\Pr[\gameotrasmdvs=1]-\frac{1}{2}| \leq \negl \] where $\gameotrasmdvs$ is shown in Fig. \ref{fig:mdvs_otr}.
	 	\label{def:mdvs_otr_as}
	 \end{definition}

\section{Multi-designated detector watermarking}\label{sec:MDDW_def}

In this section, we introduce a primitive called \emph{multi-designated detector watermarking (MDDW)}, and formalize its security notions. 

\vspace{1mm}
An MDDW scheme for an auto-regressive model $\model$ over token vocabulary $\tg$ is a tuple of algorithms $\mddw=(\setup,\wkg,\dkg,\watm,\detect)$, where 
\begin{myitemize}
    \item
    $\setup(1^\lambda) \rightarrow\pp$: On input the security parameter $\lambda$, the setup algorithm outputs a  public parameter $\pp$.
	\item
	$\wkg(\pp) \rightarrow (\wpk,\wsk)$: On input  $\pp$, the watermarking key generation algorithm  outputs a public/secret key pair $(\wpk,\wsk)$ for watermarking. 
	\item
	$\dkg(\pp) \rightarrow (\dpk,\dsk)$: On input $\pp$, the detection key generation algorithm outputs a public/secret key pair $(\dpk,\dsk)$ for a detector.
	\item $\watm(\pp,\wsk_i,\{\dpk_j\}_{j \in S},\bm{p}) \rightarrow \bm{t}$: On input $\pp$, a  watermarking secret key $\wsk_i$,  public keys of the designated detectors $\{\dpk_j\}_{j\in S}$, and a prompt $\bm{p}\in\tg$, the watermarking  algorithm outputs   $\bm{t}\in\tg$.
	\item  $\detect(\pp,\wpk_i,\dsk_{j'},\{\dpk_j\}_{j \in S},\bm{t}) \rightarrow b$: On input  $\pp$, a  watermarking public key $\wpk_i$, a secret key $\dsk_{j'}$ of a detector such that $j' \in S$,   public keys of the designated detectors $\{\dpk_j\}_{j\in S}$,  and a candidate watermarked text  $\bm{t}\in\tg$, the detection algorithm outputs a bit $b\in\{0,1\}$.
\end{myitemize}
\vspace{1mm}

We provide further explanations here. The watermarking key generation algorithm $\wkg$ (run by the model providers) and the detection key generation algorithm $\dkg$ (run by the detectors) generate their respective public/secret key pairs. To generate watermarked texts that can only be detected by designated detectors in a set $S$, the model provider (holding watermarking secret key  $\wsk_i$) runs $\watm(\pp,\wsk_i,\{\dpk_j\}_{j \in S},\bm{p})$ to produce   $\bm{t}$. A detector $j'\in S$ can then execute  $\detect(\pp,\wpk_i,\dsk_{j'},\{\dpk_j\}_{j \in S},\bm{t})$ using their own detection secret key $\dsk_{j'}$ to check whether $\bm{t}$ is watermarked .

Regarding security, we require that an MDDW scheme should satisfy  \emph{completeness}, \emph{consistency}, \emph{soundness},  \emph{distortion-freeness},  \emph{robustness}, and \emph{off-the-record property for designated set}. The formal definitions of these security requirements are  as follows.

\begin{figure}[!t]\small
	\centering
	\fbox{\parbox{1.06\linewidth}{\shortstack[l]{
				\underline{$\gamecons$:}\\
				$\pp \leftarrow \setup(1^{\lambda})$\\ $(i^*,S^*,\bm{t}^*) \leftarrow \A^{\orawk,\oradk,\orawpk,\oradpk,\orawatm,\oradet}(\pp)$\\
				If $\exists~j_0,j_1 \in S^*$, such that $\detect(\pp,\wpk_{i^*},\dsk_{j_\beta},\{\dpk_j\}_{j \in S^*},\bm{t}^*)=\beta$ for all  $\beta\in\{0,1\}$, \\
				$\qquad$where all keys are the honestly generated outputs of  the   key generation  oracles, and \\
				$\qquad$$\oradk$ is never queried on  $j_0$ or $j_1$\\
				$\quad$then return $1$\\
				Return $0$
			}
	}}
	\fbox{\parbox{1.06\linewidth}{\shortstack[l]{
				\underline{$\gameunf$:}\\
				$\pp \leftarrow \setup(1^{\lambda})$\\
				$(i^*,S^*,\bm{t}^*) \leftarrow \A^{\orawk,\oradk,\orawpk,\oradpk,\orawatm,\oradet}(\pp)$\\
				$\qquad$where $\orawk$ is never queried on $i^*$, and all $\bm{t}_1,\bm{t}_2,\cdots$ (denoting the watermarked texts \\
				$\qquad$that $\A$ receives when querying $\orawatm$ on $(i^*,S^*, *)$)   satisfy $\textsf{NOLap}_k(\bm{t}^*,\bm{t}_1,\bm{t}_2,\cdots)=1$\\
				If $\exists~j' \in S^*$, such that $\detect(\pp,\wpk_{i^*},\dsk_{j'},\{\dpk_j\}_{j\in S^*},\bm{t}^*)=1$\\
				$\qquad$where all keys are the honestly generated outputs of  the key generation  oracles, and \\
				$\qquad$$\oradk$ is never queried on $j'$,\\
				$\quad$then return $1$\\
				Return $0$
			}
	}}
	\fbox{\parbox{1.06\linewidth}{\shortstack[l]{
				\underline{$\gamedistor$:}\\
				$\pp \leftarrow \setup(1^{\lambda})$, $b\leftarrow \{0,1\}$\\
				$b' \leftarrow \A^{\orawk,\oradk,\orawpk,\oradpk,\mathcal{O}_{\textsf{Model}},\mathcal{O}_{\text{\emph{M-chl}}}^{(b)},\oradet}(\pp)$\\
				$\qquad$where $\mathcal{O}_{\textsf{Model}}(\bm{p})$ outputs $\textsf{Model}(\bm{p})$, $\mathcal{O}_{\text{\emph{M-chl}}}^{(0)}(i,S,\bm{p})$ outputs $\orawatm(i,S,\bm{p})$, $\mathcal{O}_{\text{\emph{M-chl}}}^{(1)}(i,S,\bm{p})$ \\
				$\qquad$outputs  $\genmodel(\bm{p})$, let $Q$ denote the set of query-response of  $\mathcal{O}_{\text{\emph{M-chl}}}^{(b)}$,\\
				$\qquad$all queries $(i,S,\bm{p})$ to   $\mathcal{O}_{\text{\emph{M-chl}}}^{(b)}$   satisfy that ``$i$ has never been queried to $\orask$, and \\
				$\qquad$$\qquad$$\forall~j\in S$, $j$ has never   been queried to $\oradk$'',  \\ 
				$\qquad$all queries $j$ to   $\oradk$   satisfy that   ``$\not\exists~(i,S,*,\bm{t})\in Q$ s.t. $j\in S$'',  \\ 
				$\qquad$and  all queries $(i',j',\bm{t}')$ to $\oradet$   satisfy    ``$\forall~(i',S,*,\bm{t})\in Q$ s.t. $j'\in S$, $\textsf{NOLap}_k(\bm{t}',\bm{t})=1$''\\
				If $b'=b$, then return $1$\\
				Return $0$
	}		}}
	\captionsetup{justification=centering}
	\caption{\small Games $\gamecons$, $\gameunf$ and   $\gamedistor$   for  MDDW, and the oracles are given in Fig. \ref{fig:MDDW_oracle}}
	\label{fig:mddw_cons}
\end{figure}

\begin{definition}[Completeness] 
	We say that $\mddw$ is \emph{$\delta$-complete}, if for any  $i$, any prompt $\bm{p}\in\tg$,  any detector identity set $S$ and any $j'\in S$,  it holds that
	\begin{equation*}\small
            \Pr \left [
           \begin{aligned}
            &\pp\leftarrow \setup(1^{\lambda}),\\
            &(\wpk_i,\wsk_i)\leftarrow\wkg(\pp),\\
            &\{(\dpk_j,\dsk_j)\leftarrow \dkg(\pp)\}_{j \in S},\\
            &\bm{t} \leftarrow \watm(\pp,\wsk_i,\{\dpk_j\}_{j \in S}, \bm{p})\\
            &\quad\text{s.t.}~|\bm{t}|\geq \delta
           \end{aligned}
            \textup{\Large{:}}
           \begin{aligned}
           & \detect(\pp,\wpk_i,\dsk_{j'},\\
            &\quad\quad\{\dpk_j\}_{j \in S},\bm{t})\neq 1
        \end{aligned}
            \right ]\leq \negl.
        \end{equation*}
\end{definition}

\begin{definition}[Consistency] 
	We say that $\mddw$ is  \emph{consistent}, if for any PPT adversary $\A$, 
 \[\advcons=\Pr[\gamecons=1] \leq \negl \] 
 where $\gamecons$ is shown in Fig. \ref{fig:mddw_cons}.
	\label{def:mddw_cons}
\end{definition}

\begin{figure}[!t]
\small
	\centering
	\fbox{\parbox{\linewidth}{
			\textbf{Watermarking Key Generation Oracle:} $\orawk(i)$
			\begin{enumerate}[itemsep=2pt,topsep=0pt,parsep=0pt]
				\item 
				On the first call to $\orawk$ on $i$, compute $(\wpk_i,\wsk_i)\leftarrow\wkg(\pp)$, output and store $(\wpk_i,\wsk_i)$.
				\item 
				On subsequent calls, simply output $(\wpk_i,\wsk_i)$. 
			\end{enumerate}
			\vspace{0.4\baselineskip}
			\textbf{Detecting Key Generation Oracle:} $\oradk(j)$
			\begin{enumerate}[itemsep=2pt,topsep=0pt,parsep=0pt]
				\item On the first call to $\oradk$ on  $j$, compute $(\dpk_j,\dsk_j)\leftarrow\dkg(\pp)$, output and store $(\dpk_j,\dsk_j)$.
				\item On subsequent calls, simply output $(\dpk_j,\dsk_j)$. 
			\end{enumerate}
			\vspace{0.4\baselineskip}
			\textbf{Public Watermarking Key Generation Oracle:} $\orawpk(i)$
			\begin{enumerate}[itemsep=2pt,topsep=0pt,parsep=0pt]
				\item $(\wpk_i,\wsk_i)\leftarrow\orawk(i)$; output $\wpk_i$.
			\end{enumerate}
			\vspace{0.4\baselineskip}
			\textbf{Public Detecting Key Generation Oracle:} $\oradpk(j)$
			\begin{enumerate}[itemsep=2pt,topsep=0pt,parsep=0pt]
				\item $(\dpk_j,\dsk_j)\leftarrow\oradk(j)$; output $\dpk_j$.
			\end{enumerate}
			\vspace{0.4\baselineskip}
			\textbf{Watermarking Oracle:} $\orawatm(i,S,\bm{p})$
			\begin{enumerate}[itemsep=2pt,topsep=0pt,parsep=0pt]
				\item $(\wpk_i,\wsk_i)\leftarrow\orawk(i)$, $\{\dpk_j\leftarrow \oradpk(j)\}_{j\in S}$. 
				\item Output $\bm{t}\leftarrow \watm(\pp,\wsk_i,\{\dpk_j\}_{j \in S},\bm{p})$.
			\end{enumerate}
			\vspace{0.4\baselineskip}
			\textbf{Detecting Oracle:} $\oradet(i,j'\in S,S,\bm{t})$
			\begin{enumerate}[itemsep=2pt,topsep=0pt,parsep=0pt]
				\item $\wpk_i\leftarrow\orawpk(i)$, $(\dpk_{j'},\dsk_{j'})\leftarrow\oradk(j')$.
				\item $\{\dpk_j\leftarrow \oradpk(j)\}_{j\in S}$. 
				\item Output $b\leftarrow\detect(\pp,\wpk_i,\dsk_{j'},\{\dpk_j\}_{j \in S},\bm{t})$.
			\end{enumerate}
	}		}
	\captionsetup{justification=centering}
	\caption{\small The oracles for  defining security notions of MDDW}
	\label{fig:MDDW_oracle}
\end{figure}  

\begin{definition}[Soundness] 
	We say that $\mddw$ is \emph{$k$-sound}, if for   any PPT adversary $\A$ and any prompt $\bm{p}\in\tg$,  
 \[\advunf=\Pr[\gameunf=1] \leq \negl \] 
 where $\gameunf$ is shown in Fig. \ref{fig:mddw_cons}, and the predicate $\textsf{\upshape{NOLap}}_k(\bm{t}^*,\bm{t}_1,\bm{t}_2,\cdots)$ in Fig. \ref{fig:mddw_cons} outputs 1 if  $\bm{t}^*$ does not share a $k$-length window of tokens with any of the genuinely-watermarked texts $\bm{t}_1,\bm{t}_2,\cdots$, and outputs 0 otherwise. 
\end{definition}


\begin{definition}[Distortion-freeness] 
	We say that $\mddw$ is \emph{$k$-distortion-free}, if for   any PPT adversary   $\A$,  
 \[\advdistor=|\Pr[\gamedistor=1]-\frac{1}{2}|\leq \negl \] 
 where $\gamedistor$ is shown in Fig. \ref{fig:mddw_cons}, and the predicate $\textsf{\upshape{NOLap}}_k(\bm{t}',\bm{t})$  in Fig. \ref{fig:mddw_cons} outputs 1 if  $\bm{t}'$ does not share a $k$-length window of tokens with $\bm{t}$, and outputs 0 otherwise. 
\end{definition}

\begin{remark}
	Given that MDDW is established in the public-key framework, we adopt the term ``distortion-freeness'' as used in \cite{FGJ+23}.  We stress that our definition of distortion-freeness, like that in \cite{FGJ+23},  is formalized in the multi-query setting, allowing the  distinguisher to make multiple oracle queries adaptively. This  
	 aligns with  the undetectability definition in \cite{CGZ23}. In other words, our distortion-freeness can also be viewed as the undetectability definition from \cite{CGZ23} in the MDDW context. 
\end{remark}


\begin{figure}[!t]\small
	\centering
	\fbox{\parbox{\linewidth}{\shortstack[l]{
				\underline{$\gamerob$:}\\
				$\pp \leftarrow \setup(1^{\lambda})$\\
				$(i^*,S^*,\bm{p}^*,st) \leftarrow \A_1^{\orawk,\oradk,\orawpk,\oradpk,\orawatm,\oradet}(\pp)$\\
				$\qquad$where   $\orawk$ is never queried on $i^*$\\
				$\bm{t}^*\leftarrow\watm(\pp,\wsk_{i^*},\{\dpk_j\}_{j\in S^*},\bm{p}^*)$\\   
				$\qquad$where all keys are the honestly generated outputs of  the   key generation  oracles\\
				$\bm{t} \leftarrow \A_2^{\orawk,\oradk,\orawpk,\oradpk,\orawatm,\oradet}(\bm{t}^*,st)$\\
				$\qquad$where $\textsf{\upshape{NOLap}}_k(\bm{t},\bm{t}^*)=0$, 
				and all queries $i$ to  $\orawk$   satisfy $i\neq i^*$\\
				If $\exists~j\in S^*$ s.t. $\detect(\pp,\wpk_{i^*},\dsk_{j},\bm{t})=0$, then return $1$\\
				Return $0$
	}		}}
	\fbox{\parbox{\linewidth}{\shortstack[l]{
				\underline{$\gameotrdsmddw$:}\\
				$\pp \leftarrow \setup(1^{\lambda})$, $b\leftarrow \{0,1\}$\\
				$b' \leftarrow \A^{\orawk,\oradk,\orawpk,\oradpk,\mathcal{O}_{\text{\emph{otr-chl}}}^{(b)},\oradet}(\pp)$\\
				$\qquad$where  $\mathcal{O}_{\text{\emph{otr-chl}}}^{(0)}(i,S,\bm{p})$ outputs $\orawatm(i,S,\bm{p})$, $\mathcal{O}_{\text{\emph{otr-chl}}}^{(1)}(i,S,\bm{p})$ outputs $\bm{t}\leftarrow\forgeds(\pp,$\\  
				$\qquad$$\wpk_i,\{\dsk_j\}_{j\in S},\bm{p})$ (where all keys are the honestly  generated   outputs of the   key \\
				$\qquad$generation oracles),\\
				$\qquad$let $Q$ denote the set of query-response 
				of  $\mathcal{O}_{\text{\emph{otr-chl}}}^{(b)}$, \\
				$\qquad$all queries $(i,S,\bm{p})$ to $\mathcal{O}_{\text{\emph{otr-chl}}}^{(b)}$   satisfy that  $\orawk$ is never queried on $i$, \\ 
				$\qquad$and  all queries $(i',j',S',\bm{t}')$ to $\oradet$   satisfy   ``$\not\exists~(*,*,*,\bm{t})\in Q$ s.t. $\bm{t}=\bm{t}'$'' \\
				If $b'=b$, then return $1$\\
				Return $0$
			}		
	}}
	\fbox{\parbox{\linewidth}{\shortstack[l]{
				\underline{$\gameotrasmddw$:}\\
				$\pp \leftarrow \setup(1^{\lambda})$, $b\leftarrow \{0,1\}$\\
				$b' \leftarrow \A^{\orawk,\oradk,\orawpk,\oradpk,\mathcal{O}_{\text{\emph{otr-chl}}}^{(b)},\oradet}(\pp)$\\
				$\qquad$where  $\mathcal{O}_{\text{\emph{otr-chl}}}^{(0)}(i,S,\scor,\bm{p})$ outputs $\orawatm(i,S,\bm{p})$, $\mathcal{O}_{\text{\emph{otr-chl}}}^{(1)}(i,S,\scor,\bm{p})$ outputs \\  $\qquad$$\bm{t}\leftarrow\forgeas(\pp,\wpk_i,\{\dpk_j\}_{j\in S},\{\dsk_j\}_{j\in \scor},\bm{p})$ (where all keys are the honestly \\  $\qquad$$\qquad$generated   outputs of the   key generation oracles),\\
				$\qquad$let $Q$ denote the set of query-response 
				of  $\mathcal{O}_{\text{\emph{otr-chl}}}^{(b)}$, \\
				$\qquad$all queries $(i,S,\scor,\bm{p})$ to $\mathcal{O}_{\text{\emph{otr-chl}}}^{(b)}$   satisfy that (1) $\scor\subset S$, (2) $\orawk$ is never  \\
				$\qquad$$\qquad$queried on $i$, and (3) for all $j\in S\setminus\scor$, $\oradk$ is never queried on $j$,  \\ 
				$\qquad$all queries $j$ to   $\oradk$   satisfy  ``$\not\exists~(*,S,\scor,*,*)\in Q$ s.t. $j\in S\setminus\scor$'',  \\ 
				$\qquad$and  all queries $(i',j',S',\bm{t}')$ to $\oradet$   satisfy   ``$\not\exists~(*,*,*,*,\bm{t})\in Q$ s.t. $\bm{t}=\bm{t}'$'' \\
				If $b'=b$, then return $1$\\
				Return $0$
			}		
	}}
	\captionsetup{justification=centering}
	\caption{\small Games  $\gamerob$,  $\gameotrdsmddw$ and  $\gameotrasmddw$ for MDDW, and the oracles are given in Fig.   \ref{fig:MDDW_oracle}}
	\label{fig:mddw_otr}
\end{figure}

\begin{definition}[Robustness] 
	We say that $\mddw$ is \emph{$k$-robust}, if for 
	any PPT  adversary   $\A$,  
 \[\advrob=\Pr[\gamerob=1] \leq \negl \] 
 where $\gamerob$ is shown in Fig. \ref{fig:mddw_otr},  and the predicate $\textsf{\upshape{NOLap}}_k(\bm{t},\bm{t}^*)$  in Fig. \ref{fig:mddw_otr} outputs 1 if  $\bm{t}$ does not share a $k$-length window of tokens with $\bm{t}^*$, and outputs 0 otherwise. 
\end{definition}

\begin{definition}[Off-the-record for designated set] 
	We say that $\mddw$ is \emph{off-the-record for designated set}, if there is a PPT algorithm $\forgeds$, taking $(\pp,\wpk_i,\{\dsk_j\}_{j\in S},\bm{p})$ as input and outputting $\bm{t}$, such that  for any PPT adversary $\A$,  \[\advotrdsmddw=|\Pr[\gameotrdsmddw=1]-\frac{1}{2}| \leq \negl \] where $\gameotrdsmddw$ is shown in Fig. \ref{fig:mddw_otr}.
	\label{def:mddw_otr_ds}
\end{definition}
 

Now, we introduce two \emph{optional} security notions for MDDW: \emph{off-the-record property  for any subset}  and \emph{claimability}.  It's important to note that adherence to these two securities is \emph{not} mandatory for every MDDW scheme.


\begin{definition}[Off-the-record for any subset] 
	We say that $\mddw$ is \emph{off-the-record for any subset}, if there is a PPT algorithm $\forgeas$, taking $(\pp,\wpk_i,$ $\{\dpk_j\}_{j\in S},\{\dsk_j\}_{j\in \scor},\bm{p})$ (where $\scor\subset S$) as input and outputting $\bm{t}$, such that  for any PPT adversary $\A$,  \[\advotrasmddw=|\Pr[\gameotrasmddw=1]-\frac{1}{2}| \leq \negl \] where $\gameotrasmddw$ is shown in Fig. \ref{fig:mddw_otr}.
	\label{def:mddw_otr_as}
\end{definition}


 \begin{definition}[Claimability] 
	We say that $\mddw$ is \emph{$k$-claimable}, if there are two PPT algorithms $\cla$ and $\claver$ (where $\cla$ takes  $(\pp,\wsk_i,\{\dpk_j\}_{j\in S},\bm{t})$ as input and outputs a claim $\pi$, and  $\claver$ takes  $(\pp,\wpk_i,\{\dpk_j\}_{j\in S},\bm{t},\pi)$ as input and outputs a bit), such that 
	\begin{enumerate}
		\item for any   $i$   and any detector  identity set $S$, 	
		\begin{equation*}\small
			\Pr \left[
			\begin{aligned}
				&\pp\leftarrow \setup(1^{\lambda})\\
				&(\wpk_i,\wsk_i) \leftarrow \wkg(\pp)\\
				&((\dpk_j,\dsk_j) \leftarrow \dkg(\pp))_{j\in S}\\
				&\bm{t} \leftarrow \watm(\pp,\wsk_i,\{\dpk_j\}_{j \in S},\bm{p})\\
				&\pi\leftarrow\cla(\pp,\wsk_i,\{\dpk_j\}_{j\in S},\bm{t})
			\end{aligned}
			~\textup{\Large{:}}~
			\begin{aligned}
				&\claver(\pp,\wpk_i,\{\dpk_j\}_{j \in S},\\ &~~~~~~~~~~~~~~~~~~~~~~~~~~~\bm{t},\pi) \neq 1
			\end{aligned}
			\right] = 0.
		\end{equation*}
		\item for any PPT adversary $\A$,  \[\advclaimunfmddw=\Pr[\gameclaimunfmddw=1] \leq \negl \] where $\gameclaimunfmddw$ is shown in Fig. \ref{fig:mddw_claim}, and the predicate $\textsf{\upshape{NOLap}}_k(\bm{t},\bm{t}^*)$  in Fig. \ref{fig:mddw_claim} outputs 1 if  $\bm{t}$ does not share a $k$-length window of tokens with $\bm{t}^*$, and outputs 0 otherwise. 
		\item for any PPT adversary $\A$,  \[\advunframemddw=\Pr[\gameunframemddw=1] \leq \negl \] where $\gameunframemddw$ is shown in Fig. \ref{fig:mddw_claim}.
	\end{enumerate}
	\label{def:mddw_claim}
\end{definition}

\begin{figure}[t!]\small
	\centering
	\fbox{\parbox{\linewidth}{\shortstack[l]{
				\underline{$\gameclaimunfmddw$:}\\
				$\pp \leftarrow \setup(1^{\lambda})$\\
				$(i^*,S^*,\bm{p}^*,st)\leftarrow \A_1^{\orawk,\oradk,\orawpk,\oradpk,\orawatm,\oradet,\oraclm}(\pp)$\\
				$\qquad$where  $\orawk$ is never queried on $i^*$\\ 
				$\bm{t}^*\leftarrow\watm(\pp,\wsk_{i^*},\{\dpk_j\}_{j\in S^*},\bm{p}^*)$ \\
				$\qquad$where all keys are the  honestly generated outputs of  the key generation   oracles\\
				$(\pi^*,i')\leftarrow \A_2^{\orawk,\oradk,\orawpk,\oradpk,\orawatm,\oradet,\oraclm}(\bm{t}^*,st)$\\
				$\qquad$where all queries $i$ to $\orawk$   satisfy $i\neq i^*$, and  all queries $(*,*,\bm{t})$ to $\oraclm$  satisfy\\ $\qquad$$\qquad$$\textsf{\upshape{NOLap}}_{k}(\bm{t},\bm{t}^*)=1$ \\
				If $(\claver(\pp,\wpk_{i'},\{\dpk_j\}_{j\in S^*},\bm{t}^*,\pi^*)=1)\wedge(i'\neq i^*)$, then return $1$\\
				Return $0$				
	}		}}
	\fbox{\parbox{\linewidth}{\shortstack[l]{
				\underline{$\gameunframemddw$:}\\
				$\pp \leftarrow \setup(1^{\lambda})$\\
				$(i^*,S^*,\bm{t}^*,\pi^*)\leftarrow \A^{\orawk,\oradk,\orawpk,\oradpk,\orawatm,\oradet,\oraclm}(\pp)$\\
				$\qquad$where $\orawk$ is never queried on $i^*$, and $\oraclm$ is never queried on  $(i^*,S^*,\bm{t})$\\
				$\qquad$$\qquad$satisfying $\textsf{\upshape{NOLap}}_{k}(\bm{t},\bm{t}^*)=1$\\ 
				If $\claver(\pp,\wpk_{i^*},\{\dpk_j\}_{j\in S^*},\bm{t}^*,\pi^*)=1$ and \\ 
				 $\qquad$$\exists~\mu\in S^*$,  $\detect(\pp,\wpk_{i^*},\dsk_{\mu},\{\dpk_j\}_{j \in S^*},\bm{t}^*)$=1, then return 1\\
				Return $0$				
	}		}}
	\fbox{\parbox{\linewidth}{\shortstack[l]{
				\textbf{Claiming Oracle} $\oraclm(i,S,\bm{t})$:\\
				$\qquad$(1) $(\wpk_i,\wsk_i)\leftarrow\orawk(i)$, $\{\dpk_j\leftarrow \oradpk(j)\}_{j\in S}$.\\
				$\qquad$(2)   Output $\pi\leftarrow\cla(\pp,\wsk_i,\{\dpk_j\}_{j\in S},\bm{t})$.
	}		}}
	\captionsetup{justification=centering}
	\caption{Games  $\gameclaimunfmddw$ and $\gameunframemddw$ for MDDW, and some oracles are given in Fig. \ref{fig:MDDW_oracle}}
	\label{fig:mddw_claim}
\end{figure}  

Here, we provide further explanations. When a model provider intends to claim that certain watermarked texts are generated by their LLM (e.g., to establish copyright), they run  $\cla$ with their watermarking secret key to generate a claim $\pi$. Any one can use $\claver$ to verify the claim. 

The notion of claimability for MDDW is inspired by the  claimability  concepts in ring signature \cite{PS19} and designated-verifier signature  \cite{yamashita2023designated}. Generally speaking, the first requirement in the claimability definition is correctness. The second requirement is that even an  adversarial model provider cannot successfully claim a watermarked text they did not produce. The third requirement is that no one can falsely accuse another provider of generating   watermarked texts.

\section{MDDW construction}\label{sec:mddw_const}


In this section, we present a framework for building    MDDW from MDVS and demonstrate its  fulfillment of the required security properties. Additionally, we show that if the underlying MDVS scheme is off-the-record for any subset, the constructed MDDW also achieves the off-the-record property for any subset. Furthermore, by introducing  the notion of claimability for MDVS, we demonstrate that our generic MDDW scheme achieves claimability when  the underlying MDVS scheme possesses this property. Finally, we provide  a generic method for constructing claimable MDVS by showing a transformation that converts any MDVS scheme into one that is claimable.

\subsection{Generic construction of MDDW}\label{sec:generic_construction}


Let  $\mdvs=(\mdvs.\setup,\skg,\vkg,\sign,\verify)$ be an MDVS scheme with signature length  $\textsf{len}_{\text{sig}}$. Let $\M$  denote  the message space of $\mdvs$, and $\mathcal{SG}$ denote  the signature space of $\mdvs$. Let $H_1:\{0,1\}^*\rightarrow\M$, $H_2:\{0,1\}^*\rightarrow\mathcal{SG}$ and  $H_3:\{0,1\}^*\rightarrow\{0,1\}$ be  hash functions,  which will all be modeled as  random oracles in the security proof. 

Our MDDW scheme $\mddw_{n,\ell,\textsf{len}_{\text{sig}}}=(\setup,\wkg,\dkg,$ $\watm,\detect)$, for some predefined parameters $n$ and $\ell$, is as follows.\footnote{Note that according to Assumption \ref{assumption},  any contiguous block of $\ell$ tokens contains at least $\alpha$ bits of min-entropy.}
\begin{myitemize}
	\item[\textbullet]
	$\setup(1^\lambda)$: Compute $\pp\leftarrow \mdvs.\setup(1^\lambda)$, and return $\pp$ as the public parameter  of MDDW. 
	\item[\textbullet]
	$\wkg(\pp)$: Generate  $(\spk,\ssk)\leftarrow\skg(\pp)$, and set $\wpk=\spk$ and $\wsk=\ssk$. Return $(\wpk,\wsk)$. 
	\item[\textbullet]
	$\dkg(\pp)$: Generate  $(\vpk,\vsk)\leftarrow\vkg(\pp)$, and set $\dpk=\vpk$ and $\dsk=\vsk$. Return $(\dpk,\dsk)$. 
	\item[\textbullet] $\watm(\pp,\wsk_i,\{\dpk_j\}_{j \in S},\bm{p})$: Its description is shown in \textbf{Algorithm \ref{alg:Watm}}.
	\item[\textbullet] $\detect(\pp,\wpk_i,\dsk_{j'},\{\dpk_j\}_{j \in S},\bm{t})$: Its description is shown in \textbf{Algorithm \ref{alg:Det}}. 
\end{myitemize}

\begin{algorithm}[h]\small
	\caption{$\watm(\pp,\wsk_i,\{\dpk_j\}_{j \in S},\bm{p})$ \label{alg:Watm} }
	\begin{algorithmic}[1]
		\State $\ssk_i\leftarrow \wsk_i$, $\{\vpk_j\}_{j \in S}\leftarrow\{\dpk_j\}_{j \in S}$, $\bm{t}\leftarrow\epsilon$
		
		\While{$|\bm{t}|+\ell+\ell\cdot\textsf{len}_{\text{sig}}<n$}
		
		\State $\bm{t}\leftarrow\bm{t}||\genmodel_\ell(\bm{p},\bm{t})$
		
		\State $\hat{\bm{\sigma}}\leftarrow\sign(\pp,\ssk_i,\{\vpk_j\}_{j\in S}, H_1(\bm{t}[-\ell:]))$
		
		\State $\bm{\sigma}\leftarrow \hat{\bm{\sigma}} \oplus H_2(\bm{t}[-\ell:])$
		
		\State $\bm{\sigma}_{\textup{prev}}\leftarrow\epsilon$, 		
		$~\bm{m}\leftarrow\epsilon$
		
		\While{$\bm{\sigma}\neq\epsilon$}		
		
		\State $\sigma_{\text{bit}}\leftarrow \bm{\sigma}[1]$, $\bm{\sigma}\leftarrow\bm{\sigma}[2:]$
		
		\State $\bm{x}\leftarrow\genmodel_\ell(\bm{p},\bm{t})$
		
		\While{$H_3(\bm{m}||\bm{x}||\bm{\sigma}_{\textup{prev}})\neq\sigma_{\text{bit}}$}
		\State $\bm{x}\leftarrow\genmodel_\ell(\bm{p},\bm{t})$
		\EndWhile
		\State $\bm{m}\leftarrow\bm{m}||\bm{x}$, $~\bm{t}\leftarrow\bm{t}||\bm{x}$, $\bm{\sigma}_{\textup{prev}}\leftarrow\bm{\sigma}_{\textup{prev}}||\sigma_{\text{bit}}$
		\EndWhile
		\EndWhile
		\If{$|\bm{t}|<n$}
		\State $\bm{t}\leftarrow\bm{t}||\genmodel_{n-|\bm{t}|}(\bm{p},\bm{t})$
		\EndIf
		\State \Return $\bm{t}$
	\end{algorithmic}
\end{algorithm}
\begin{algorithm}[h!]\small
	\caption{$\detect(\pp,\wpk_i,\dsk_{j'},\{\dpk_j\}_{j \in S},\bm{t})$ \label{alg:Det} }
	\begin{algorithmic}[1]
		\State $\spk_i\leftarrow \wpk_i$, $\vsk_{j'}\leftarrow\dsk_{j'}$, $\{\vpk_j\}_{j \in S}\leftarrow\{\dpk_j\}_{j \in S}$
		
		\For{$\mu=1,2,\cdots,n-(\ell+\ell\cdot\textsf{len}_{\text{sig}})$}
		\State $\widetilde{m}\leftarrow H_1(\bm{t}[\mu:\mu+\ell-1])$,  $\bm{\sigma}\leftarrow\epsilon$, 
		$\bm{m}\leftarrow\epsilon$
		
		\For{$\varphi=1,2,\cdots,\textsf{len}_{\text{sig}}$}
		\State $\bm{m}\leftarrow\bm{m}||\bm{t}[\mu+\varphi\cdot\ell:\mu+\varphi\cdot\ell+\ell-1]$
		
		\State $\bm{\sigma}\leftarrow\bm{\sigma}||H_3(\bm{m}||\bm{\sigma})$
		\EndFor
		
		\State $\hat{\bm{\sigma}}\leftarrow \bm{\sigma} \oplus H_2(\bm{t}[\mu:\mu+\ell-1])$
		
		\If{$\verify(\pp,\spk_i,\vsk_{j'},\{\vpk_j\}_{j \in S},\widetilde{m},\hat{\bm{\sigma}})=1$}
		\State \Return 1
		\EndIf
		\EndFor
		\State \Return 0
	\end{algorithmic}
\end{algorithm}


\vspace{1mm}
Now, we show that our proposed $\mddw$ satisfies the required security properties defined in Sec. \ref{sec:MDDW_def}. 
Formally, we have the following theorem.

\begin{theorem}\label{thm:MDDW_required_security} 
The constructed MDDW scheme $\mddw$ possesses the following properties:
\begin{itemize}
	\item \emph{(Completeness).} If  Assumption \ref{assumption} holds and  $\mdvs$ satisfies correctness, then $\mddw$ is  $\ell$-complete. 
	\item \emph{(Consistency).} If the underlying $\mdvs$ is consistent, then  $\mddw$ is also  consistent. 
	\item \emph{(Soundness).} If the underlying $\mdvs$ is   unforgeable, then $\mddw$ is $\ell$-sound. 
	\item \emph{(Distortion-freeness).} 	If   Assumption \ref{assumption} holds, then $\mddw$ is $\ell$-distortion-free. 
	\item \emph{(Robustness).} 	$\mddw$ is $(2\textsf{\upshape{len}}_{\text{\upshape{sig}}}+2)\ell$-robust. 
	\item \emph{(Off-the-record for designated set).} If the underlying $\mdvs$ is off-the-record for designated set, then $\mddw$ is also off-the-record for designated set. 
\end{itemize}
\end{theorem}

\begin{remark}
	If the underlying MDVS scheme $\mdvs$  meets strong unforgeability, $\mddw$ might achieve  $(\textsf{len}_{\text{sig}}+1)\ell$-soundness. 
\end{remark}

For the optional off-the-record property for any subset, we have the following theorem. 
\begin{theorem}[Off-the-record for any subset]\label{thm:OTR_as}
	If $\mdvs$ is off-the-record for any subset, then $\mddw$ is also off-the-record for any subset. 
\end{theorem}

%
%
%

The proofs of Theorem \ref{thm:MDDW_required_security} and Theorem \ref{thm:OTR_as} are placed in Appendix \ref{sec:security_proof_append} and Appendix \ref{sec:security_proof_append_otr_as}, respectively.




\subsection{MDDW construction with claimability}\label{sec:MDVS_claimaibility}
To make the above generic MDDW scheme  achieve claimability, the underlying MDVS scheme needs to meet some corresponding  security property. However, to the best of our knowledge, no claimability notion for MDVS has been introduced before. 
Here, we firstly introduce the notion of claimability  for MDVS, and then formally prove that if   the underlying MDVS scheme meets claimability, the above generic MDDW scheme is claimable. 
\vspace{1mm}

\noindent\underline{\textbf{Claimability for MDVS.}} The notion of claimability  for MDVS extends from the  established claimability  concepts for ring signature \cite{PS19} and  designated-verifier signature \cite{yamashita2023designated}. It is crucial to emphasize that in this paper, claimability is defined as an optional requirement for MDVS. 



\begin{definition}[Claimability for MDVS] 
	We say that $\mdvs=(\setup,\skg,$ $\vkg,\sign,\verify)$ is \emph{claimable}, if there are two PPT algorithms $\cla$ and $\claver$ (where $\cla$ takes  $(\pp,\ssk_i,\{\vpk_j\}_{j\in S},\sigma)$ as input and outputs a claim $\pi$, and  $\claver$ takes  $(\pp,\spk_i,\{\vpk_j\}_{j\in S},\sigma,\pi)$ as input and outputs a bit), such that 
	\begin{enumerate}
		\item  for any signer $i$, any message $m \in \M$, and any verifier identity  set $S$, 	
		\begin{equation*}\small
			\Pr \left[
			\begin{aligned}
				&\pp\leftarrow \setup(1^{\lambda})\\
				&(\spk_i,\ssk_i) \leftarrow \skg(\pp)\\
				&\{(\vpk_j,\vsk_j) \leftarrow \vkg(\pp)\}_{j\in S}\\
				&\sigma \leftarrow \sign(\pp,\ssk_i,\{\vpk_j\}_{j \in S}, m)\\
				&\pi\leftarrow\cla(\pp,\ssk_i,\{\vpk_j\}_{j\in S},\sigma)
			\end{aligned}
			~\textup{\Large{:}}~
			\begin{aligned}
				&\claver(\pp,\spk_i,\{\vpk_j\}_{j \in S}, \sigma,\pi) \neq 1
			\end{aligned}
			\right] = 0.
		\end{equation*}
		\item for any PPT adversary $\A=(\A_1,\A_2)$,  \[\advclaimunfmdvs=\Pr[\gameclaimunfmdvs=1] \leq \negl \] where $\gameclaimunfmdvs$ is shown in Fig. \ref{fig:mdvs_claim}.
		\item for any PPT adversary $\A$,  \[\advunframemdvs=\Pr[\gameunframemdvs=1]\leq \negl \] where $\gameunframemdvs$ is shown in Fig. \ref{fig:mdvs_claim}.
	\end{enumerate}
	\label{def:mdvs_claim}
\end{definition}

\begin{figure}[t!]\small
	\centering
	\fbox{\parbox{\linewidth}{\shortstack[l]{
				\underline{$\gameclaimunfmdvs$:}\\
				$\pp \leftarrow \setup(1^{\lambda})$\\
				$(i^*,S^*,m^*,st)\leftarrow \A_1^{\orask,\oravk,\oraspk,\oravpk,\orasign,\oraver,\oraclm}(\pp)$\\
				$\qquad$where  $\orask$ is never queried on $i^*$\\ 
				$\sigma^*\leftarrow\sign(\pp,\ssk_{i^*},\{\vpk_j\}_{j\in S^*},m^*)$ \\
				$\qquad$where all keys are the  honestly generated outputs of  the key generation   oracles\\
				$(\pi^*,i')\leftarrow \A_2^{\orask,\oravk,\oraspk,\oravpk,\orasign,\oraver,\oraclm}(\sigma^*,st)$\\
				$\qquad$where all queries $i$ to $\orask$   satisfy $i\neq i^*$, and all queries $(*,*,\sigma)$ to $\oraclm$ satisfy \\
				$\qquad$$\qquad$$\sigma\neq \sigma^*$\\
				If $(\claver(\pp,\spk_{i'},\{\vpk_j\}_{j\in S^*},\sigma^*,\pi^*)=1)\wedge(i'\neq i^*)$, then return $1$\\
				Return $0$				
	}		}}
	\fbox{\parbox{\linewidth}{\shortstack[l]{
				\underline{$\gameunframemdvs$:}\\
				$\pp \leftarrow \setup(1^{\lambda})$\\
				$(i^*,S^*,m^*,\sigma^*,\pi^*)\leftarrow \A^{\orask,\oravk,\oraspk,\oravpk,\orasign,\oraver,\oraclm}(\pp)$\\
				$\qquad$where $\orask$ is never queried on $i^*$, and $\oraclm$ is never queried on  $(i^*,S^*,\sigma^*)$ \\
				If $\claver(\pp,\spk_{i^*},\{\vpk_j\}_{j\in S^*},\sigma^*,\pi^*)=1$, then return 1 \\ 
				Return $0$				
	}		}}
	\fbox{\parbox{\linewidth}{\shortstack[l]{
				\textbf{Claiming Oracle} $\oraclm(i,S,\sigma)$:\\
				$\qquad$(1) $(\spk_i,\ssk_i)\leftarrow\orask(i)$, $\{\vpk_j\leftarrow \oravpk(j)\}_{j\in S}$.\\
				$\qquad$(2)   Output $\pi\leftarrow\cla(\pp,\ssk_i,\{\vpk_j\}_{j\in S},\sigma)$.
	}		}}
	\captionsetup{justification=centering}
	\caption{Games  $\gameclaimunfmdvs$ and $\gameunframemdvs$ for MDVS}
	\label{fig:mdvs_claim}
\end{figure}  

\begin{remark}
In the above definition, the adversary $\A$ succeeds in game $\gameunframemdvs$ if  its output tuple $(i^*,S^*,m^*,\sigma^*,\pi^*)$ satisfies  $\claver(\pp,\spk_{i^*},\{\vpk_j\}_{j\in S^*},\sigma^*,\pi^*)$ $=1$.  It's worth noting that unlike the prerequisites in the definition of claimability for ring signatures \cite{PS19} and designated-verifier signatures \cite{yamashita2023designated}, our third condition doesn't mandate that ``$\sigma^*$ is accepted by some designated verifier in $S^*$''. Thus, the security requirement in our claimability for MDVS appears to be more stringent. 
\end{remark}

\noindent\underline{\textbf{Generic MDDW scheme with claimability.}} 
For the claimability of our generic MDDW scheme, we have the following theorem. 

\begin{theorem}[Claimability]\label{thm:claimability}
	If the underlying MDVS scheme $\mdvs$ is claimable, then the generic MDDW scheme $\mddw$ in Sec. \ref{sec:generic_construction} is  $(\textsf{\upshape{len}}_{\text{\upshape{sig}}}+1)\ell$-claimable. 
\end{theorem}

The proof of Theorem \ref{thm:claimability}  is  placed in Appendix \ref{sec:security_proof_append_claim}.

\subsection{Instantiation of claimable MDVS}
\label{sec:trans}
 To instantiate claimable MDVS (CMDVS), we show a transformation that converts any MDVS  into a CMDVS, with the help of   a standard digital  signature, a pseudorandom function and a commitment scheme.  Our method is inspired by \cite{PS19,yamashita2023designated}, where \cite{PS19} shows how to construct claimable ring signature and \cite{yamashita2023designated} shows how to construct claimable DVS from ring signature.

The intuition of our method is as follows. To generate a CMDVS signature, the signer firstly generates an MDVS signature $\sigma_{\mdvs}$ with the signing algorithm of the underlying MDVS,  signs $\sigma_{\mdvs}$ with the the standard signature scheme to obtain a standard signature $\sigma_{\sig}$,  
and then takes  the commitment scheme to commit $\sigma_{\sig}$, obtaining a commitment $\com$. The generated  CMDVS signature consists of  $(\sigma_{\mdvs},\com)$.  When making a  claim, the signer just opens the commitment $\com$, outputting    $\sigma_{\sig}$ and the randomness used to generate $\com$. To verify the claim, one  firstly checks if the opening is correct, and then checks if the standard signature is valid. The unforgeability of the standard signature scheme guarantees that the claim is indeed generated by the signer.  

The detailed construction of CMDVS is as follows.


Let $\mdvs=(\setup, \skg,\vkg,\sign,\verify)$ be an MDVS scheme. 
Let $\sig=(\setup,\kg,\sign,\verify)$ be a  signature scheme, $\prf=(\kg,\eval)$ be a pseudorandom function, and $\commit=(\setup,\Com,\dec)$ be a commitment scheme. The definitions of signature, pseudorandom function and commitment are given in Appendix \ref{app:preliminary} for completeness.


Our CMDVS scheme  $\cmdvs=(\setup,\skg,\vkg,\sign,\verify,\cla,$ $\claver)$ is shown in Fig. \ref{fig:algo_cmdvs}.


\begin{figure*}[!h]
    \small
	\centering
		\fbox{
  \begin{minipage}
				[t]{0.98\textwidth}{
                \underline{$\setup(1^{\lambda})$}:\\
                $\pp_{\mdvs} \leftarrow \mdvs.\setup(1^{\lambda})$, 
                $\pp_{\sig} \leftarrow \sig.\setup(1^{\lambda})$,
                $\pp_{\commit} \leftarrow \commit.\setup(1^{\lambda})$\\
                Return $\pp=(\pp_{\mdvs},\pp_{\sig},\pp_{\commit})$\\

                \underline{$\skg(\pp)$}:\\
                $k \leftarrow \prf.\kg(1^{\lambda})$, 
                $(\sgsk,\sgpk) \leftarrow \sig.\kg(\pp_{\sig})$\\ 
                $(\rssk,\rspk) \leftarrow \mdvs.\skg(\pp_{\mdvs})$\\
               Return $\spk=(\sgpk,\rspk)$, $\ssk=(k,\sgsk,\rssk)$\\

               \underline{$\vkg(\pp)$}:\\
                Return $(\rvsk,\rvpk) \leftarrow \mdvs.\vkg(\pp_{\mdvs})$\\
                
                \underline{$\sign(\pp,\ssk_i,\{\vpk_j\}_{j\in S},m)$}:\\
                $\sigma_{\mdvs} \leftarrow \mdvs.\sign(\pp_{\mdvs},\rssks,\{\vpk_j\}_{j \in S},m)$\\
                $r_{\sig} \leftarrow \prf.\eval(\ks,(\spk_i,\sigma_{\mdvs},0))$, 
                $\sigma_{\sig} \leftarrow \sig.\sign(\pp_{\sig},\sgsks,(\spk_i,\sigma_{\mdvs});r_{\sig})$\\
                $r_{\commit} \leftarrow \prf.\eval(\ks,(\spk_i,\sigma_{\mdvs},1))$\\ 
                $\com \leftarrow \commit.\Com(\pp_{\commit},(\spk_i,\{\vpk_j\}_{j\in S},\sigma_{\sig});r_{\commit})$\\
                Return $\sigma = (\sigma_{\mdvs},\com)$\\

                    \underline{$\verify(\pp,\spk_i,\vsk_{j'},\{\vpk_j\}_{j \in S},\sigma,m)$}:\\
                Return $b\leftarrow \mdvs.\verify(\pp_{\mdvs},\rspks,\vsk_{\mdvs,j'},\{\vpk_{\mdvs,j}\}_{j \in S},m,\sigma_{\mdvs})$\\


                    \underline{$\cla(\pp,\ssk_i,\{\vpk_j\}_{j \in S},\sigma)$}:\\
                    $r_{\sig} \leftarrow \prf.\eval(\ks,(\spk_i,\sigma_{\mdvs},0))$, 
                    $r_{\commit} \leftarrow \prf.\eval(\ks,(\spk_i,\sigma_{\mdvs},1))$\\ 
                    $\sigma_{\sig} \leftarrow \sig.\sign(\pp_{\sig},\sgsks,(\spk_i,\sigma_{\mdvs});r_{\sig})$\\
                    If $\commit.\dec(\pp_{\commit},\com,r_{\commit},(\spk_i,\{\vpk_j\}_{j\in S},\sigma_{\sig}))=0$: Return $\perp$\\
                    Return $\pi=(r_{\commit},\sigma_{\sig})$\\
                    
                    \underline{$\claver(\pp,\spk_i,\{\vpk_j\}_{j \in S},\sigma,\pi)$}:\\
                    If $\com \neq \commit.\Com(\pp_{\commit},(\spk_i,\{\vpk_j\}_{j\in S},\sigma_{\sig});r_{\commit})$: 
                    Return $0$\\
                    Return $b = \sig.\verify(\pp_{\sig},\sgpks,(\spk_i,\sigma_{\mdvs}),\sigma_{\sig})$
			 	}
			 \end{minipage}
}	
	\captionsetup{justification=centering}
	\caption{Construction of $\cmdvs$
    } 
	\label{fig:algo_cmdvs}
\end{figure*}

For security, we present the following theorem, the proof of which is given in Appendix \ref{sec:security_proof_append_cmdvs}. 

\begin{theorem}\label{thm:CMDVS_security} 
	If $\mdvs$ satisfies correctness, consistency, unforgeability, and off-the-record property for designated set (resp., for any subset), $\sig$ satisfies  unforgeability, and $\commit$ is hiding and binding, then $\cmdvs$ is an MDVS scheme achieving correctness, consistency,  unforgeability,  off-the-record property for designated set (resp., for any subset), and claimability. 
\end{theorem}

\section{Evaluation}\label{sec:evaluation}

In this experiment section, we evaluate our proposed scheme from the perspective of computational overhead. Specifically, we analyze the time required for text generation and watermark detection using three different LLMs. The details are as follows:



\vspace{1mm}
\noindent \underline{\textbf{Schemes.}} To highlight the performance of our scheme, we will implement the following schemes.
\begin{enumerate}
    \item \textbf{Original LLMs.} The first scheme is the original LLM scheme without any watermarking solutions. We set the performance of this scheme as the baseline for our experiments. 
%
\item \underline{$\pdws$.} This  is the publicly detectable watermarking scheme \cite{FGJ+23} based on the BLS signature \cite{boneh2004short}, which is recalled in Appendix \ref{sec:concrete_scheme} for completeness. Although this scheme \cite{FGJ+23} does not actually achieve  distortion-freeness (when the  public keys are known, as discussed in the Introduction), it is more relevant to compare our scheme with an asymmetric solution in our experiments. For those interested in the performance of a symmetric solution (e.g., \cite{CGZ23}),  \cite{FGJ+23} shows that its performance exceeds that of \cite{CGZ23} in both text generation and watermark detection. The BLS signature is recalled in Appendix \ref{sec:concrete_scheme} for completeness. 
    \item \textbf{DDW.} This is the designated-detector watermarking  scheme built within the MDDW framework (in Sec. \ref{sec:mddw_const}) when $|S|=1$, using a concrete designated-verifier signature  \cite{LV04,steinfeld2003universal}. The details of $\ddw$ are given in Appendix \ref{sec:efficientddw}.
    \item \textbf{MDDW.} This is the  MDDW scheme  based on the framework in Sec. \ref{sec:mddw_const} with the underlying  MDVS    proposed in \cite{au2014strong}, which is recalled in Appendix \ref{sec:concrete_scheme}.
\end{enumerate}


\vspace{1mm}
\noindent \underline{\textbf{Metrics.}} 
%
In this paper, we evaluate our scheme from the perspective of computational overhead. We measure the computational overhead  by the \emph{time of text generation} and the \emph{time of watermark detection}. More exactly,
\begin{itemize}
    \item \underline{Time of text generation.} 
    It measures the time required to 
    generate the text of a pre-defined length (e.g., $50$, $100$, etc.). 
    \item \underline{Time of watermark detection.} It measures the time required to successfully detect a watermark in a given watermarked text.
\end{itemize}


\vspace{1mm}
\noindent  \underline{\textbf{Implementation.}}  
We implement our experiments in Python language (version 3.8.10) based on elliptic curve groups with key size of 128 bits. The experiments are conducted on a PC, running Ubuntu 20.04 on one \textup{Intel}$^\circledR$ $\textup{Core}^\text{TM}$ i9-10900K CPU@3.70 GHz and NVIDIA GeForce RTX 3080, and using 64 GB memory.
We employ three LLMs to evaluate our different schemes, including \textbf{BLOOMZ} 3B model \cite{muennighoff2022crosslingual},  \textbf{OPT} 1.3B model \cite{zhang2022opt} and \textbf{Gemma} 2B model \cite{team2024gemma}.


Aligned with prior works, we embed our watermarks in \emph{tokens}, (each contains $3$ to $5$ words), to ensure sufficient entropy. 
Aligned with prior works, we embed   watermarks in \emph{tokens} of $3$ to $5$ words for  sufficient entropy. 
The   parameters are set empirically 
(e.g., when a token contains $5$ words in BLOOMZ, the outputs of the watermarking schemes  contain no odd text), rather than from rigorous theoretical analysis. 


\vspace{1mm}
\noindent \underline{\textbf{Experimental results.}} Here, we present the experimental results. These are illustrated in Fig. \ref{fig:all}. Specifically, Fig. \ref{fig:1} shows the generation time when using the BLOOMZ model.
Fig. \ref{fig:2} displays the generation time when using the OPT model.
Fig. \ref{fig:3} illustrates the generation time when using the Gemma model.
Fig. \ref{fig:4} depicts the time required for watermark detection when using the BLOOMZ model. Note that the computation of watermark detection is unrelated to the LLMs, hence we present only the results for the BLOOMZ model.


\begin{figure}[t]
	\centering
	\captionsetup[subfloat]{captionskip=1.0pt,farskip=1.0pt}\captionsetup{justification=centering}
	\subfloat[Generation time for BLOOMZ]{\label{fig:1}\includegraphics[width=0.43\textwidth]{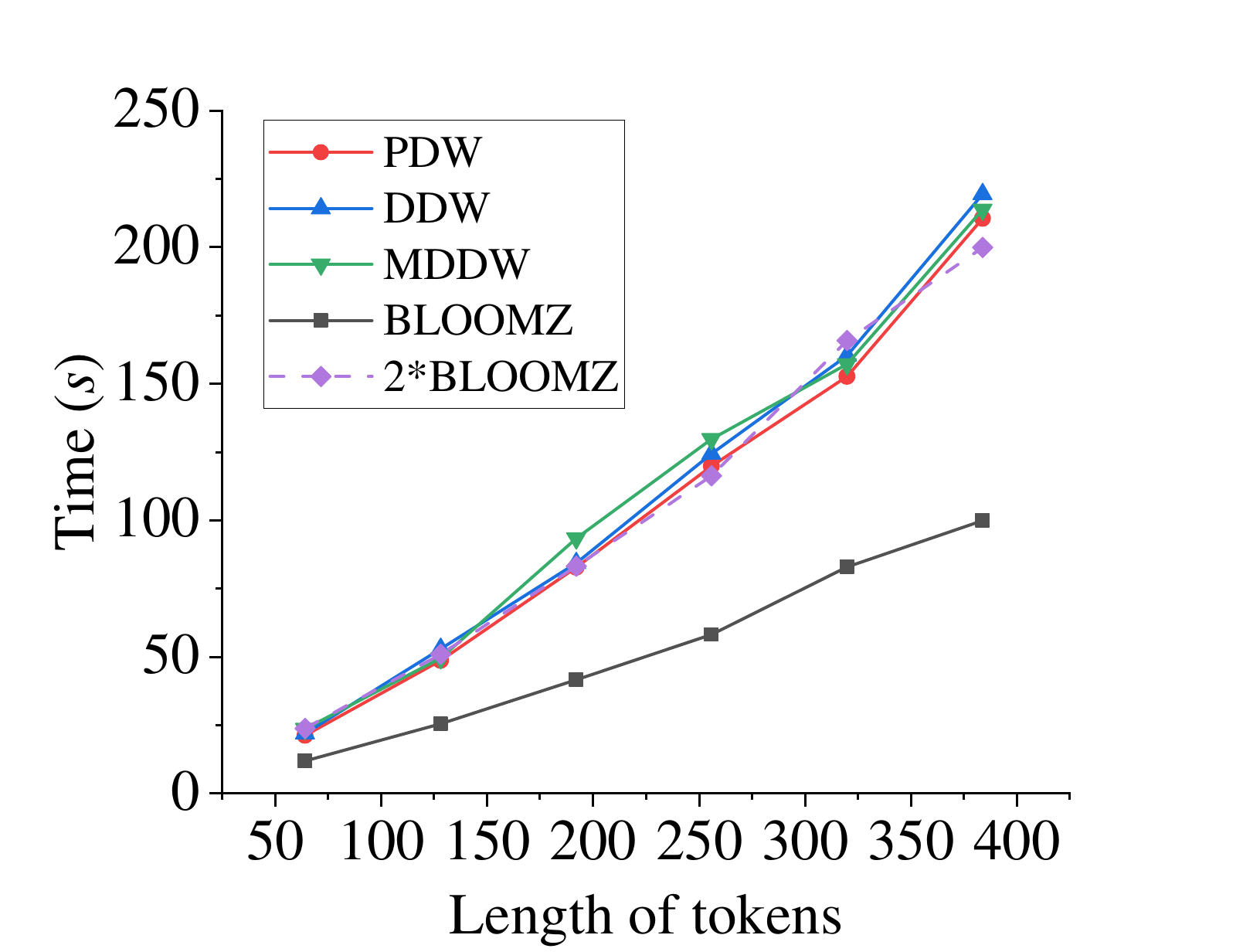}}
	\hspace{0.01\textwidth}
	\subfloat[Generation time for OPT]{\label{fig:2}\includegraphics[width=0.43\textwidth]{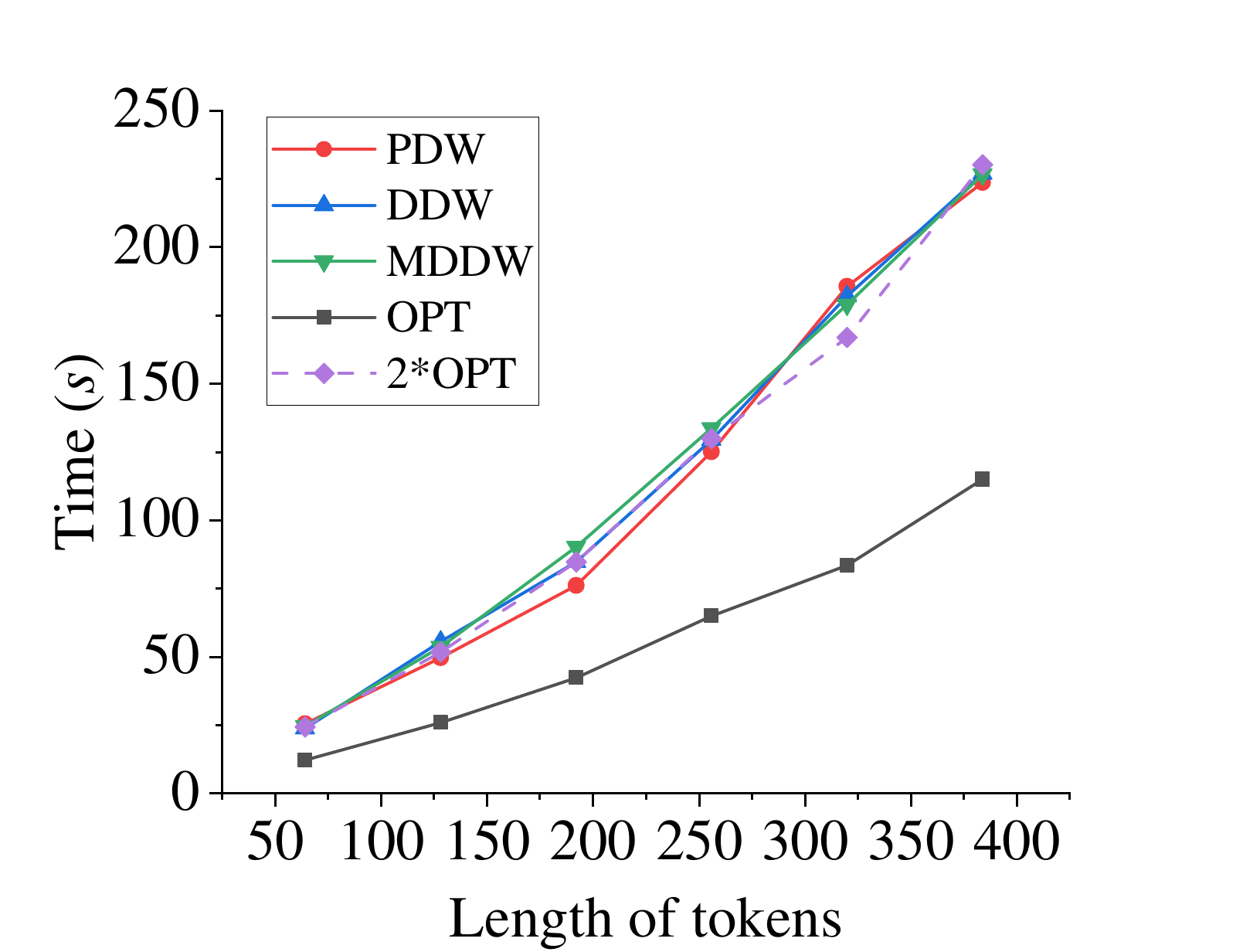}}\\
	\subfloat[Generation time for Gemma]{\label{fig:3}\includegraphics[width=0.43\textwidth]{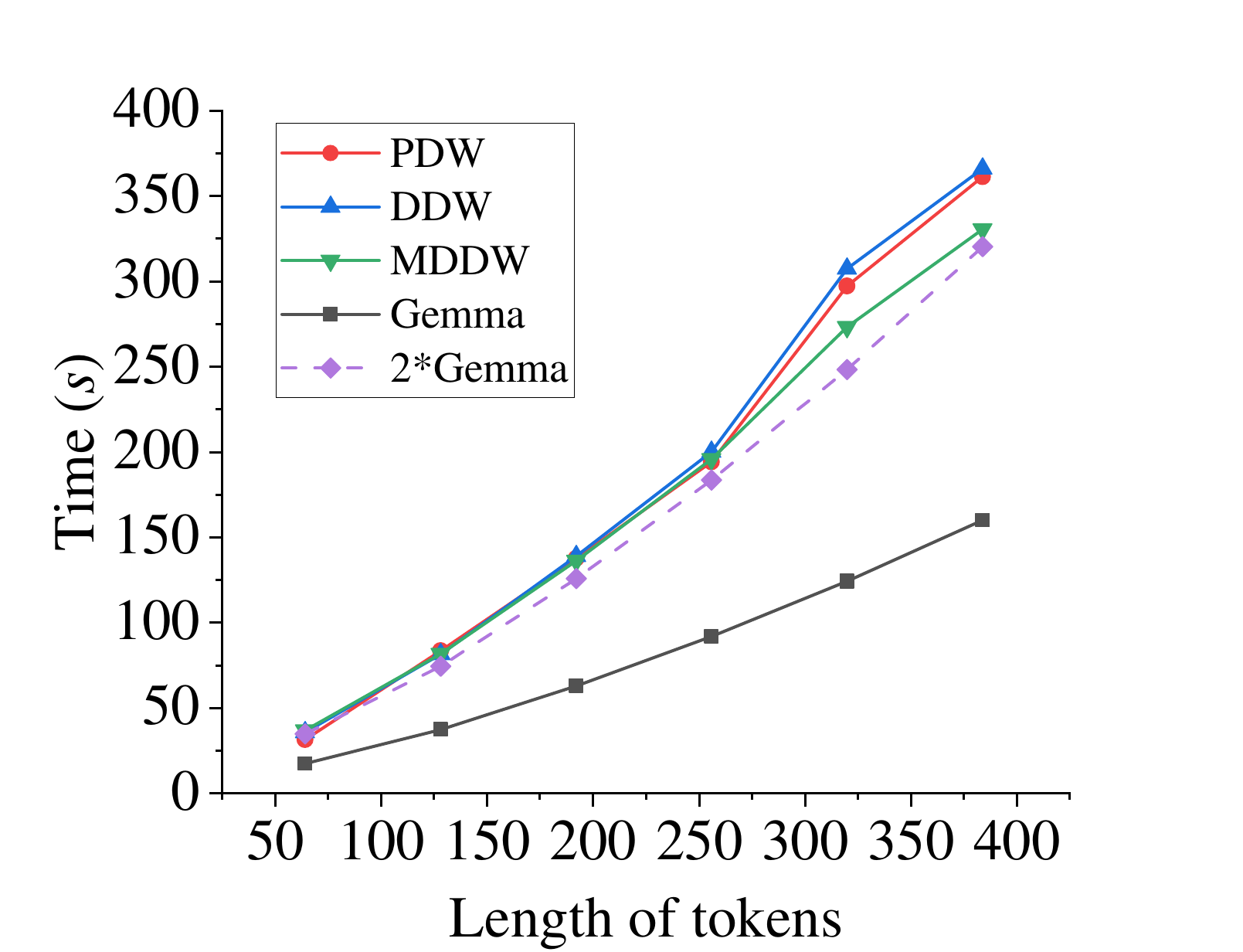}}
	\hspace{0.01\textwidth}
	\subfloat[Watermark detection time]{\label{fig:4}\includegraphics[width=0.43\textwidth]{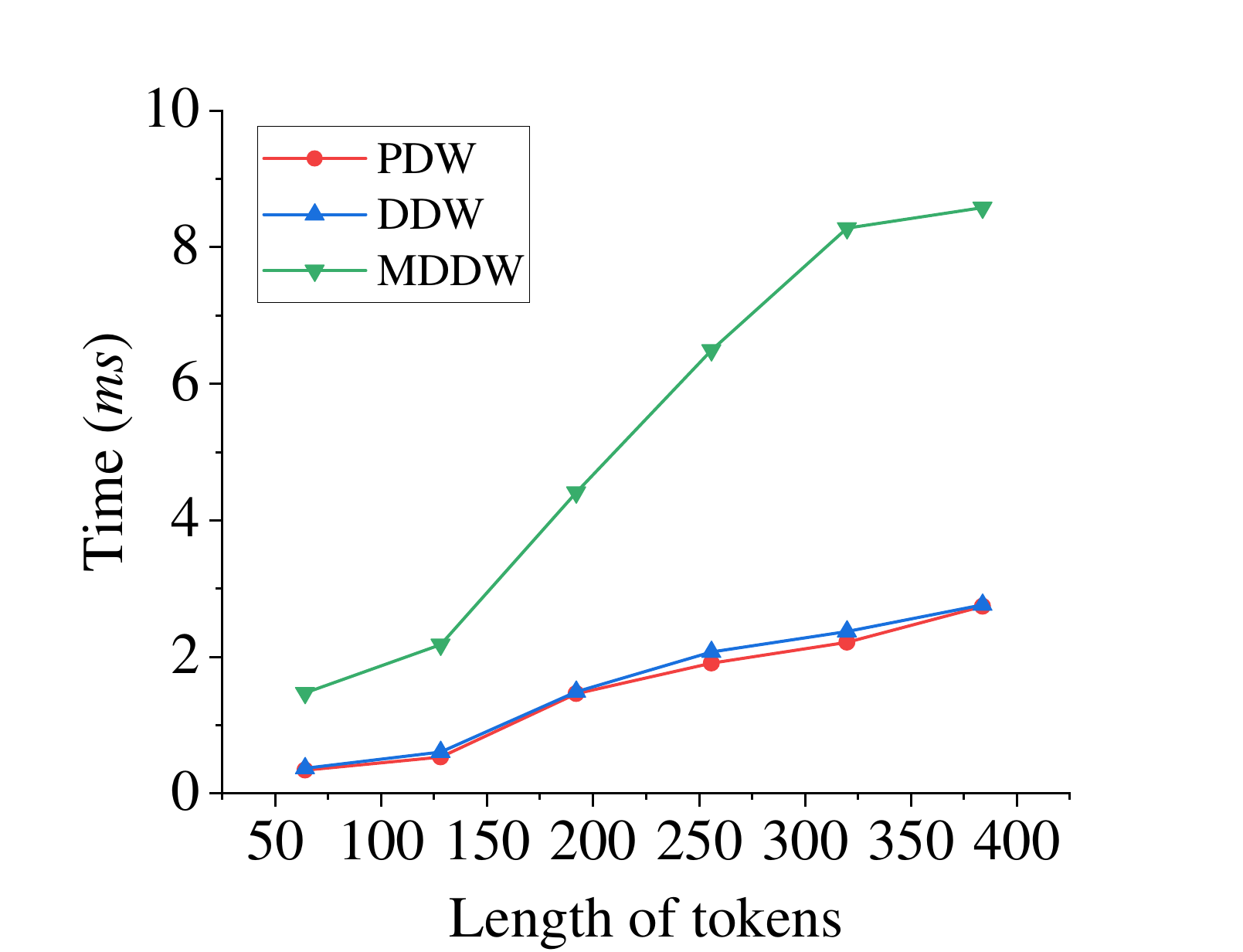}}\\
\setlength{\abovecaptionskip}{2pt}
\setlength{\belowcaptionskip}{-12pt}
\caption{Figures for time of text generation and watermark detection (the dashed lines are twice of the generation time of the corresponding LLMs)}
	\label{fig:all}
\end{figure}

\vspace{1mm}
\noindent \underline{\emph{Time of text generation.}} From Fig. \ref{fig:1} to Fig. \ref{fig:3}, we observe that the text generation time for our schemes ($\ddw$ and $\mddw$) is nearly identical to that of $\pdws$. However, it is important to note that $\pdws$ cannot support distortion-freeness, while ours not only satisfies distortion-freeness, but also satisfies other advanced security properties (e.g., off-the-record property).

The dashed lines in Fig. \ref{fig:1} to Fig. \ref{fig:3} indicate twice the text generation times of the original LLMs. It is evident that the watermarking schemes, including $\pdws$ and our proposed methods, exhibit approximately the same time shown by the dashed line.
Since both the scheme in $\pdws$ and our schemes employ rejection sampling, there is only approximately $1/2$ probability that a randomly sampled token meets the requirement. Theoretically, the expected time to sample a token is about $1.5$ times that of the original LLM. 
The remaining time is dedicated to additional computations, such as generating signatures, hash computation, etc.

From Fig. \ref{fig:1} to Fig. \ref{fig:3}, it appears that   $\ddw$
does not offer advantages over   $\mddw$.
In fact, these figures show the time of text generation when the output lengths are the same. Given the same output length, the number of watermarks in  $\mddw$
is only one-fourth of that in $\ddw$
in our experiments. Furthermore, from the perspective of generating a watermark (excluding the embedding time), the generation time for four watermarks in $\ddw$ is almost the same as that for one watermark in $\mddw$.
Thus, the DDW in Appendix \ref{sec:efficientddw} is more efficient when considering only one designated verifier.

\vspace{1mm}
\noindent \underline{\emph{Time of watermark detection.}} 
Fig. \ref{fig:4} shows the time required for watermark detection. From the figure, it is evident that the detection time for   $\ddw$
is comparable to that of $\pdws$ and significantly shorter than that of   $\mddw$.
A detailed theoretical analysis in Appendix \ref{sec:efficientddw} reveals that   $\mddw$
demands more exponential computations, while $\ddw$ relies on one bilinear map computation (so does $\pdws$ using the BLS signature with minor modifications in Appendix \ref{sec:concrete_scheme}). Although it is challenging to determine which type of computation is more time-consuming purely from a theoretical standpoint, the experimental results clearly indicate that $\ddw$ is more efficient when considering only one designated verifier. Additionally, its performance is very close to that of $\pdws$.

Overall, our scheme offers more advanced functionalities and greater flexibility than existing schemes (e.g., $\pdws$), and the performance of our solutions is also acceptable. 




 \vspace{2mm}
\noindent \textbf{Acknowledgements.} 
We thank Sam Gunn and Miranda Christ for their  valuable  feedback on the initial version of this paper. 

\bibliographystyle{alpha}
\bibliography{soak}

\appendix
\section{Preliminaries: cryptographic assumptions and lemmas}
\label{sec:assumption}
Let $\genG(1^{\lambda})$ be a bilinear map setup algorithm, which takes the security parameter $\lambda$ as input and outputs $\pp=(e,\G,\G_T,g,p)$, where $p$ is a prime of $\Theta(\lambda)$ bits, $\G$  and $\G_T$ are cyclic groups of order $p$, $g$ is  a generator of $\G$, and $e: \G  \times \G  \rightarrow \G_T$ is a nondegenerate and efficiently computable bilinear map. 


\begin{definition}[The DBDH assumption] 
	We say that the \emph{decisional bilinear Diffie-Hellman (DBDH) assumption} holds with respect to  $\genG$,
	if for any PPT adversary $\mathcal{D}$,
	\[\textup{\textbf{Adv}}^{\textup{dbdh}}_{\genG,\mathcal{D}}(\lambda):=|\textup{Pr}[\textup{\textbf{G}}^{\textup{dbdh}}_{\genG, \mathcal{D}}(\lambda)=1]-\frac{1}{2}|\leq\negl,\]
	 where $\textbf{\upshape{G}}^{\textup{dbdh}}_{\genG, \mathcal{D}}(\lambda)$ is in Fig. \ref{fig:keadlpddh}.
\end{definition}

\begin{definition}[The CBDH assumption] 
	We say that the \emph{computational bilinear Diffie-Hellman (CBDH) assumption} holds  with respect to  $\genG$, if for any PPT adversary $\mathcal{A}$,
	 \[\textup{\textbf{Adv}}^{\textup{cbdh}}_{\genG,\mathcal{A}}(\lambda):=\textup{Pr}[\textup{\textbf{G}}^{\textup{cbdh}}_{\genG, \mathcal{A}}(\lambda)=1]\leq\negl,\] 
	  where $\textbf{\upshape{G}}^{\textup{cbdh}}_{\genG, \mathcal{A}}(\lambda)$ is in Fig. \ref{fig:keadlpddh}.
\end{definition}

\begin{definition}[The GBDH assumption] 
	We say that the \emph{gap bilinear Diffie-Hellman (GBDH) assumption} holds with respect to  $\genG$, if for any PPT adversary $\mathcal{A}$,
	\[\textup{\textbf{Adv}}^{\textup{gbdh}}_{\genG,\mathcal{A}}(\lambda):=\textup{Pr}[\textup{\textbf{G}}^{\textup{gbdh}}_{\genG, \mathcal{A}}(\lambda)=1]\leq\negl,\] 
  where $\textbf{\upshape{G}}^{\textup{gbdh}}_{\genG, \mathcal{A}}(\lambda)$ is in Fig. \ref{fig:keadlpddh}.
\end{definition}

    \begin{figure}[!t]\scriptsize
        \centering
        \fbox{
            \begin{minipage}
                [t]{0.27\textwidth}{
                    \underline{$\textbf{G}^{\textup{dbdh}}_{\genG, \mathcal{D}}(\lambda)$}:\\
                    $b \leftarrow \{0,1\}$\\
                    $\pp\leftarrow \genG(\lambda)$\\
                    $(x,y,z)  \leftarrow (\mathbb{Z}^*_p)^3$\\
                    $X:=g^x,Y=g^y,Z=g^z$\\
                    \textup{If} $b=1$: $R=e(g,g)^{xyz}$\\
                    \textup{Else}: $R\leftarrow \mathbb{G}_T$\\
                    $b' \leftarrow \mathcal{D}(\pp,X,Y,Z,R)$\\
                   	If $b'=b$,  Return $1$\\
                   Else Return $0$
                }
            \end{minipage}
            \vrule
            \hspace{0.01\textwidth}
            \begin{minipage}
                [t]{0.27\textwidth}{
                    \underline{$\textbf{G}^{\textup{cbdh}}_{\genG, \mathcal{A}}(\lambda)$}:\\
                    $b \leftarrow \{0,1\}$,
                    $\pp\leftarrow \genG(\lambda)$\\
                    $(x,y,z)  \leftarrow (\mathbb{Z}^*_p)^3$\\
                    $X:=g^x,Y=g^y,Z=g^z$\\
                    $R \leftarrow \mathcal{A}(\pp,X,Y,Z)$\\
                    If $R=e(g,g)^{xyz}$, Return 1\\
                    Else Return 0
                }
            \end{minipage}
            \vrule
            \hspace{0.01\textwidth}
            \begin{minipage}
                [t]{0.29\textwidth}{
                    \underline{$\textbf{G}^{\textup{gbdh}}_{\genG, \mathcal{A}}(\lambda)$}:\\
                    $b \leftarrow \{0,1\}$,
                    $\pp\leftarrow \genG(\lambda)$\\
                    $(x,y,z)  \leftarrow (\mathbb{Z}^*_p)^3$\\
                    $X:=g^x,Y=g^y,Z=g^z$\\
                    $R \leftarrow \mathcal{A}^{\mathcal{O}_{\textup{dbdh}}}(\pp,X,Y,Z)$\\
                     If $R=e(g,g)^{xyz}$, Return 1\\
                    Else Return 0\\

                    \underline{$\mathcal{O}_{\textup{dbdh}}(X',Y',Z',R')$}\\
                    \hfill$\sslash Z'=g^{z'}$\\
                    If $R=e(X',Y')^{z'}$, Return 1\\
                    Else Return 0
                }
            \end{minipage}
        }
        \captionsetup{justification=centering}
    \caption{Games $\textbf{\upshape{G}}^{\textup{dbdh}}_{\genG, \mathcal{D}}(\lambda)$, $\textbf{\upshape{G}}^{\textup{cbdh}}_{\genG, \mathcal{A}}(\lambda)$ and $\textbf{\upshape{G}}^{\textup{gbdh}}_{\genG, \mathcal{A}}(\lambda)$}
    \label{fig:keadlpddh}
    \end{figure}

For two random variables $X_1$ and $X_2$ over a set $S$, their statistical distance is denoted as  $\mathbf{SD}(X_1,X_2)=\frac{1}{2}\sum_{x\in S}|\Pr[X_1=x]-\Pr[X_2=x]|$. 

\begin{lemma}[\cite{DORS08}]\label{lem:leftover_hash}
	Assume that a family of hash functions $\{H_x:\{0,1\}^n\rightarrow\{0,1\}^\ell\}_{x\in X}$ is universal (i.e., for all $a\neq b\in\{0,1\}^n$, $\Pr_{x\leftarrow X}[H_x(a)=H_x(b)]=2^{\ell}$). Then, for any random variable $W$, $\mathbf{SD}((H_X(W),X), (U_{\ell},X))\leq \frac{1}{2}\sqrt{2^{-\mathbf{H}_{\infty}(W)}2^{\ell}}$, where $U_{\ell}$ denotes the uniform distribution over $\{0,1\}^\ell$. 
\end{lemma}
\section{Preliminaries: pseudorandom function, commitment, and signature}
\label{app:preliminary}

\begin{definition}[Pseudorandom function] 
    A pseudorandom function (PRF) $\prf$ is a pair of polynomial time algorithms $(\kg,\eval)$ that works as follows:
    \begin{itemize}
        \item $\kg(1^{\lambda}) \rightarrow k$: On inputting a security parameter $1^{\lambda}$, it outputs a key $k$.
        \item $\eval(k,x) \rightarrow r$: On inputting a key $k$ and a string $x \in \{0,1\}^*$, it outputs a string $r \in \{0,1\}^{\lambda}$. 
    \end{itemize}
    A PRF should satisfy the following condition. For any sufficiently large security parameter $1^{\lambda}$, and $k \leftarrow \kg(1^{\lambda})$, any truly random function $F$ whose range is the same as $\eval(k,\cdot)$, and any PPT distinguisher $\D$, it holds that \[|\Pr[\D^{\eval(k,\cdot)}(1^{\lambda}) = 1] - \Pr[\D^{F(\cdot)}(1^{\lambda}) = 1] | \leq \negl.\]
    \label{def:prf}
\end{definition}

\begin{definition}[Commitment] 
        A \emph{commitment} scheme with message space $\mathcal{M}$ contains three PPT algorithms $\commit =(\setup,\Com,\Dec)$:  
        \begin{itemize}
            \item[$\bullet$] $\setup(1^{\lambda}) \rightarrow \pp$: The setup algorithm takes the security parameter $1^{\lambda}$ as input and output a public parameter $\pp$. 
            \item[$\bullet$] $\Com(\pp,m;\rcom) \rightarrow \com$: The  commitment algorithm  takes as input the public parameter $\pp$ and $m \in \mathcal{M}$, with an inner randomized input $\rcom$, and outputs a commitment $\com$.
            \item[$\bullet$] $\Dec(\pp,\com,\rcom,m) \rightarrow b$: The decommitment algorithm  takes as input the public parameter $pp$, a commitment $\com$, a decommitment $\rcom$ and a message $m \in \mathcal{M}$, and outputs a bit $b\in\{0,1\}$ depending on whether $m$ is the committed message of $\com$. One simple method is to re-generate the commitment via $\Com$ algorithm with input $(\pp,m,\rcom)$ and check if the output equals to $\com$.
        \end{itemize}

        A commitment scheme enjoys the following properties:
        \begin{itemize}
            \item \textup{\textbf{Correctness.}} For all $m \in \mathcal{M}$ and all $\rcom\in\RS_{\Com}$, it holds that
            \begin{equation*}
                \textup{Pr}[\pp \leftarrow \setup(1^{\lambda}), \com \leftarrow \Com(\pp,m;\rcom):\Dec(\pp,\com,\rcom,m)=1]=1.
            \end{equation*}
    
            \item \textup{\textbf{Binding.}} For any  PPT adversary $\A$,   
            \begin{eqnarray*}
            \Pr\left[
                \begin{array}{l}
                \pp \leftarrow \setup(1^{\lambda})\\
                (\com,m,\rcom,m',\rcom')\leftarrow \A(\pp)    
                \end{array}
                \text{\upshape{\Large{:}}}
                \begin{array}{c}
                m\neq m'\\
                \wedge~\Dec(\pp,\com,\rcom,m)=1\\
                 \wedge~\Dec(\pp,\com,\rcom',m')=1
                \end{array}
             \right]
             \leq \negl.
          \end{eqnarray*}
            
            \item \textup{\textbf{Hiding.}} For any  PPT adversary $\A=(\A_1,\A_2)$, 
            \begin{eqnarray*}
            \left|\Pr\left[
                \begin{array}{l}
                    \pp \leftarrow \setup(1^{\lambda}), 
                    b \leftarrow \{0,1\} \\
                    (m_0,m_1,st)\leftarrow \A_1(\pp)\\
                    \com
                    \leftarrow \Com(\pp,m_b)\\
                    b'\leftarrow\A_2(\com,st)
                \end{array}
                \text{\upshape{\Large{:}}}~ b'=b
             \right]-\frac{1}{2}
             \right|\leq \negl.
          \end{eqnarray*}
          If $\A$ is unbounded and $\negl$ is fixed to be $0$, we say that the scheme $\commit$ is \emph{perfect hiding}.
        \end{itemize}
        \label{def:commit}
    \end{definition}

    \begin{definition}[Signature] 
        A signature scheme for a message space $\M$  consists of  four algorithms $\sig = (\setup,$ $\kg,\sign,\verify)$.
        \begin{itemize}
            \item[$\bullet$] $\setup(1^{\lambda}) \rightarrow \pp$: On input the security parameter $1^{\lambda}$, the setup algorithm outputs a public parameter $\pp$.
            \item[$\bullet$] $\kg(\pp) \rightarrow (pk,sk)$: On input $\pp$, the key generation algorithm  outputs a key pair  $(pk,sk)$.
            \item[$\bullet$] $\sign(\pp,sk,m) \rightarrow \sigma$: On input $\pp$,  $sk$ and a message $m \in \M$, the signing algorithm  outputs a signature $\sigma$.
            \item[$\bullet$] $\verify(\pp,pk,m,\sigma) \rightarrow b$: On input $\pp$, $pk$, $m$ and a signature $\sigma$, the verification algorithm outputs a bit $b$.
        \end{itemize}
        
        Correctness requires that for all $m \in \mathcal{M}$, it holds that:
        \begin{equation*}
        \Pr \left [
             \begin{array}{l}
                    \pp \leftarrow \setup(1^{\lambda})\\
                    (pk,sk) \leftarrow \kg(\pp)
                    \end{array}
                   ~ \textup{\Large{:}}~
                    \begin{array}{c}
                        \verify(\pp,pk,m,\sign(\pp,sk,m)) = 1 
                    \end{array}
                    \right ] = 
                    1 .
        \end{equation*}
        \label{def:sig}
        \end{definition}

        The signature should satisfy \emph{existential unforgeability}. In other words, we say $\sig$ is existentially unforgeable, if the following probability is negligible,

        \begin{equation*}
           \textbf{\textup{Adv}}_{\sig,\A}^{\textup{unforg}}(\lambda) = \Pr \left [
                 \begin{array}{l}
                        Q \leftarrow \emptyset, ~\pp \leftarrow \setup(1^{\lambda})\\
                        (pk,sk) \leftarrow \kg(\pp)\\
                        (m^*,\sigma^*) \leftarrow \A^{\orasign}(\pp,pk)
                        \end{array}
                       ~ \textup{\Large{:}}~
                        \begin{array}{c}
                            (m^*,*) \notin Q\\
                            \verify(\pp,pk,m^*,\sigma^*) = 1 
                        \end{array}
                        \right ] \leq \negl ,
            \end{equation*}
 where the oracle $\orasign$ is defined as follows:
    \begin{itemize}
        \item $\orasign(m')$: On receiving $m'$, find $(m',*)$ in $Q$. If some entry $(m',\sigma') \in Q$ is found, then output the corresponding $\sigma'$. Otherwise, compute $\sigma' \leftarrow \sign(\pp,sk,m')$, set $Q \leftarrow Q \cup \{(m',\sigma')\}$, and output $\sigma'$.
    \end{itemize}
\section{Preliminaries: A MDVS scheme in \cite{au2014strong} and BLS signature \cite{boneh2004short}}
\label{sec:concrete_scheme}

\noindent \underline{\textbf{MDVS.}} Here, we recall a MDVS scheme \cite{au2014strong} in Fig. \ref{fig:algo_mdvs_two}.

It is based on discrete logarithm setting. Essentially, the signature is a ring signature with only two party. One party is the signer and the other one is all the designated verifiers. Since it is based on discrete logarithm and the key pairs of the verifiers are in the form of $(sk=x,pk=g^x)$, we can compute a public key for a ``virtual'' role representing all the designated verifiers in this way: $pk = \prod_{i \in S} pk_i$. Finally, adopting ring signatures to generate a signature for the signer and the virtual role. There are some improvements against rogue key attack \cite{shim2008rogue,au2014strong}. Here we omit the details.

The existential unforgeability of the ring signature guarantees the existential unforgeability of the scheme \cite{au2014strong}. Off-the-record property for designated set lies in that all designated verifiers can together generate such a ring signature, and the security of anonymity of ring signature guarantees the indistinguishability from a real signature output by the signer in \cite{au2014strong}. Although consistency is not defined in \cite{au2014strong}, the scheme is consistent, since the verification of the ring signature is public. 

One of the advantages of the scheme \cite{au2014strong} is that the size of the MDVS signature is constant, no matter how many designated verifiers there are.

\begin{figure}[!h]
    \small
	\centering
		\fbox{
			\begin{minipage}
				[t]{0.85\textwidth}{
				\underline{$\setup(1^{\lambda})$}:\\
                $(\G,g,p) \leftarrow \genG(1^{\lambda})$
                \hfill$\sslash g$ is a generator of $\G$ with prime order $p$\\
                Choose a hash function $\hashf$: $\{0,1\}^* \rightarrow \Z$\\
                Return $\pp=(\G,g,p,\hashf)$\\

                \underline{$\kg(\pp)$}:\\
                Return $(sk \leftarrow \Z$, $pk:=g^{sk})$\\

                \underline{$\sign(\pp,m,\sks,\Sett=(\pkvi)_{i \in [n]})$}:\\
                $Y:=1$\\
                For $i\in [n]$: $h_i := \hashf(\pks,(\pkvi)_{i\in[n]},m,i)$, $Y=Y\cdot \pkvi^{h_i}$\\
                $(r,c_2,z_2) \leftarrow (\Z)^3$, $T_1 := g^r$, $T_2 := Y^{c_2}g^{z_2}$\\
                $c := \hashf(T_1,T_2,\pks,(\pkvi, h_i)_{i\in[n]},m,Y)$\\
                $c_1 := c - c_2$, $z_1 := r - c_1 \cdot \sks$\\
                Return $\sigma = (c_1,c_2,z_1,z_2)$\\

                \underline{$\verify(\pp,m,\pks,\sigma,\Sett)$}:\\
                    $Y:=1$\\
                    For $i\in [n]$: $h_i := \hashf(\pks,(\pkvi)_{i\in[n]},m,i)$, $Y=Y\cdot \pkvi^{h_i}$\\
                    $T_1' := \pks^{c_1}g^{z_1}$, $T_2' := Y^{c_2}g^{z_2}$\\
                    $c':=\hashf(T_1',T_2',\pks,(\pkvi, h_i)_{i\in[n]},m,Y)$\\
                    If $(c'=c_1+c_2)$: Return $b=1$\\
                    Else Return $b=0$

			 	}
			 \end{minipage}

                    
}	
	\captionsetup{justification=centering}
	\caption{Algorithms for multi-designated verifiers signature $\mdvs$ in \cite{au2014strong}}
	\label{fig:algo_mdvs_two}
\end{figure}

\noindent \underline{\textbf{BLS signature.}} BLS signature \cite{boneh2004short} is recalled in Fig. \ref{fig:algo_bls}. Here, we make a bit modification. More exactly, the original signature is $\sigma := h^{sk}$ and the verification is to check whether  $e(\sigma,g) = e(h,\pks)$ or not. Here the signature would be a hash value, i.e., $\sigma := \hashf'(e(h^{sk},g))$, and then the verification is to check whether  $\sigma = \hashf'(e(h,\pks))$.

\begin{figure}[!h]
    \small
	\centering
		\fbox{
			\begin{minipage}
				[t]{0.5\textwidth}{
				\underline{$\setup(1^{\lambda})$}:\\
                $(\G,\G_T,e,g,p) \leftarrow \genG(1^{\lambda})$\\
                Choose a hash function $\hashf:\{0,1\}^* \rightarrow \G$\\
                Choose a hash function $\hashf':\G_T \rightarrow \{0,1\}^l$\\
                Return $\pp=(\G,\G_T,e,g,p,\hashf,\hashf')$\\

                \underline{$\kg(\pp)$}:\\
                $sk \leftarrow \Z$, $pk := g^{sk}$\\
                Return $(sk,pk)$
			 	}
			 \end{minipage}
			 \hspace{0.01\textwidth}
			 \begin{minipage}
				[t]{0.41\textwidth}{   
                    \underline{$\sign(\pp,m,\sks)$}:\\
                $h \leftarrow \hashf(m)$\\
                Return $\sigma: = \hashf'(e(h^{sk},g))$\\

                \underline{$\verify(\pp,m,\pks,\sigma)$}:\\
                If $(\sigma = \hashf'(e(h,\pks)))$: Return $1$\\
                Return $0$
			 	}
			 \end{minipage}
}	
	\captionsetup{justification=centering}
	\caption{BLS signature}
	\label{fig:algo_bls}
\end{figure}

\section{Proof of Theorem \ref{thm:MDDW_required_security}}\label{sec:security_proof_append}
We provide the proofs of completeness, consistency, soundness, distortion-freeness, robustness, and off-the-record property for designated set of $\mddw$ here. 
\subsection{Proof of completeness (in Theorem \ref{thm:MDDW_required_security})}
\begin{proof}
	For  any  $i$, any prompt $\bm{p}\in\tg^*$,  any detector set $S$, any $j'\in S$, any normally generated $\pp$, $(\wpk_i,\wsk_i)$ and $\{\dpk_j,\dsk_j\}_{j\in S}$, and for any token $\widehat{\bm{t}}\leftarrow\watm(\pp,\wsk_i,\{\dpk_j\}_{j \in S},\bm{p})$, if   $\watm$  successfully embeds   a message/signature pair in $\widehat{\bm{t}}$, obviously the correctness of the underlying MDVS scheme guarantees  $\detect(\pp,$ $\wpk_i,\dsk_{j'},\{\dpk_j\}_{j \in S},\widehat{\bm{t}})=1$. The only possibility that $\watm$ fails in embedding  is that  the rejection sampling algorithm 	fails to find the next batch of tokens whose hash is consistent with the target bit (i.e., line 10-11 in Algorithm \ref{alg:Watm}). Note that by Assumption \ref{assumption}, any $\ell$ consecutive tokens generated by $\genmodel_{\ell}$  contains $\alpha$ bits of entropy. So no one can predict the input to the random oracle in line 10  of Algorithm \ref{alg:Watm} except with probability $2^{-\alpha}$, which is negligible. In other words, with overwhelming probability, the random oracle in line 10  of Algorithm \ref{alg:Watm} will return a uniformly and independently sampled bit. 
	Hence, with overwhelming probability, we derive that for each sampling attempt  of the next batch of tokens, it will succeed in finding a consistent hash value  with probability $\frac{1}{2}$.  After $\textsf{poly}(\lambda)$ 	attempts, the rejection sampling will find the next batch of tokens with probability $1-2^{-\textsf{poly}(\lambda)}$, i.e., the probability that $\watm$ fails is negligible. \qed
\end{proof}
\subsection{Proof of consistency (in Theorem \ref{thm:MDDW_required_security})}
\begin{proof}
	For any PPT adversary $\A$ attacking the consistency of $\mddw$, we show a PPT adversary $\B$ attacking the  consistency of the underlying $\mdvs$ as follows.

	Upon receiving $\pp$ (note that the $\pp$ generated by $\setup$ of $\mddw$ is the same as that generated by the underlying $\mdvs.\setup$), $\B$ sends $\pp$ to $\A$. $\B$ maintains three local arrays, $L_{\text{ro},1}$,  $L_{\text{ro},2}$ and  $L_{\text{ro},3}$, to keep track of $\A$'s random oracle queries.  Then, $\B$  answers $\A$'s oracle queries as below: 
	\begin{myitemize}
		\item \underline{$\mathcal{O}_{RO,\kappa}(str)~ (\kappa\in\{1,2,3\})$:} If there is some $(str,y)\in L_{\text{ro},\kappa}$,  $\mathcal{B}$ returns $y$; otherwise,  $\mathcal{B}$  samples $y\leftarrow\M$, adds $(str,y)$ to $L_{\text{ro},\kappa}$, and returns $y$. 
		
		
		\item\underline{$\mathcal{O}_{WK}(i)$:}  $\B$ queries its own oracle $\mathcal{O}_{SK}(i)$ to obtain $(\spk_i,\ssk_i)$, and returns $(\wpk_i,\wsk_i)=(\spk_i,\ssk_i)$ to $\A$. 
		
		\item\underline{$\mathcal{O}_{DK}(j)$:}  $\B$ queries its own oracle $\mathcal{O}_{VK}(j)$ to obtain $(\vpk_j,\vsk_j)$, and returns $(\dpk_j,\dsk_j)=(\vpk_j,\vsk_j)$ to $\A$. 
		
		\item\underline{$\mathcal{O}_{WPK}(i)$:}  $\B$ queries its own oracle $\mathcal{O}_{SPK}(i)$ to obtain $\spk_i$, and returns $\wpk_i=\spk_i$ to $\A$. 
		
		\item\underline{$\mathcal{O}_{DPK}(j)$:}  $\B$ queries its own oracle $\mathcal{O}_{VPK}(j)$ to obtain $\vpk_j$, and returns $\dpk_j=\vpk_j$ to $\A$. 
		
		\item\underline{$\mathcal{O}_{W}(i,S,\bm{p})$:}  $\B$ 
		runs   $\textsf{SimWatMar}$ (with the help of its  signing oracle $\mathcal{O}_S$)   as shown in  Algorithm \ref{alg:SimWatm}. More specifically, Algorithm \ref{alg:SimWatm} is the same as Algorithm \ref{alg:Watm}, except that   the original  hash functions $H_1$, $H_2$ and $H_3$  (in line 4, line 5 and line 10, respectively) in  Algorithm \ref{alg:Watm} are replaced with random oracles, and the original $\sign$ algorithm  (in line 4) in  Algorithm \ref{alg:Watm} is replaced with $\B$'s signing oracle $\mathcal{O}_S$. Finally, $\B$ 
		returns the output $\bm{t}$ of $\textsf{SimWatMar}$  to $\A$.

		\item\underline{$\mathcal{O}_{D}(i,j',S,\bm{t})$:}  $\B$ 
		runs  $\textsf{SimDetect}$ (with the help of its verification oracle $\mathcal{O}_V$) as shown in  Algorithm \ref{alg:SimDet}. More specifically, Algorithm \ref{alg:SimDet} is the same as Algorithm \ref{alg:Det}, except that   the original  hash functions $H_1$, $H_2$ and $H_3$ (in line 3, line 7 and line 6, respectively) in  Algorithm \ref{alg:Det} are  replaced with  random oracles, and the original $\verify$ algorithm (in line 8) in  Algorithm \ref{alg:Det} is replaced with $\B$'s verification oracle $\mathcal{O}_V$. 
		Finally, $\B$ returns the output of $\textsf{SimDetect}$  to $\A$. 
	\end{myitemize}
	
	Upon receiving $(i^*,S^*,\bm{t}^*)$ from $\A$, $\B$  checks whether there are $j_0^*,j_1^*\in S^*$ such that $\textsf{SimDetect}^{\mathcal{O}_V,\{\mathcal{O}_{RO,\kappa}\}_{\kappa\in\{1,2,3\}}}(\pp,i^*,j_{\beta}^*,\{\dpk_{j}\}_{j\in S^*},\bm{t}^*)=\beta$ for all $\beta\in\{0,1\}$.
	\begin{itemize}
		\item If not, $\B$ returns a random tuple $(i^*_{\textsf{rand}},S^*_{\textsf{rand}},m^*_{\textsf{rand}},$ $\sigma^*_{\textsf{rand}})$ as its final output.
		\item Otherwise, $\B$ runs the algorithm $\textsf{SimDetect}$ to find the tuple $(\widetilde{m},\hat{\sigma})$ such that 
		\begin{equation*}
			\small
			\hspace{-0.4cm}
			\textsf{SimDetect}^{\mathcal{O}_V,\{\mathcal{O}_{RO,\kappa}\}_{\kappa\in\{1,2,3\}}}(\pp,i^*,j_{1}^*,\{\dpk_{j}\}_{j\in S^*},\bm{t}^*)=1,
		\end{equation*}	
		(i.e., $\mathcal{O}_V(i^*,j_{1}^*,S^*,$ $\widetilde{m},\hat{\bm{\sigma}})=1$ in line 8 of Algorithm \ref{alg:SimDet}). Then, $\B$ returns a random tuple $(i^*,S^*,\widetilde{m},\hat{\bm{\sigma}}$ as its final output. 
	\end{itemize}
	
	That is the construction of adversary $\B$. 
	
	\begin{algorithm}[h]\small
		\caption{$\textsf{SimWatMar}^{\mathcal{O}_S,\{\mathcal{O}_{RO,\kappa}\}_{\kappa\in\{1,2,3\}}}(\pp,i,S,\bm{p})$\label{alg:SimWatm}}

		\begin{algorithmic}[1]
			\State $\bm{t}\leftarrow\epsilon$
			
			\While{$|\bm{t}|+\ell+\ell\cdot\textsf{len}_{\text{sig}}<n$}
			\State $\bm{t}\leftarrow\bm{t}||\genmodel_\ell(\bm{p},\bm{t})$
			
			\State $\hat{\bm{\sigma}}\leftarrow\underline{\mathcal{O}_{S}(i,S, \underline{\mathcal{O}_{RO,1}}(\bm{t}[-\ell:]))}$
			
			\State $\bm{\sigma}\leftarrow \hat{\bm{\sigma}} \oplus \underline{\mathcal{O}_{RO,2}}(\bm{t}[-\ell:])$
			
			\State $\bm{\sigma}_{\textup{prev}}\leftarrow\epsilon$, 
			$~\bm{m}\leftarrow\epsilon$
			
			\While{$\bm{\sigma}\neq\epsilon$}
			\State $\sigma_{\text{bit}}\leftarrow \bm{\sigma}[1]$, $~\bm{\sigma}\leftarrow\bm{\sigma}[2:]$
			
			\State $\bm{x}\leftarrow\genmodel_\ell(\bm{p},\bm{t})$
			
			\While{$\underline{\mathcal{O}_{RO,3}}(\bm{m}||\bm{x}||\bm{\sigma}_{\textup{prev}})\neq\sigma_{\text{bit}}$}
			\State $\bm{x}\leftarrow\genmodel_\ell(\bm{p},\bm{t})$
			\EndWhile
			\State $\bm{m}\leftarrow\bm{m}||\bm{x}$			
			\State $\bm{t}\leftarrow\bm{t}||\bm{x}$			
			\State $\bm{\sigma}_{\textup{prev}}\leftarrow\bm{\sigma}_{\textup{prev}}||\sigma_{\text{bit}}$
			\EndWhile
			\EndWhile
			\If{$|\bm{t}|<n$}
			\State $\bm{t}\leftarrow\bm{t}||\genmodel_{n-|\bm{t}|}(\bm{p},\bm{t})$
			\EndIf
			
			\State \Return $\bm{t}$
		\end{algorithmic}
	\end{algorithm}
	\begin{algorithm}[h]\small
		\caption{$\textsf{SimDetect}^{\mathcal{O}_{V},\{\mathcal{O}_{RO,\kappa}\}_{\kappa\in\{1,2,3\}}}(\pp,i,j',S,\bm{t})$\label{alg:SimDet}}
		\begin{algorithmic}[1]
			\For{$\mu=1,2,\cdots, n-(\ell+\ell\cdot\textsf{len}_{\text{sig}})$}
			\State $\widetilde{m}\leftarrow \underline{\mathcal{O}_{RO,1}}(\bm{t}[\mu:\mu+\ell-1])$
			
			\State $\bm{\sigma}\leftarrow\epsilon$, 		
			$~\bm{m}\leftarrow\epsilon$
			
			\For{$\varphi\in[\textsf{len}_{\text{sig}}]$}
			\State $\bm{m}\leftarrow\bm{m}||\bm{t}[\mu+\varphi\cdot\ell:\mu+\varphi\cdot\ell+\ell-1]$
			
			\State $\bm{\sigma}\leftarrow\bm{\sigma}||\underline{\mathcal{O}_{RO,3}}(\bm{m}||\bm{\sigma})$
			\EndFor
			\State $\hat{\bm{\sigma}} \leftarrow\bm{\sigma}\oplus \underline{\mathcal{O}_{RO,2}}(\bm{t}[\mu:\mu+\ell-1])$
			
			\If{$\underline{\mathcal{O}_V(i,j',S,\widetilde{m},\hat{\bm{\sigma}})}=1$}
			\State \Return 1
			\EndIf
			\EndFor
			\State \Return 0
		\end{algorithmic}
	\end{algorithm}
	
	It is easy to see that $\B$ perfectly simulates game  $\gamecons$ for $\A$, and if $\A$ wins $\gamecons$, then  $\B$ also wins $\textup{\textbf{G}}^{\textup{cons}}_{\mdvs,\B}(\lambda)$. So we obtain 
	\[\textup{\textbf{Adv}}^{\textup{cons}}_{\mdvs,\B}(\lambda)\geq\advcons.\]
	 \qed
\end{proof}

\subsection{Proof of soundness  (in Theorem \ref{thm:MDDW_required_security})} 
\begin{proof}
	
	For any PPT adversary $\A$ attacking the soundness of $\mddw$, we show a PPT adversary $\B$ attacking the existential unforgeability of the underlying $\mdvs$ as follows. 
	
	Upon receiving $\pp$ (note that the $\pp$ generated by $\setup$ of $\mddw$ is the same as that generated by the underlying $\mdvs.\setup$), $\B$ sends $\pp$ to $\A$. $\B$ maintains three local arrays, $L_{\text{ro},1}$,  $L_{\text{ro},2}$ and  $L_{\text{ro},3}$, to keep track of $\A$'s random oracle queries.  Then, $\B$  answers $\A$'s oracle queries as below: 
	\begin{myitemize}
		\item \underline{$\mathcal{O}_{RO,\kappa}(str)~ (\kappa\in\{1,2,3\})$:} If there is some $(str,y)\in L_{\text{ro},\kappa}$,  $\mathcal{B}$ returns $y$; otherwise,  $\mathcal{B}$  samples $y\leftarrow\M$, adds $(str,y)$ to $L_{\text{ro},\kappa}$, and returns $y$. 
		
		
		\item\underline{$\mathcal{O}_{WK}(i)$:}  $\B$ queries its own oracle $\mathcal{O}_{SK}(i)$ to obtain $(\spk_i,\ssk_i)$, and returns $(\wpk_i,\wsk_i)=(\spk_i,\ssk_i)$ to $\A$. 
		
		\item\underline{$\mathcal{O}_{DK}(j)$:}  $\B$ queries its own oracle $\mathcal{O}_{VK}(j)$ to obtain $(\vpk_j,\vsk_j)$, and returns $(\dpk_j,\dsk_j)=(\vpk_j,\vsk_j)$ to $\A$. 
		
		\item\underline{$\mathcal{O}_{WPK}(i)$:}  $\B$ queries its own oracle $\mathcal{O}_{SPK}(i)$ to obtain $\spk_i$, and returns $\wpk_i=\spk_i$ to $\A$. 
		
		\item\underline{$\mathcal{O}_{DPK}(j)$:}  $\B$ queries its own oracle $\mathcal{O}_{VPK}(j)$ to obtain $\vpk_j$, and returns $\dpk_j=\vpk_j$ to $\A$. 
		
		\item\underline{$\mathcal{O}_{W}(i,S,\bm{p})$:}  $\B$ 
		runs   $\textsf{SimWatMar}$ (with the help of its  signing oracle $\mathcal{O}_S$)   as shown in  Algorithm \ref{alg:SimWatm}. More specifically, Algorithm \ref{alg:SimWatm} is the same as Algorithm \ref{alg:Watm}, except that   the original  hash functions $H_1$, $H_2$ and $H_3$  (in line 4, line 5 and line 10, respectively) in  Algorithm \ref{alg:Watm} are replaced with random oracles, and the original $\sign$ algorithm  (in line 4) in  Algorithm \ref{alg:Watm} is replaced with $\B$'s signing oracle $\mathcal{O}_S$. Finally, $\B$ 
		returns the output $\bm{t}$ of $\textsf{SimWatMar}$  to $\A$.

		\item\underline{$\mathcal{O}_{D}(i,j',S,\bm{t})$:}  $\B$ 
		runs  $\textsf{SimDetect}$ (with the help of its verification oracle $\mathcal{O}_V$) as shown in  Algorithm \ref{alg:SimDet}. More specifically, Algorithm \ref{alg:SimDet} is the same as Algorithm \ref{alg:Det}, except that   the original  hash functions $H_1$, $H_2$ and $H_3$ (in line 3, line 7 and line 6, respectively) in  Algorithm \ref{alg:Det} are  replaced with  random oracles, and the original $\verify$ algorithm (in line 8) in  Algorithm \ref{alg:Det} is replaced with $\B$'s verification oracle $\mathcal{O}_V$. 
		Finally, $\B$ returns the output of $\textsf{SimDetect}$  to $\A$. 
	\end{myitemize}
	
	Receiving $(i^*,S^*,\bm{t}^*)$ from $\A$, $\B$ firstly  sets $\widehat{i^*}=i^*$ and $\widehat{S^*}=S^*$. Then, $\B$ checks whether there is some $j'\in S^*$ such that (i) $\textsf{SimDetect}^{\mathcal{O}_V,\{\mathcal{O}_{RO,\kappa}\}_{\kappa\in\{1,2,3\}}}(\pp,$ $i^*,j',(\dpk_{j})_{i\in S^*},\bm{t}^*)=1$, and (ii) $j'$ has never been queried to $\mathcal{O}_{DK}$ by $\A$. 
	\begin{myitemize}
		\item[-] If so, $\B$ extracts $(\widetilde{m},\hat{\bm{\sigma}})$ according to the steps in Algorithm \ref{alg:SimDet} (from line 1 to line 7), and sets $(\widehat{m^*},\widehat{\sigma^*})=(\widetilde{m},\hat{\bm{\sigma}})$. 
		\item[-] Otherwise, $\B$ uniformly samples  $(\widehat{m^*},\widehat{\sigma^*})$ from the message space and the signature space. 
	\end{myitemize} 
	Finally, $\B$ returns $(\widehat{i^*},\widehat{S^*},\widehat{m^*},\widehat{\sigma^*})$ as its final output.

	That is the construction of adversary $\B$. Now we   analyze its advantage. 
	
	It is easy to see that $\B$ perfectly simulates game  $\gameunf$ for $\A$. Note that when $\A$ wins  $\gameunf$, $\A$'s final output $(i^*,S^*,\bm{t}^*)$ satisfies that there is some $j'\in S^*$, such that (i) $\A$ has never queried the oracle $\orawk$ (resp., $\oradk$) on  $i^*$ (resp., $j'$); (ii) $\detect(\pp,\wpk_{i^*},\dsk_{j'},\{\dpk_j\}_{j\in S^*},\bm{t}^*)=1$; (iii)  $\textsf{NOLap}_{\ell}(\bm{t}^*,\bm{t}_1,\bm{t}_2,\cdots)=1$, where $\bm{t}_1,\bm{t}_2,\cdots$ are the watermarked texts that $\A$ receives via querying $\orawatm$. 	
	The fact (ii) implies that there must be some $\mu'\in\{0,1,\cdots,n-(\ell+\textsf{len}_{\text{sig}})\}$, such that for $\widetilde{m}=\mathcal{O}_{RO,1}(\bm{t}^*[\mu':\mu'+\ell-1])$ and for $\hat{\bm{\sigma}}$ generated according to line 2-7 in Algorithm \ref{alg:Det} (when $\mu=\mu'$), $\verify(\pp,\spk_{i^*},\vsk_{j'},\{\vpk_j\}_{j\in S^*},\widetilde{m},\hat{\bm{\sigma}})=1$. Moreover, the fact (iii)  guarantees that during the process of answering $\A$'s $\mathcal{O}_W$-oracle queries,  the random oracle $\mathcal{O}_{RO,1}$ has never been programmed on $\bm{t}^*[\mu':\mu'+\ell-1]$.  So the property of random oracle guarantees that $\widetilde{m}=\mathcal{O}_{RO,1}(\bm{t}^*[\mu':\mu'+\ell-1])$ has never been submitted as a signing query by $\B$. 
	
	Hence, $\B$ wins  $\textup{\textbf{G}}^{\textup{unforg}}_{\mdvs,\B}(\lambda)$  if  $\A$ wins $\gameunf$, i.e., 
	\begin{eqnarray*}
		\textup{\textbf{Adv}}^{\textup{unforg}}_{\mdvs,\B}(\lambda)\geq \advunf.
	\end{eqnarray*}\qed
\end{proof}

\subsection{Proof of distortion-freeness (in Theorem \ref{thm:MDDW_required_security})}
\begin{proof}
	We show the proof with a sequence of games. 
	
	\noindent\textbf{Game} $\textbf{G}_0$: This game is the original $\gamedistor$ when $b$ is fixed to be $0$. Note that in this game, when $\A$ queries the oracle $\mathcal{O}^{(0)}_{M\text{-}chl}$ on $(i,S,\bm{p})$, the challenger runs Algorithm \ref{alg:Watm} with input $(\pp,\wsk_i,$ $\{\dpk_j\}_{j \in S},\bm{p})$ and returns the output to $\A$.
	
	
\noindent\textbf{Game} $\textbf{G}_1$: This game is the same as  $\textbf{G}_0$, except that when $\A$ queries the oracle $\mathcal{O}^{(0)}_{M\text{-}chl}$ on $(i,S,\bm{p})$, the challenger runs Algorithm \ref{alg:Watm} with the following modifications: for any  $\bm{t}\leftarrow\bm{t}||\genmodel_\ell(\bm{p},\bm{t})$ in line 3 of Algorithm \ref{alg:Watm}, if $\bm{t}[-\ell:]$ has been queried to the random oracle $\mathcal{O}_{RO,2}$ by $\A$, then the challenger aborts the game and returns $0$ as the final output.
	
	\noindent\textbf{Game} $\textbf{G}_2$: This game is the original $\gamedistor$ when $b$ is fixed to be $1$.
	
	We present the following two lemmas with postponed proofs.
	\begin{lemma}\label{lem:MDDW_distor_1}
	$|\Pr[\mathbf{G}_1=1]-\Pr[\mathbf{G}_0=1]|\leq \negl$.
\end{lemma}
	\begin{lemma}\label{lem:MDDW_distor_2}
		$|\Pr[\mathbf{G}_2=1]-\Pr[\mathbf{G}_1=1]|\leq \negl$.
	\end{lemma}
	
	Combining these two lemmas, we derive that for any PPT adversary $\A$, 
	\begin{eqnarray*}
		\advdistor\leq\negl. 
	\end{eqnarray*}
	
	Hence, what remains is to prove the above two lemmas.
	
	\begin{proof}[of Lemma \ref{lem:MDDW_distor_1}]
	Note that $\genmodel_\ell$ generates tokens with min-entropy at least $\alpha$. So for each $\bm{t}\leftarrow\bm{t}||\genmodel_\ell(\bm{p},\bm{t})$ in line 3 of Algorithm \ref{alg:Watm}, the probability that  $\A$ has queried  the random oracle $\mathcal{O}_{RO,2}$ on $\bm{t}[-\ell:]$ is at most $2^{-\alpha}$. Within Algorithm \ref{alg:Watm}, the loop spanning from line 2 to line 12 will iterate at most $\frac{n}{\ell}$ times.	Hence, 
	\begin{align*}
		|\Pr[\mathbf{G}_2=1]-\Pr[\mathbf{G}_1=1]|\leq \frac{n}{\ell}\cdot2^{-\alpha}\leq\negl.  
	\end{align*}
	 \qed
\end{proof}
	\begin{proof}[of Lemma \ref{lem:MDDW_distor_2}]
		%
	Compared with the oracle $\mathcal{O}^{(0)}_{M\text{-}chl}$ in $\textbf{G}_1$, the oracle  $\mathcal{O}^{(1)}_{M\text{-}chl}$ in $\textbf{G}_2$ can be seen as that ignoring 
	the process of checking whether $H_3(\bm{m}||\bm{x}||\bm{\sigma}_{\text{prev}})=\sigma_{\text{bit}}$ (i.e., line 10 in Algorithm \ref{alg:Watm}) in the oracle $\mathcal{O}^{(0)}_{M\text{-}chl}$ in $\textbf{G}_1$, where $H_3$ is modeled as the random oracle  $\mathcal{O}_{RO,3}$. 
		
		Note that in    $\textbf{G}_1$, to answer $\A$'s query $(i,S,\bm{p})$ (to  $\mathcal{O}_{M\text{-}chl}^{(0)}$), while $|\bm{t}|+\ell+\textsf{len}_{\text{sig}}$ (i.e., lines 2-12 in Algorithm \ref{alg:Watm}), the challenger firstly generates $\ell$ tokens with $\genmodel_\ell$ (i.e., line 3);  then,  it samples the next $\textsf{len}_{\text{sig}}\cdot\ell$ tokens satisfying that for each  $i\in [\textsf{len}_{\text{sig}}]$,   the next $i$-{th} tuple of tokens (with length $\ell$) are generated with $\genmodel_\ell$ (i.e., line 9) conditioned on the hash to be the $i$-th bit of $\bm{\sigma}$. Since each bit of $\bm{\sigma}$ is uniformly and independently distributed over $\{0,1\}$, and $\genmodel_\ell$ generate tokens with min-entropy at least $\alpha$,  according to Lemma \ref{lem:leftover_hash}, for each of  $\A$'s query $(i,S,\bm{p})$, the statistical distance of  the answer from   $\mathcal{O}_{M\text{-}chl}^{(0)}$ in $\textbf{G}_1$ and that from  $\mathcal{O}_{M\text{-}chl}^{(1)}$ in $\textbf{G}_2$ is $\textsf{len}_{\text{sig}}\cdot 2^{-\frac{\alpha+1}{2}}$, which is negligible. 
		
		The total number of $\A$'s oracle queries  is polynomial. So we actually derive $\mathbf{SD}(\textbf{G}_2, \textbf{G}_1)\leq\negl$.  \qed \end{proof}
	\qed
\end{proof}
\subsection{Proof of robustness (in Theorem \ref{thm:MDDW_required_security})}
\begin{proof}
	According to the descriptions of Algorithm \ref{alg:Watm} and Algorithm \ref{alg:Det}, it is easy to see that for every $(2\textsf{len}_{\text{sig}}+2)\ell$ consecutive tokens $\bm{t}$  generated from Algorithm \ref{alg:Watm}, at least a pair of message and signature  is embedded in $\bm{t}$.  
	Hence, Algorithm \ref{alg:Det} will extract the  message/signature pair from $\bm{t}$, leading to a successful detection output. \qed
\end{proof}

\subsection{Proof of off-the-record property for designated set (in Theorem \ref{thm:MDDW_required_security})}
\begin{proof}
Since $\mdvs$ is off-the-record for designated set, there is a PPT algorithm $\forgeds_{\mdvs}$, taking $(\pp,\spk_i,\{\vsk_j\}_{j\in S},m)$   as input and outputting $\sigma$, such that for all PPT adversary $\A'$,  \[\textup{\textbf{Adv}}^{\textup{otr-ds}}_{\mdvs,\A',\forgeds_{\mdvs}}(\lambda)=|\Pr[\textup{\textbf{G}}^{\textup{otr-ds}}_{\mdvs,\A',\forgeds_{\mdvs}}(\lambda)=1]-\frac{1}{2}| \leq \negl. \]

Now we construct a PPT algorithm $\forgeds$ for $\mddw$ as shown in Fig. \ref{alg:Forgeds}. 
\begin{algorithm}[h]
	\small
	\caption{$\forgeds(\pp,\wpk_i,\{\dsk_j\}_{j\in S},\bm{p})$ \label{alg:Forgeds} }
	
	\begin{algorithmic}[1]
		\State $\spk_i\leftarrow \wpk_i$
		\State $\{\vsk_j\}_{j \in S}\leftarrow\{\dsk_j\}_{j \in S}$
		
		\State $\bm{t}\leftarrow\epsilon$
		
		\While{$|\bm{t}|+\ell+\ell\cdot\textsf{len}_{\text{sig}}<n$}
		
		\State $\bm{t}\leftarrow\bm{t}||\genmodel_\ell(\bm{p},\bm{t})$
		
		\State $\hat{\bm{\sigma}}\leftarrow\forgeds_{\mdvs}(\pp,\spk_i,\{\vsk_j\}_{j\in S}, H_1(\bm{t}[-\ell:]))$
		
		\State $\bm{\sigma}\leftarrow \hat{\bm{\sigma}} \oplus H_2(\bm{t}[-\ell:])$
		
		\State $\bm{\sigma}_{\textup{prev}}\leftarrow\epsilon$, 		
		$~\bm{m}\leftarrow\epsilon$
		
		\While{$\bm{\sigma}\neq\epsilon$}		
		
		\State $\sigma_{\text{bit}}\leftarrow \bm{\sigma}[1]$, $\bm{\sigma}\leftarrow\bm{\sigma}[2:]$
		
		\State $\bm{x}\leftarrow\genmodel_\ell(\bm{p},\bm{t})$
		
		\While{$H_3(\bm{m}||\bm{x}||\bm{\sigma}_{\textup{prev}})\neq\sigma_{\text{bit}}$}
		\State $\bm{x}\leftarrow\genmodel_\ell(\bm{p},\bm{t})$
		\EndWhile
		\State $\bm{m}\leftarrow\bm{m}||\bm{x}$, $~\bm{t}\leftarrow\bm{t}||\bm{x}$
		
		\State $\bm{\sigma}_{\textup{prev}}\leftarrow\bm{\sigma}_{\textup{prev}}||\sigma_{\text{bit}}$
		\EndWhile
		\EndWhile
		\If{$|\bm{t}|<n$}
		\State $\bm{t}\leftarrow\bm{t}||\genmodel_{n-|\bm{t}|}(\bm{p},\bm{t})$
		\EndIf
		\State \Return $\bm{t}$
	\end{algorithmic}
\end{algorithm}

For any PPT adversary $\A$ attacking the off-the-record property (for designated set) of $\mddw$, we show a PPT adversary $\A'$ attacking the off-the-record property (for designated set) of $\mdvs$ as follows.

Receiving $\pp$ (note that the $\pp$ generated by $\setup$ of $\mddw$ is the same as that generated by the underlying $\mdvs.\setup$), $\A'$  sends $\pp$ to $\A$. $\A'$ maintains three local arrays, $L_{\text{ro},1}$,  $L_{\text{ro},2}$ and  $L_{\text{ro},3}$, to keep track of $\A$'s random oracle queries.  Then, $\A'$  answers $\A$'s oracle queries as below: 
\begin{myitemize}
	\item \underline{$\mathcal{O}_{RO,\kappa}(str)~ (\kappa\in\{1,2,3\})$:} If there is some $(str,y)\in L_{\text{ro},\kappa}$,  $\A'$ returns $y$; otherwise,  $\A'$ samples $y\leftarrow\M$, adds $(str,y)$ to $L_{\text{ro},\kappa}$, and returns $y$. 
	
	
	\item\underline{$\orawk(i)$:}  $\A'$ queries its own oracle $\orask(i)$ to obtain $(\spk_i,\ssk_i)$, and returns $(\wpk_i,\wsk_i)=(\spk_i,\ssk_i)$ to $\A$. 
	
	\item\underline{$\oradk(j)$:}  $\A'$ queries its own oracle $\oravk(j)$ to obtain $(\vpk_j,\vsk_j)$, and returns $(\dpk_j,\dsk_j)=(\vpk_j,\vsk_j)$ to $\A$. 
	
	\item\underline{$\orawpk(i)$:}  $\A'$ queries its own oracle $\oraspk(i)$ to obtain $\spk_i$, and returns $\wpk_i=\spk_i$ to $\A$. 
	
	\item\underline{$\oradpk(j)$:}  $\A'$ queries its own oracle $\oravpk(j)$ to obtain $\vpk_j$, and returns $\dpk_j=\vpk_j$ to $\A$. 
	
\item\underline{$\mathcal{O}_{\text{\emph{otr-chl}}}^{(b)}(i,S,\bm{p})$:}  $\A'$ 
runs   $\textsf{SimOtr-Chl}_{\text{DS}}$ (with the help of its  challenge oracle $\mathcal{O}_{\text{\emph{otr-chl}}}^{(b)}$ for $\mdvs$, denoted as $\mathcal{O}_{\text{\emph{otr-chl}}}^{(b),\mdvs}$)   as shown in  Algorithm \ref{alg:SimOtrChl}. In particular,  Algorithm \ref{alg:SimOtrChl} is the same as Algorithm \ref{alg:SimWatm}, except that the signing oracle  $\orasign$ (in line 4) in  Algorithm \ref{alg:SimWatm} is replaced with $\A'$'s challenge  oracle  $\mathcal{O}_{\text{\emph{otr-chl}}}^{(b),\mdvs}$. Finally, $\A'$ 
returns the output $\bm{t}$ of $\textsf{SimOtr-Chl}_{\text{DS}}$  to $\A$.

\item\underline{$\oradet(i,j',S,\bm{t})$:}  $\A'$ 
runs  $\textsf{SimDetect}$ (with the help of its verification oracle $\oraver$) as shown in  Algorithm \ref{alg:SimDet}. More specifically, Algorithm \ref{alg:SimDet} is the same as Algorithm \ref{alg:Det}, except that   the original  hash functions $H_1$, $H_2$ and $H_3$ (in line 3, line 7 and line 6, respectively) in  Algorithm \ref{alg:Det} are  replaced with  random oracles, and the original $\verify$ algorithm (in line 8) in  Algorithm \ref{alg:Det} is replaced with $\A'$'s verification oracle $\oraver$. 
Finally, $\A'$ returns the output of $\textsf{SimDetect}$  to $\A$. 
\end{myitemize}

Finally, $\A'$ returns $\A$'s  output $b'$ as its own final output. 

	\begin{algorithm}
	\small
	\caption{$\textsf{SimOtr-Chl}_{\text{DS}}^{\mathcal{O}_{\text{\emph{otr-chl}}}^{(b),\mdvs},\{\mathcal{O}_{RO,\kappa}\}_{\kappa\in\{1,2,3\}}}(\pp,i,S,\bm{p})$\label{alg:SimOtrChl}}

	\begin{algorithmic}[1]
		\State $\bm{t}\leftarrow\epsilon$
		
		\While{$|\bm{t}|+\ell+\ell\cdot\textsf{len}_{\text{sig}}<n$}
		\State $\bm{t}\leftarrow\bm{t}||\genmodel_\ell(\bm{p},\bm{t})$
		
		\State $\hat{\bm{\sigma}}\leftarrow\mathcal{O}_{\text{\emph{otr-chl}}}^{(b),\mdvs}(i,S, \mathcal{O}_{RO,1}(\bm{t}[-\ell:]))$
		
		\State $\bm{\sigma}\leftarrow \hat{\bm{\sigma}} \oplus \mathcal{O}_{RO,2}(\bm{t}[-\ell:])$
		
		\State $\bm{\sigma}_{\textup{prev}}\leftarrow\epsilon$, 
		$~\bm{m}\leftarrow\epsilon$
		
		\While{$\bm{\sigma}\neq\epsilon$}
		\State $\sigma_{\text{bit}}\leftarrow \bm{\sigma}[1]$, $~\bm{\sigma}\leftarrow\bm{\sigma}[2:]$
		
		\State $\bm{x}\leftarrow\genmodel_\ell(\bm{p},\bm{t})$
		
		\While{$\mathcal{O}_{RO,3}(\bm{m}||\bm{x}||\bm{\sigma}_{\textup{prev}})\neq\sigma_{\text{bit}}$}
		\State $\bm{x}\leftarrow\genmodel_\ell(\bm{p},\bm{t})$
		\EndWhile
		\State $\bm{m}\leftarrow\bm{m}||\bm{x}$ 			
		\State $\bm{t}\leftarrow\bm{t}||\bm{x}$ 			
		\State $\bm{\sigma}_{\textup{prev}}\leftarrow\bm{\sigma}_{\textup{prev}}||\sigma_{\text{bit}}$
		\EndWhile
		\EndWhile
		\If{$|\bm{t}|<n$}
		\State $\bm{t}\leftarrow\bm{t}||\genmodel_{n-|\bm{t}|}(\bm{p},\bm{t})$
		\EndIf
		
		\State \Return $\bm{t}$
	\end{algorithmic}
\end{algorithm}

That is the construction of $\A'$. 

It is obvious that $\A'$ perfectly simulates game $\gameotrdsmddw$ for $\A$, and   if $\A$ wins $\gameotrdsmddw$, then  $\A'$ also wins $\textup{\textbf{G}}^{\textup{otr-ds}}_{\mdvs,\A',\forgeds_{\mdvs}}(\lambda)$. So we derive 
\[\textup{\textbf{Adv}}^{\textup{otr-ds}}_{\mdvs,\A',\forgeds_{\mdvs}}(\lambda)\geq\advotrdsmddw.\]\qed
\end{proof}

\section{Proof of Theorem \ref{thm:OTR_as}}\label{sec:security_proof_append_otr_as}
\begin{proof}
	Since $\mdvs$ is off-the-record for  any subset, there is a PPT algorithm $\forgeas_{\mdvs}$, taking $(\pp,\spk_i,\{\vpk_j\}_{j\in S},\{\vsk_j\}_{j\in \scor},m)$ (where $\scor\subset S$)  as input and outputting $\sigma$, such that for all PPT adversary $\A'$,  \[\textup{\textbf{Adv}}^{\textup{otr-as}}_{\mdvs,\A',\forgeas_{\mdvs}}(\lambda)=|\Pr[\textup{\textbf{G}}^{\textup{otr-as}}_{\mdvs,\A',\forgeas_{\mdvs}}(\lambda)=1]-\frac{1}{2}| \leq \negl. \]
	
	Now we construct a PPT algorithm $\forgeas$ for $\mddw$ as shown in Fig. \ref{alg:Forgeas}.
	
	\begin{algorithm}
	\small
	\caption{$\forgeas(\pp,\wpk_i,\{\dpk_j\}_{j\in S},\{\dsk_j\}_{j\in \scor},\bm{p})$ \label{alg:Forgeas} }
	\begin{algorithmic}[1]
		\State $\spk_i\leftarrow \wpk_i$
		\State $\{\vpk_j\}_{j \in S}\leftarrow\{\dpk_j\}_{j \in S}$
		\State $\{\vsk_j\}_{j \in \scor}\leftarrow\{\dsk_j\}_{j \in \scor}$	
		\State $\bm{t}\leftarrow\epsilon$
		
		\While{$|\bm{t}|+\ell+\ell\cdot\textsf{len}_{\text{sig}}<n$}
		
		\State $\bm{t}\leftarrow\bm{t}||\genmodel_\ell(\bm{p},\bm{t})$
		
		\State $\hat{\bm{\sigma}}\leftarrow\forgeas_{\mdvs}(\pp,\spk_i,\{\vpk_j\}_{j\in S}, \{\vsk_j\}_{j \in \scor},H_1(\bm{t}[-\ell:]))$
		
		
		\State $\bm{\sigma}\leftarrow \hat{\bm{\sigma}} \oplus H_2(\bm{t}[-\ell:])$
		
		\State $\bm{\sigma}_{\textup{prev}}\leftarrow\epsilon$, 		
		$~\bm{m}\leftarrow\epsilon$
		
		\While{$\bm{\sigma}\neq\epsilon$}		
		
		\State $\sigma_{\text{bit}}\leftarrow \bm{\sigma}[1]$, $\bm{\sigma}\leftarrow\bm{\sigma}[2:]$
		
		\State $\bm{x}\leftarrow\genmodel_\ell(\bm{p},\bm{t})$
		
		\While{$H_3(\bm{m}||\bm{x}||\bm{\sigma}_{\textup{prev}})\neq\sigma_{\text{bit}}$}
		\State $\bm{x}\leftarrow\genmodel_\ell(\bm{p},\bm{t})$
		\EndWhile
		\State $\bm{m}\leftarrow\bm{m}||\bm{x}$, $~\bm{t}\leftarrow\bm{t}||\bm{x}$
		
		\State $\bm{\sigma}_{\textup{prev}}\leftarrow\bm{\sigma}_{\textup{prev}}||\sigma_{\text{bit}}$
		\EndWhile
		\EndWhile
		\If{$|\bm{t}|<n$}
		\State $\bm{t}\leftarrow\bm{t}||\genmodel_{n-|\bm{t}|}(\bm{p},\bm{t})$
		\EndIf
		\State \Return $\bm{t}$
	\end{algorithmic}
\end{algorithm}

	For any PPT adversary $\A$ attacking the off-the-record property (for any subset) of $\mddw$, we show a PPT adversary $\A'$ attacking the off-the-record property (for any subset) of $\mdvs$ as follows. 
	
	Receiving $\pp$ (note that the $\pp$ generated by $\setup$ of $\mddw$ is the same as that generated by the underlying $\mdvs.\setup$), $\A'$  sends $\pp$ to $\A$. $\A'$ maintains three local arrays, $L_{\text{ro},1}$,  $L_{\text{ro},2}$ and  $L_{\text{ro},3}$, to keep track of $\A$'s random oracle queries.  Then, $\A'$  answers $\A$'s oracle queries as below: 
	\begin{myitemize}
		\item \underline{$\mathcal{O}_{RO,\kappa}(str)~ (\kappa\in\{1,2,3\})$:} If there is some $(str,y)\in L_{\text{ro},\kappa}$,  $\A'$ returns $y$; otherwise,  $\A'$ samples $y\leftarrow\M$, adds $(str,y)$ to $L_{\text{ro},\kappa}$, and returns $y$. 
		
		
		\item\underline{$\orawk(i)$:}  $\A'$ queries its own oracle $\orask(i)$ to obtain $(\spk_i,\ssk_i)$, and returns $(\wpk_i,\wsk_i)=(\spk_i,\ssk_i)$ to $\A$. 
		
		\item\underline{$\oradk(j)$:}  $\A'$ queries its own oracle $\oravk(j)$ to obtain $(\vpk_j,\vsk_j)$, and returns $(\dpk_j,\dsk_j)=(\vpk_j,\vsk_j)$ to $\A$. 
		
		\item\underline{$\orawpk(i)$:}  $\A'$ queries its own oracle $\oraspk(i)$ to obtain $\spk_i$, and returns $\wpk_i=\spk_i$ to $\A$. 
		
		\item\underline{$\oradpk(j)$:}  $\A'$ queries its own oracle $\oravpk(j)$ to obtain $\vpk_j$, and returns $\dpk_j=\vpk_j$ to $\A$. 
		
	\item\underline{$\mathcal{O}_{\text{\emph{otr-chl}}}^{(b)}(i,S,\scor,\bm{p})$:}  $\A'$ 
	runs   $\textsf{SimOtr-Chl}_{\text{AS}}$ (with the help of its  challenge oracle $\mathcal{O}_{\text{\emph{otr-chl}}}^{(b)}$ for $\mdvs$, denoted as $\mathcal{O}_{\text{\emph{otr-chl}}}^{(b),\mdvs}$)   as shown in  Algorithm \ref{alg:SimOtrChl_AS}. In particular,  Algorithm \ref{alg:SimOtrChl_AS} is the same as Algorithm \ref{alg:SimWatm}, except that the signing oracle  $\orasign$ (in line 4) in  Algorithm \ref{alg:SimWatm} is replaced with $\A'$'s challenge  oracle  $\mathcal{O}_{\text{\emph{otr-chl}}}^{(b),\mdvs}$. Finally, $\A'$ 
	returns the output $\bm{t}$ of $\textsf{SimOtr-Chl}_{\text{AS}}$  to $\A$.
	
	\item\underline{$\oradet(i,j',S,\bm{t})$:}  $\A'$ 
	runs  $\textsf{SimDetect}$ (with the help of its verification oracle $\oraver$) as shown in  Algorithm \ref{alg:SimDet}. More specifically, Algorithm \ref{alg:SimDet} is the same as Algorithm \ref{alg:Det}, except that   the original  hash functions $H_1$, $H_2$ and $H_3$ (in line 3, line 7 and line 6, respectively) in  Algorithm \ref{alg:Det} are  replaced with  random oracles, and the original $\verify$ algorithm (in line 8) in  Algorithm \ref{alg:Det} is replaced with $\A'$'s verification oracle $\oraver$. 
	Finally, $\A'$ returns the output of $\textsf{SimDetect}$  to $\A$. 
	\end{myitemize}

	Finally, $\A'$ returns $\A$'s  output $b'$ as its  own final output. 
	
\begin{algorithm}
	\small
	\caption{$\textsf{SimOtr-Chl}_{\text{AS}}^{\mathcal{O}_{\text{\emph{otr-chl}}}^{(b),\mdvs},\{\mathcal{O}_{RO,\kappa}\}_{\kappa\in\{1,2,3\}}}(\pp,i,S,\scor,\bm{p})$\label{alg:SimOtrChl_AS}}

	\begin{algorithmic}[1]
		\State $\bm{t}\leftarrow\epsilon$
		
		\While{$|\bm{t}|+\ell+\ell\cdot\textsf{len}_{\text{sig}}<n$}
		\State $\bm{t}\leftarrow\bm{t}||\genmodel_\ell(\bm{p},\bm{t})$
		
		\State $\hat{\bm{\sigma}}\leftarrow\mathcal{O}_{\text{\emph{otr-chl}}}^{(b),\mdvs}(i,S,\scor, \mathcal{O}_{RO,1}(\bm{t}[-\ell:]))$
		
		\State $\bm{\sigma}\leftarrow \hat{\bm{\sigma}} \oplus \mathcal{O}_{RO,2}(\bm{t}[-\ell:])$
		
		\State $\bm{\sigma}_{\textup{prev}}\leftarrow\epsilon$, 
		$~\bm{m}\leftarrow\epsilon$
		
		\While{$\bm{\sigma}\neq\epsilon$}
		\State $\sigma_{\text{bit}}\leftarrow \bm{\sigma}[1]$, $~\bm{\sigma}\leftarrow\bm{\sigma}[2:]$
		
		\State $\bm{x}\leftarrow\genmodel_\ell(\bm{p},\bm{t})$
		
		\While{$\mathcal{O}_{RO,3}(\bm{m}||\bm{x}||\bm{\sigma}_{\textup{prev}})\neq\sigma_{\text{bit}}$}
		\State $\bm{x}\leftarrow\genmodel_\ell(\bm{p},\bm{t})$
		\EndWhile
		\State $\bm{m}\leftarrow\bm{m}||\bm{x}$ 			
		\State $\bm{t}\leftarrow\bm{t}||\bm{x}$ 			
		\State $\bm{\sigma}_{\textup{prev}}\leftarrow\bm{\sigma}_{\textup{prev}}||\sigma_{\text{bit}}$
		\EndWhile
		\EndWhile
		\If{$|\bm{t}|<n$}
		\State $\bm{t}\leftarrow\bm{t}||\genmodel_{n-|\bm{t}|}(\bm{p},\bm{t})$
		\EndIf
		
		\State \Return $\bm{t}$
	\end{algorithmic}
\end{algorithm}

	That is the construction of $\A'$.
	
	It is obvious that $\A'$ perfectly simulates game $\gameotrasmddw$ for $\A$, and   if $\A$ wins $\gameotrasmddw$, then  $\A'$ also wins $\textup{\textbf{G}}^{\textup{otr-as}}_{\mdvs,\A',\forgeas_{\mdvs}}(\lambda)$. So we derive 
	\[\textup{\textbf{Adv}}^{\textup{otr-as}}_{\mdvs,\A',\forgeas_{\mdvs}}(\lambda)\geq\advotrasmddw.\]\qed
\end{proof}

\section{Proof of Theorem \ref{thm:claimability}}\label{sec:security_proof_append_claim}

\begin{proof}
	
	Since $\mdvs$ is claimable, there are two PPT algorithms  $\cla_{\mdvs}$ and $\claver_{\mdvs}$. 
	
	We construct two PPT algorithms $\cla$ and $\claver$ for $\mddw$ as shown in Algorithm \ref{alg:Claim} and  Algorithm \ref{alg:ClaimVer}, respectively.  
	
\begin{algorithm}
	\small
	\caption{$\cla(\pp,\wsk_i,\{\dpk_j\}_{j \in S},\bm{t})$ \label{alg:Claim} }
	\begin{algorithmic}[1]
		\State $\ssk_i\leftarrow \wsk_i$
		\State $\{\vpk_j\}_{j \in S}\leftarrow\{\dpk_j\}_{j \in S}$
		\State  $\mu\leftarrow 0$ 
		\State $\textsf{Count}\leftarrow 0$
		
		\While{$\mu < n-(\ell+\ell\cdot\textsf{len}_{\text{sig}})$}
		\State $\widetilde{m}\leftarrow H_1(\bm{t}[\mu:\mu+\ell-1])$,  $\bm{\sigma}\leftarrow\epsilon$, 
		$\bm{m}\leftarrow\epsilon$
		
		\For{$\varphi\in[\textsf{len}_{\text{sig}}]$}
		\State $\bm{m}\leftarrow\bm{m}||\bm{t}[\mu+\varphi\cdot\ell:\mu+\varphi\cdot\ell+\ell-1]$
		
		\State $\bm{\sigma}\leftarrow\bm{\sigma}||H_3(\bm{m}||\bm{\sigma})$
		\EndFor
		
		\State $\hat{\bm{\sigma}}\leftarrow \bm{\sigma} \oplus H_2(\bm{t}[\mu:\mu+\ell-1])$
		
		\State $\textsf{Count}\leftarrow\textsf{Count}+1$
		
		\State $\pi_{\textsf{Count}}\leftarrow\cla_{\mdvs}(\pp,\ssk_i,\{\vpk_j\}_{j\in S},\hat{\bm{\sigma}})$ 
		
		\State $\mu\leftarrow \mu+(\ell+\textsf{len}_{\text{sig}}\cdot\ell)\cdot\textsf{Count}$
		
		\EndWhile
		
		\State $\pi\leftarrow (\pi_1,\cdots,\pi_{\lfloor\frac{n-\textsf{len}_{\text{sig}}-\ell}{(\textsf{len}_{\text{sig}}+1)\cdot\ell} \rfloor})$
		
		\State \Return $\pi$
	\end{algorithmic}
\end{algorithm}

\begin{algorithm}
	\small
	\caption{$\claver(\pp,\wpk_i,\{\dpk_j\}_{j \in S},\bm{t},\pi)$ \label{alg:ClaimVer} }
	\begin{algorithmic}[1]
		\State $\spk_i\leftarrow \wpk_i$
		\State $\{\vpk_j\}_{j \in S}\leftarrow\{\dpk_j\}_{j \in S}$
		\State $(\pi_1,\cdots,\pi_{\lfloor\frac{n-\ell\cdot\textsf{len}_{\text{sig}}-\ell}{(\textsf{len}_{\text{sig}}+1)\cdot\ell} \rfloor})\leftarrow \pi$
		\State $\mu\leftarrow 0$
		\State $\textsf{Count}\leftarrow 0$
		
		\While{$\mu < n-(\ell+\textsf{len}_{\text{sig}})$}
		\State $\widetilde{m}\leftarrow H_1(\bm{t}[\mu:\mu+\ell-1])$,  $\bm{\sigma}\leftarrow\epsilon$, 
		$\bm{m}\leftarrow\epsilon$
		
		\For{$\varphi\in[\textsf{len}_{\text{sig}}]$}
		\State $\bm{m}\leftarrow\bm{m}||\bm{t}[\mu+\varphi\cdot\ell:\mu+\varphi\cdot\ell+\ell-1]$
		
		\State $\bm{\sigma}\leftarrow\bm{\sigma}||H_3(\bm{m}||\bm{\sigma})$
		\EndFor
		
		\State $\hat{\bm{\sigma}}\leftarrow \bm{\sigma} \oplus H_2(\bm{t}[\mu:\mu+\ell-1])$
		
		\State $\textsf{Count}\leftarrow\textsf{Count}+1$

		\If{$\claver_{\mdvs}(\pp,\spk_i,\{\vpk_j\}_{j\in S},\hat{\bm{\sigma}},\pi_{\textsf{Count}})=1$}
		\State \Return 1
		\EndIf
		
		\State $\mu\leftarrow \mu+(\ell+\textsf{len}_{\text{sig}}\cdot\ell)\cdot\textsf{Count}$
		
		\EndWhile

		\State \Return $0$
	\end{algorithmic}
\end{algorithm}

		%
		%
		%
		%
		%
		%
		%
		%
		%
		%
		%
		%
		%
	
	Now, we show that $\mddw$ with the two algorithms, $\cla$ and $\claver$,  satisfies each of the three conditions of Definition \ref{def:mddw_claim}. 
	\vspace{1mm}
	
	\noindent\underline{\textbf{(1) For the first condition of Definition \ref{def:mddw_claim}}}
	
	The first condition of Definition \ref{def:mddw_claim} for $\mddw$ is immediately fulfilled  by the first condition  of Definition \ref{def:mdvs_claim} for  the underlying $\mdvs$ (encompassing  $\cla_{\mdvs}$ and $\claver_{\mdvs}$). 
		\vspace{1mm}
	
	\noindent\underline{\textbf{(2) For the second  condition of Definition \ref{def:mddw_claim}}}
	
	We turn to the second condition of Definition \ref{def:mddw_claim} for $\mddw$. 
	
	For any PPT adversary $\A=(\A_1,\A_2)$ attacking  the second condition of Definition \ref{def:mddw_claim} for $\mddw$, we show a PPT adversary $\B=(\B_1,\B_2)$ attacking  the second condition of Definition \ref{def:mdvs_claim} for  $\mdvs$ as follows. 
	
	Upon receiving $\pp$ (note that the $\pp$ generated by $\setup$ of $\mddw$ is the same as that generated by the underlying $\mdvs.\setup$), $\B_1$ sends $\pp$ to $\A_1$. $\B$ maintains three local arrays, $L_{\text{ro},1}$,  $L_{\text{ro},2}$ and  $L_{\text{ro},3}$, to keep track of $\A$'s random oracle queries.  Then, $\B_1$  answers $\A_1$'s oracle queries as below: 
	\begin{myitemize}
		\item \underline{$\mathcal{O}_{RO,\kappa}(str)~ (\kappa\in\{1,2,3\})$:} If there is some $(str,y)\in L_{\text{ro},\kappa}$,  $\B_1$ returns $y$; otherwise,  $\B_1$  samples $y\leftarrow\M$, adds $(str,y)$ to $L_{\text{ro},\kappa}$, and returns $y$. 
		
		
		\item\underline{$\orawk(i)$:}  $\B_1$ queries its own oracle $\orask(i)$ to obtain $(\spk_i,\ssk_i)$, and returns $(\wpk_i,\wsk_i)=(\spk_i,\ssk_i)$ to $\A_1$. 
		
		\item\underline{$\oradk(j)$:}  $\B_1$ queries its own oracle $\oravk(j)$ to obtain $(\vpk_j,\vsk_j)$, and returns $(\dpk_j,\dsk_j)=(\vpk_j,\vsk_j)$ to $\A_1$. 
		
		\item\underline{$\orawpk(i)$:}  $\B_1$ queries its own oracle $\oraspk(i)$ to obtain $\spk_i$, and returns $\wpk_i=\spk_i$ to $\A_1$. 
		
		\item\underline{$\oradpk(j)$:}  $\B_1$ queries its own oracle $\oravpk(j)$ to obtain $\vpk_j$, and returns $\dpk_j=\vpk_j$ to $\A_1$. 
		
		\item\underline{$\orawatm(i,S,\bm{p})$:}  $\B_1$ 
		runs   $\textsf{SimWatMar}$ (with the help of its  signing oracle $\mathcal{O}_S$)   as shown in  Algorithm \ref{alg:SimWatm}. More specifically, Algorithm \ref{alg:SimWatm} is the same as Algorithm \ref{alg:Watm}, except that   the original  hash functions $H_1$, $H_2$ and $H_3$  (in line 4, line 5 and line 10, respectively) in  Algorithm \ref{alg:Watm} are replaced with random oracles, and the original $\sign$ algorithm  (in line 4) in  Algorithm \ref{alg:Watm} is replaced with $\B_1$'s signing oracle $\mathcal{O}_S$. Finally, $\B_1$  
		returns the output $\bm{t}$ of $\textsf{SimWatMar}$  to $\A_1$.

	\item\underline{$\mathcal{O}_{D}(i,j',S,\bm{t})$:}  $\B_1$ 
	runs  $\textsf{SimDetect}$ (with the help of its verification oracle $\mathcal{O}_V$) as shown in  Algorithm \ref{alg:SimDet}. More specifically, Algorithm \ref{alg:SimDet} is the same as Algorithm \ref{alg:Det}, except that   the original  hash functions $H_1$, $H_2$ and $H_3$ (in line 3, line 7 and line 6, respectively) in  Algorithm \ref{alg:Det} are  replaced with  random oracles, and the original $\verify$ algorithm (in line 8) in  Algorithm \ref{alg:Det} is replaced with $\B_1$'s verification oracle $\mathcal{O}_V$. 
	Finally, $\B_1$ returns the output of $\textsf{SimDetect}$  to $\A_1$. 
	
	\item\underline{$\oraclm(i,S,\bm{t})$:}  $\B_1$ 
	runs  $\textsf{SimClaim}$ (with the help of its own claiming oracle $\oraclm^{\mdvs}$, which denotes the claiming oralce for $\B$ in game  $\textup{\textbf{G}}^{\textup{clm-unf}}_{\mdvs,\B}(\lambda)$) as shown in  Algorithm \ref{alg:SimClaim}. More specifically, Algorithm \ref{alg:SimClaim} is the same as Algorithm \ref{alg:Claim}, except that   the original  hash functions $H_1$, $H_2$ and $H_3$ (in line 4, line 8 and line 7, respectively) in  Algorithm \ref{alg:Claim} are  replaced with  random oracles, and the original $\cla$ algorithm (in line 10) in  Algorithm \ref{alg:Claim} is replaced with $\B_1$'s claiming oracle $\oraclm^{\mdvs}$. 
	Finally, $\B_1$ returns the output of $\textsf{SimClaim}$  to $\A_1$. 
	\end{myitemize}
	
	\begin{algorithm}[h]
	\small
	\caption{$\textsf{SimClaim}(\pp,i,S,\bm{t})$ \label{alg:SimClaim} }
	\begin{algorithmic}[1]
		\State $\mu\leftarrow 0$
		\State $\textsf{Count}\leftarrow 0$
		
		\While{$\mu < n-(\ell+\ell\cdot\textsf{len}_{\text{sig}})$}
		\State $\widetilde{m}\leftarrow \mathcal{O}_{RO,1}(\bm{t}[\mu:\mu+\ell-1])$,  $\bm{\sigma}\leftarrow\epsilon$, 
		$\bm{m}\leftarrow\epsilon$
		
		\For{$\varphi\in[\textsf{len}_{\text{sig}}]$}
		\State $\bm{m}\leftarrow\bm{m}||\bm{t}[\mu+\varphi\cdot\ell:\mu+\varphi\cdot\ell+\ell-1]$
		
		\State $\bm{\sigma}\leftarrow\bm{\sigma}||\mathcal{O}_{RO,3}(\bm{m}||\bm{\sigma})$
		\EndFor
		
		\State $\hat{\bm{\sigma}}\leftarrow \bm{\sigma} \oplus \mathcal{O}_{RO,2}(\bm{t}[\mu:\mu+\ell-1])$
		
		\State $\textsf{Count}\leftarrow\textsf{Count}+1$
		
		\State $\pi_{\textsf{Count}}\leftarrow\oraclm^{\mdvs}(i,S,\hat{\bm{\sigma}})$ 
		
		\State $\mu\leftarrow \mu+(\ell+\textsf{len}_{\text{sig}}\cdot\ell)\cdot\textsf{Count}$
		
		\EndWhile
		
		\State $\pi\leftarrow (\pi_1,\cdots,\pi_{\lfloor\frac{n-\textsf{len}_{\text{sig}}-\ell}{(\textsf{len}_{\text{sig}}+1)\cdot\ell} \rfloor})$
		
		\State \Return $\pi$
	\end{algorithmic}
\end{algorithm}

Receiving $(i^*,S^*,\bm{p}^*)$ from $\A_1$, $\B_1$ 
uniformly samples $\theta\leftarrow\{1,\cdots, \lfloor\frac{n-\ell\cdot\textsf{len}_{\text{sig}}-\ell}{(\textsf{len}_{\text{sig}}+1)\cdot\ell} \rfloor\}$, and generates $\bm{t}^*$ via running  $\textsf{SimWatMar}_{\theta}(\pp,i^*,S^*,\bm{p}^*)$   (as shown in  Algorithm \ref{alg:SimWatm_detect_robust}). In particular, Algorithm \ref{alg:SimWatm_detect_robust} is the same as Algorithm \ref{alg:SimWatm} except that 
when the while loop (spanning from line 2 to line 14 in Algorithm \ref{alg:SimWatm}) enters its $\theta$-th cycle, $B_1$ outputs  $(i^*,S^*,m^*= \mathcal{O}_{RO,1}(\bm{t}[-\ell:]))$, and after receiving $\sigma^*$ from the challenger, $B_2$ sets $\hat{\bm{\sigma}}=\sigma^*$ and then continues the procedure of Algorithm \ref{alg:SimWatm}. 
	
		\begin{algorithm}[h]
		\small
		\caption{$\textsf{SimWatMar}_{\theta}^{\mathcal{O}_S,\{\mathcal{O}_{RO,\kappa}\}_{\kappa\in\{1,2,3\}}}(\pp,i,S,\bm{p})$\label{alg:SimWatm_detect_robust}}

		\begin{algorithmic}[1]
			\State $\bm{t}\leftarrow\epsilon$
			
			\State $\textsf{Count}\leftarrow 0$
			
			\While{$|\bm{t}|+\ell+\ell\cdot\textsf{len}_{\text{sig}}<n$}
			
			\State $\textsf{Count}\leftarrow \textsf{Count}+1$
			
			\State $\bm{t}\leftarrow\bm{t}||\genmodel_\ell(\bm{p},\bm{t})$
			
			\If{$\textsf{Count}\neq \theta$}
			\State $\hat{\bm{\sigma}}\leftarrow\mathcal{O}_{S}(i,S, \mathcal{O}_{RO,1}(\bm{t}[-\ell:]))$
			
			\Else
			\State $\B_1$ outputs $(i^*,S^*,m^*= \mathcal{O}_{RO,1}(\bm{t}[-\ell:]))$, and then $\B_2$ receives $\sigma^*$ as input. 
			
			\State $\hat{\bm{\sigma}}\leftarrow\sigma^*$
			\EndIf
			
			\State $\bm{\sigma}\leftarrow \hat{\bm{\sigma}} \oplus \mathcal{O}_{RO,2}(\bm{t}[-\ell:])$
			
			\State $\bm{\sigma}_{\textup{prev}}\leftarrow\epsilon$, 
			$~\bm{m}\leftarrow\epsilon$
			
			\While{$\bm{\sigma}\neq\epsilon$}
			\State $\sigma_{\text{bit}}\leftarrow \bm{\sigma}[1]$, $~\bm{\sigma}\leftarrow\bm{\sigma}[2:]$
			
			\State $\bm{x}\leftarrow\genmodel_\ell(\bm{p},\bm{t})$
			
			\While{$\mathcal{O}_{RO,3}(\bm{m}||\bm{x}||\bm{\sigma}_{\textup{prev}})\neq\sigma_{\text{bit}}$}
			\State $\bm{x}\leftarrow\genmodel_\ell(\bm{p},\bm{t})$
			\EndWhile
			\State $\bm{m}\leftarrow\bm{m}||\bm{x}$ 			
			\State $\bm{t}\leftarrow\bm{t}||\bm{x}$			
			\State $\bm{\sigma}_{\textup{prev}}\leftarrow\bm{\sigma}_{\textup{prev}}||\sigma_{\text{bit}}$
			\EndWhile
			\EndWhile
			\If{$|\bm{t}|<n$}
			\State $\bm{t}\leftarrow\bm{t}||\genmodel_{n-|\bm{t}|}(\bm{p},\bm{t})$
			\EndIf
			
			\State \Return $\bm{t}$
		\end{algorithmic}
	\end{algorithm}

	$\B_2$ sends $\bm{t}^*$ to $\A_2$, and  then  answers $\A_2$'s oracle queries in the same manner as $\B_1$ does. 

Finally, receiving $\A_2$'s final output  $(\widehat{\pi},\widehat{i'})$,  $\B_2$ parses $\widehat{\pi}=(\pi_1,\cdots,\pi_{\lfloor\frac{n-\ell\cdot\textsf{len}_{\text{sig}}-\ell}{(\textsf{len}_{\text{sig}}+1)\cdot\ell} \rfloor})$, and then returns $(\pi^*,i')=(\pi_{\theta},\widehat{i'})$ as its final output.

That is the construction of $\B$. Now we turn to analyze $\B$'s advantage. 

It is easy to see that $\B$ perfectly simulates game  $\gameclaimunfmddw$ for $\A$. 
Since $\theta$ is uniformly sampled from $\{1,\cdots, \lfloor\frac{n-\ell\cdot\textsf{len}_{\text{sig}}-\ell}{(\textsf{len}_{\text{sig}}+1)\cdot\ell} \rfloor\}$, according to the description of $\claver$ (in Algorithm \ref{alg:ClaimVer}), we derive 
\[\textup{\textbf{Adv}}^{\textup{clm-unf}}_{\mdvs,\B}(\lambda)\geq\frac{1}{\lfloor\frac{n-\ell\cdot\textsf{len}_{\text{sig}}-\ell}{(\textsf{len}_{\text{sig}}+1)\cdot\ell} \rfloor}\advclaimunfmddw.\] 
	
		\vspace{1mm}
	
	\noindent\underline{\textbf{(3) For the third  condition of Definition \ref{def:mddw_claim}}}
	

	For any PPT adversary $\A$ attacking the third condition of Definition \ref{def:mddw_claim} for $\mddw$, we show a PPT adversary $\A'$ attacking the third condition of Definition \ref{def:mdvs_claim} for $\mdvs$ as follows. 
	
	Receiving $\pp$ (note that the $\pp$ generated by $\setup$ of $\mddw$ is the same as that generated by the underlying $\mdvs.\setup$), $\A'$  sends $\pp$ to $\A$. $\A'$ maintains three local arrays, $L_{\text{ro},1}$,  $L_{\text{ro},2}$ and  $L_{\text{ro},3}$, to keep track of $\A$'s random oracle queries.  Then, $\A'$  answers $\A$'s oracle queries as below: 
		\begin{myitemize}
		\item \underline{$\mathcal{O}_{RO,\kappa}(str)~ (\kappa\in\{1,2,3\})$:} If there is some $(str,y)\in L_{\text{ro},\kappa}$,  $\A'$ returns $y$; otherwise,  $\A'$ samples $y\leftarrow\M$, adds $(str,y)$ to $L_{\text{ro},\kappa}$, and returns $y$. 
		
		
		\item\underline{$\orawk(i)$:}  $\A'$ queries its own oracle $\orask(i)$ to obtain $(\spk_i,\ssk_i)$, and returns $(\wpk_i,\wsk_i)=(\spk_i,\ssk_i)$ to $\A$. 
		
		\item\underline{$\oradk(j)$:}  $\A'$ queries its own oracle $\oravk(j)$ to obtain $(\vpk_j,\vsk_j)$, and returns $(\dpk_j,\dsk_j)=(\vpk_j,\vsk_j)$ to $\A$. 
		
		\item\underline{$\orawpk(i)$:}  $\A'$ queries its own oracle $\oraspk(i)$ to obtain $\spk_i$, and returns $\wpk_i=\spk_i$ to $\A$.

		\item\underline{$\oradpk(j)$:}  $\A'$ queries its own oracle $\oravpk(j)$ to obtain $\vpk_j$, and returns $\dpk_j=\vpk_j$ to $\A$. 
		
		\item\underline{$\orawatm(i,S,\bm{p})$:}  $\A'$  	runs   $\textsf{SimWatMar}$ (with the help of its  signing oracle $\mathcal{O}_S$)   as shown in  Algorithm \ref{alg:SimWatm}. More specifically, Algorithm \ref{alg:SimWatm} is the same as Algorithm \ref{alg:Watm}, except that   the original  hash functions $H_1$, $H_2$ and $H_3$  (in line 4, line 5 and line 10, respectively) in  Algorithm \ref{alg:Watm} are replaced with random oracles, and the original $\sign$ algorithm  (in line 4) in  Algorithm \ref{alg:Watm} is replaced with $\A'$ 's signing oracle $\mathcal{O}_S$. Finally, $\A'$   	returns the output $\bm{t}$ of $\textsf{SimWatMar}$  to $\A$.

		\item\underline{$\oradet(i,j',S,\bm{t})$:}  $\A'$ 
		runs  $\textsf{SimDetect}$ (with the help of its verification oracle $\oraver$) as shown in  Algorithm \ref{alg:SimDet}. More specifically, Algorithm \ref{alg:SimDet} is the same as Algorithm \ref{alg:Det}, except that   the original  hash functions $H_1$, $H_2$ and $H_3$ (in line 3, line 7 and line 6, respectively) in  Algorithm \ref{alg:Det} are  replaced with  random oracles, and the original $\verify$ algorithm (in line 8) in  Algorithm \ref{alg:Det} is replaced with $\A'$'s verification oracle $\oraver$. 
		Finally, $\A'$ returns the output of $\textsf{SimDetect}$  to $\A$. 
		
		\item\underline{$\oraclm(i,S,\bm{t})$:}  $\A'$ 
		runs  $\textsf{SimClaim}$ (with the help of its own claiming oracle $\oraclm^{\mdvs}$, which denotes the claiming oralce for  $\A'$  in game  $\textup{\textbf{G}}^{\textup{non-fram}}_{\mdvs,\A'}(\lambda)$) as shown in  Algorithm \ref{alg:SimClaim}. More specifically, Algorithm \ref{alg:SimClaim} is the same as Algorithm \ref{alg:Claim}, except that   the original  hash functions $H_1$, $H_2$ and $H_3$ (in line 4, line 8 and line 7, respectively) in  Algorithm \ref{alg:Claim} are  replaced with  random oracles, and the original $\cla$ algorithm (in line 10) in  Algorithm \ref{alg:Claim} is replaced with  $\A'$'s claiming oracle $\oraclm^{\mdvs}$. 
		Finally,  $\A'$ returns the output of $\textsf{SimClaim}$  to $\A$. 
	\end{myitemize}

Receiving $\A$'s final output $(i^*,S^*,\bm{t}^*,\pi^*)$, $\A'$ firstlys checks whether $\claver(\pp,$ $\wpk_{i^*},\{\dpk_{j}\}_{j\in S^*},\bm{t}^*,\pi^*)=1$. If not, $\A'$  returns $(i^*,S^*,m^*_{\text{rand}},\sigma^*_{\text{rand}},\pi^*_{\text{rand}})$ as its own final output, where $(m^*_{\text{rand}},\sigma^*_{\text{rand}},\pi^*_{\text{rand}})$ are uniformly sampled. Otherwise, there must be some $\textsf{Count}\in[\lfloor\frac{n-\ell\cdot\textsf{len}_{\text{sig}}-\ell}{(\textsf{len}_{\text{sig}}+1)\cdot\ell} \rfloor]$ satisfying $\claver_{\mdvs}(\pp,\spk_i,$ $\{\vpk_j\}_{j\in S},\hat{\bm{\sigma}},\pi_{\textsf{Count}})=1$, where $\hat{\bm{\sigma}}$ is generated according to lines 4-11 in Algorithm \ref{alg:ClaimVer}. In this case, $\A'$ extracts $\hat{\bm{\sigma}}$, $\pi_{\textsf{Count}}$ and the corresponding $\widetilde{m}$ (for $\hat{\bm{\sigma}}$). Then, $\A'$ returns $(i^*,S^*,\widetilde{m},\hat{\bm{\sigma}},\pi_{\textsf{Count}},\pi^*)$ as its final output. 

That is the construction of $\A'$. Now we turn to analyze $\A'$'s advantage. 

It is easy to see that $\A'$ perfectly simulates game  $\gameunframemddw$ for $\A$, and  $\A'$ wins $\textup{\textbf{G}}^{\textup{clm-fram}}_{\mdvs,\A'}(\lambda)$ if $\A$ wins $\gameunframemddw$. So we derive  
\[\textup{\textbf{Adv}}^{\textup{non-fram}}_{\mdvs,\A'}(\lambda)\geq\advunframemddw.\] 
	\qed
\end{proof}
\section{Proof of Theorem \ref{thm:CMDVS_security}}\label{sec:security_proof_append_cmdvs}
\begin{proof}
Since the proofs that $\cmdvs$ satisfies correctness, consistency and existential unforgeability are straightforward, we will only focus on proving  off-the-record property and claimability.

\vspace{1mm}
\noindent \underline{\textbf{Off-the-record.}} Now, we prove that $\cmdvs$ is off-the-record for any subset if the underlying $\mdvs$ is off-the-record for any subset.  The proof for $\cmdvs$ being off-the-record for   designated set, assuming the underlying $\mdvs$ is off-the-record for   designated set, is nearly identical and is thus omitted here. 

Since $\mdvs$ is off-the-record for any subset, there is a PPT algorithm $\mdvs.\forgeas$. We provide a forging algorithm $\forgeas$, based on $\mdvs.\forgeas$, for $\cmdvs$ as shown in Fig. \ref{fig:forge_cmdvs}, where $\mathcal{COM}$ is the commitment space of $\commit$. 

\begin{figure*}[!h]
	\scriptsize
	\centering
	\fbox{
		\begin{minipage}
			[t]{0.71\textwidth}{
				\underline{$\forgeas(\pp,\spk_i,\{\vpk_j\}_{j\in S}, \{\spk_j\}_{j\in \scor},m)$}:\\
				$\sigma_{\mdvs} \leftarrow \mdvs.\forgeas(\pp_{\mdvs},\rspks,\{\vpk_j\}_{j\in S},\{\vsk_{j}\}_{j \in \scor},m)$\\
				$\com \leftarrow  \mathcal{COM}$\\
				Return $\sigma = (\sigma_{\mdvs},\com)$			 	}
		\end{minipage}	
	}	
	\captionsetup{justification=centering}
	\caption{Forging algorithm $\forgeas$ for $\cmdvs$ 
	}
	\label{fig:forge_cmdvs}
\end{figure*}

For any PPT adversary $\A$, we consider the following sequence of games. 



\begin{itemize}
	\item $\textup{Game}_0$: $\textup{Game}_0$ is   $\textbf{G}_{\cmdvs,\A,\forgeas}^{\textup{otr-as}}(\lambda)$ when $b=0$. 
	Thus, we have \[\Pr[\textup{Game}_0=1] = \Pr[\textbf{G}_{\cmdvs,\A,\forgeas}^{\textup{otr-as}}(\lambda)=1|b=0].\]
	\item $\textup{Game}_1$: $\textup{Game}_1$ is the same as $\textup{Game}_0$, except that when the adversary $\A$ queries to the signing oracle $\mathcal{O}_{\text{\emph{otr-chl}}}^{(b)}$ on $(i,S,\scor,m)$, the challenger generates the MDVS signature $\sigma_{\mdvs}$ via the forging algorithm $\mdvs.\forgeas$ of the underlying $\mdvs$. By the off-the-record property for any subset of $\mdvs$, we derive  \[|\Pr[\textup{Game}_1=1] - \Pr[\textup{Game}_0=1] | \leq \negl .\] 
	\item $\textup{Game}_2$: $\textup{Game}_2$ is the same as $\textup{Game}_1$, except that when the adversary $\A$ queries to the signing oracle $\mathcal{O}_{\text{\emph{otr-chl}}}^{(b)}$ on $(i,S,\scor,m)$, the challenger chooses a random commitment $\com$ over the commitment space $\mathcal{COM}$. By the hiding property of $\commit$, we have 
	\[|\Pr[\textup{Game}_2=1] - \Pr[\textup{Game}_1=1] | \leq \negl .\]
\end{itemize} 
Note that  $\textup{Game}_2$ is equivalent to $\textbf{G}_{\cmdvs,\A,\forgeas}^{\textup{otr-as}}(\lambda)$ when $b=1$. Thus, 
\[\Pr[\textup{Game}_2=1] = \Pr[\textbf{G}_{\cmdvs,\A,\forgeas}^{\textup{otr-as}}(\lambda)=1|b=1].\]
Therefore, it holds that 
\begin{align*}
	&\textbf{Adv}_{\cmdvs,\A,\forgeas}^{\textup{otr-as}}(\lambda)\\
	=&|\Pr[\textbf{G}_{\cmdvs,\A,\forgeas}^{\textup{otr-as}}(\lambda)=1]-\frac{1}{2}|\\
	=&\frac{1}{2}|\Pr[\textbf{G}_{\cmdvs,\A,\forgeas}^{\textup{otr-as}}(\lambda)=1|b=0]-\Pr[\textbf{G}_{\cmdvs,\A,\forgeas}^{\textup{otr-as}}(\lambda)=1|b=1]|\\
	=&\frac{1}{2}|\Pr[\textup{Game}_0=1]-\Pr[\textup{Game}_2=1]|\leq \negl.
\end{align*}


\noindent \underline{\textbf{Claimability.}} Now, we show that  $\cmdvs$ satisfies each of the three conditions of Definition  \ref{def:mdvs_claim}. 

\vspace{1mm}
\noindent\underline{\emph{(1) For the first condition of Definition \ref{def:mdvs_claim}}}

For any signer $i$, any message $m \in \M$ and any verifier identity set $S$, given $\pp \leftarrow \setup(1^{\lambda})$, $(\spk_i,\ssk_i) \leftarrow \skg(\pp)$, $\{(\vpk_j,\vsk_j) \leftarrow \vkg(\pp)\}_{j \in S}$, $\sigma \leftarrow \sign(\pp,\ssk_i,\{\vpk_j\}_{j\in S}, m)$ and $\pi \leftarrow \cla(\pp,\ssk_i,\{\vpk_j\}_{j\in S},$ $\sigma)$, we have $\ssk_i=(\ks,\sgsks,\rssks)$, $\sigma = (\sigma_{\mdvs},\com)$ and $\pi=(r_{\commit},\sigma_{\sig})$. According to the   descriptions of $\sign$ and $\cla$ in Fig. \ref{fig:algo_cmdvs}, we derive   $r_{\sig} = \prf.\eval(\ks,$ $(\spk_i,\sigma_{\mdvs},0))$, $\sigma_{\sig} = \sig.\sign(\pp_{\sig},\sgsks,(\spk_i,\sigma_{\mdvs});r_{\sig})$, and $r_{\commit} = \prf.\eval(\ks,(\spk_i,$ $\sigma_{\mdvs},1))$. The correctness of $\commit$ implies   $\commit.\Com(\pp_{\commit},$ $(\spk_i,\{\vpk_j\}_{j\in S},\sigma_{\sig});r_{\commit})=\com$,
and the correctness of $\sig$ implies that $\sig.\verify(\pp_{\sig},\spk_{\sig,i},(\spk_i,\sigma_{\mdvs}),\sigma_{\sig})=1$. Hence, $\claver$ would output $1$, which implies that the scheme $\cmdvs$ is correct.
\vspace{1mm}

\noindent\underline{\emph{(2) For the second  condition of Definition \ref{def:mdvs_claim}}}

Assume that there exits a PPT adversary $\A=(\A_1,\A_2)$ that successfully attacks  the second  condition of Definition \ref{def:mdvs_claim} with non-negligible advantage . We construct a PPT adversary $\B$ to break the binding property of $\commit$ with non-negligible advantage  as follows.

On receiving the public parameter $\pp_{\commit}$, the adversary $\B$  generates $\pp_{\mdvs}$ and $\pp_{\sig}$ via invoking $\mdvs.\setup$ and $\sig.\setup$, respectively. Subsequently, $\B$ sends   $\pp=(\pp_{\mdvs},\pp_{\sig},\pp_{\commit})$ to  $\A_1$, and then answers $\A_1$'s oracle queries by itself.

After receiving $(i^*,S^*,m^*)$ from   $\A_1$, where $i^*$ is never queried to the oracle $\orask$ by $\A_1$, $\B$ proceeds as follows: 
\begin{enumerate}
	\item[(1)] $(\spk_{i^*},\ssk_{i^*})\leftarrow\orask(i^*)$, and  $\{\vpk_j\leftarrow\oravpk(j)\}_{j\in S^*}$, where the oracles $\orask$ and $\oravpk$ are   described in Fig.  \ref{fig:MDVS_oracle}.
	\item[(2)] $\sigma^*\leftarrow\sign(\pp,\ssk_{i^*},\{\vpk_j\}_{j\in S^*},m^*)$.
\end{enumerate}
$\B$ sends $\sigma^*$ to $\A_2$, and then answers  $\A_2$'s oracle queries, with the condition that queries of $i^*$ to $\orask$ and $(*,*,\sigma^*)$ to $\oraclm$ are not permitted. It is important  to note that $\ssk_{i^*}$ can be parsed as $\ssk_{i^*}=(k_{i^*},\ssk_{\sig,i^*},\ssk_{\mdvs,i^*})$, and $\sigma^*$ can be parsed as $\sigma^*=(\sigma_{\mdvs}^*,\com^*)$.



Upon receiving $(\pi^*,i')$ from $\A$, $\B$ firstly checks if $(\claver(\pp,\spk_{i'},\{\vpk_j\}_{j\in S^*},$ $\sigma^*,\pi^*) = 1)\wedge(i' \neq i^*)$. If not, $\B$ returns uniformly sampled $(\com_{\text{rand}},m_{\text{rand}},r_{\text{rand}},$ $m'_{\text{rand}},r'_{\text{rand}})$ as its final output. Otherwise, $\B$ parses $\pi^*=(r'_{\commit},\sigma'_{\sig})$, and returns
\[(\com^*,m=(\spk_{i^*},\{\vpk_j\}_{j\in S^*},\sigma^*_{\sig}),r^*_{\commit},m'=(\spk_{i'},\sigma'_{\sig}),r'_{\commit})\]
as its final output, where $\sigma^*_{\sig}$ is generated via ``$r_{\sig}\leftarrow\prf.\eval(k_{i^*},(\spk_{i^*},\sigma^*_{\mdvs},$ $0))$,    $\sigma^*_{\sig}\leftarrow\sig.\sign(\pp_{\sig},\ssk_{\sig,i^*},(\spk_{i^*},\sigma^*_{\mdvs});r_{\sig})$'', and $r^*_{\commit}$ is generated via ``$r^*_{\commit}\leftarrow\prf.\eval(k_{i^*},$ $(\spk_{i^*},\sigma^*_{\mdvs},1))$''.  


Given that $i' \neq i^*$, it follows that  $\spk_{i'} \neq \spk_{i^*}$, which in turn means that $m \neq m'$. Moreover, the correctness of $\commit$ guarantees  that  $\commit.\dec(\pp_{\commit},$ $\com^*,r^*_{\commit},m)=1$. The fact that $\claver(\pp,\spk_{i'},\{\vpk_j\}_{j\in S^*},\sigma^*,\pi^*) = 1$ implies that $\commit.\dec(\pp_{\commit},\com^*,$ $r'_{\commit},m')=1$.  

Hence, if $\A$ successfully attacks  the second  condition of Definition \ref{def:mdvs_claim} with non-negligible advantage,  $\B$ will also successfully  breaks the binding property  of the $\commit$ with non-negligible advantage.


\vspace{1mm}

\noindent\underline{\emph{(3) For the third   condition of Definition \ref{def:mdvs_claim}}} 

For any PPT adversary $\A$ attacking the third   condition of Definition \ref{def:mdvs_claim} (of $\cmdvs$), its advantage is 
\begin{eqnarray}\label{eq:clm_begin}
\textup{\textbf{Adv}}^{\textup{non-fram}}_{\cmdvs,\A}(\lambda)=\Pr[\textup{\textbf{G}}^{\textup{non-fram}}_{\cmdvs,\A}(\lambda)=1].
\end{eqnarray}

Assume that $\A$ makes at most $\poly$ oracle queries. 

Now, we consider the following sequence of games.

\noindent\underline{$\textup{Game}_0$:} $\textup{Game}_0$ is the same as $\textup{\textbf{G}}^{\textup{non-fram}}_{\cmdvs,\A}(\lambda)$. Thus, we have 
\begin{eqnarray}\Pr[\textup{Game}_0=1] = \Pr[\textup{\textbf{G}}^{\textup{non-fram}}_{\cmdvs,\A}(\lambda)=1].
\end{eqnarray}

\noindent\underline{$\textup{Game}_1$:} $\textup{Game}_1$ is the same as  $\textup{Game}_0$, except that at the beginning of this game, the challenger uniformly samples $\widehat{i}\leftarrow[\poly]$, and then if $\A$'s final output $(i^*,*,*,*,*)$ satisfies $i^*\neq \widehat{i}$, the challenger returns 0 directly. 

Clearly, we derive 
\begin{eqnarray}
	\Pr[\textup{Game}_1 = 1] = \frac{1}{\poly} \Pr[\textup{Game}_0 = 1].
\end{eqnarray}

\noindent\underline{$\textup{Game}_2$:} $\textup{Game}_2$ is the same as  $\textup{Game}_1$, except that when  $\A$ makes a query $(i',S',\sigma')$ (to $\oraclm$) satisfying that 
\begin{enumerate}
	\item[(i)]  $i'=\widehat{i}$, and 
	\item[(ii)]  $\sigma'$ is not output by $\orasign$ on any query $(i',S',*)$,
\end{enumerate}
the challenger returns $\bot$ to $\A$ as a response directly. 

We present the following lemma with a postponed proof. 
\begin{lemma}\label{lem:clm_3_1}
	$|\Pr[\textup{Game}_2 = 1] -  \Pr[\textup{Game}_1 = 1]|\leq\negl$.
\end{lemma}

\noindent\underline{$\textup{Game}_3$:} $\textup{Game}_3$ is identical to   $\textup{Game}_2$, except that when responding to $\A$'s  query $(i',*,*)$ to either $\orasign$ or $\oraclm$  where  $i'=\widehat{i}$, the challenger  samples $r_{\sig}$ uniformly at random   (instead  of using $\prf.\eval$ to compute $r_{\sig}$) during the  $\orasign(\widehat{i},*,*)$ or   $\oraclm(\widehat{i},*,*)$ operations.

Due to the pseudorandomness of $\prf$, we derive  
\begin{eqnarray}\label{eq:clm_end}
	|\Pr[\textup{Game}_3= 1]-\Pr[\textup{Game}_2 = 1]| \leq \negl.
\end{eqnarray}

Once more, we introduce the following lemma, with the proof to be addressed subsequently. 
\begin{lemma}\label{lem:clm_3_2}
	$\Pr[\textup{Game}_3 = 1] \leq\negl$.
\end{lemma}

Combining Eqs. (\ref{eq:clm_begin})-(\ref{eq:clm_end}), Lemma \ref{lem:clm_3_1} and Lemma \ref{lem:clm_3_2}, we obtain 
\begin{eqnarray*}
	\textup{\textbf{Adv}}^{\textup{non-fram}}_{\cmdvs,\A}(\lambda)\leq\negl.
\end{eqnarray*}
Therefore, what remains is to prove the above two lemmas. 

\begin{proof}[of Lemma \ref{lem:clm_3_1}]
	Let $\textsf{Evt}$ be  the event that $\A$ makes a query $(i'=\widehat{i},S',\sigma')$  to $\oraclm$,  where $\sigma'$ is not output by $\orasign$ on any query $(\widehat{i},S',*)$, and $\oraclm(\widehat{i},S',\sigma')$ does not return $\bot$. 
	
	It is straightforward to see that 
	\[|\Pr[\textup{Game}_2 = 1] -  \Pr[\textup{Game}_1 = 1]|\leq\Pr[\textsf{Evt}].\]
		

Assume that $\Pr[\textsf{Evt}]$ is non-negligible.
	
Note that $\sigma'$ can be parsed as $\sigma'=(\sigma'_{\mdvs},\com')$. 

The fact that ``$\oraclm(\widehat{i},S',\sigma')$ does not return $\bot$'' means that $\cla(\pp,\ssk_{\widehat{i}},$ $\{\vpk_{j}\}_{j\in S'},\sigma')$ returns some $\pi'=(\sigma'_{\sig},r'_{\commit})$, where $(\ssk_{\widehat{i}},\{\vpk_{j}\}_{j\in S'})$ are the honestly generated keys. According to the description of the algorithm  $\cla$ as shown in Fig. \ref{fig:algo_cmdvs}, we obtain   $\commit.\dec(\pp_{\commit},\com',r'_{\commit},(\spk_{\widehat{i}},\{\vpk_{j}\}_{j\in S'},$ $\sigma'_{\sig}))=1$. 


If  $\sigma'_{\mdvs}$ is found in some  response,  generated by the challenger for $\A$'s $\orasign$-oracle query $(i',S',*)$   satisfying $i'=\widehat{i}$,   we claim that in this case $\Pr[\textsf{Evt}]=0$. The reason is as follows. Assume that there is such a response $\sigma''=(\sigma'_{\mdvs},\com'')$, returned by the challenger for $\A$'s $\orasign$-oracle query  $(\widehat{i},S',*)$. It is important to note that $\sigma'=(\sigma'_{\mdvs},\com')$ is not output by $\orasign$ on any query $(\widehat{i},S',*)$, so we derive $\com''\neq\com'$. Since $\sigma''=(\sigma'_{\mdvs},\com'')$ is generated via the algorithm $\sign$, according to the description of $\sign$ (as shown in Fig. \ref{fig:algo_cmdvs}), the correctness of $\commit$ guarantees that $\commit.\dec(\pp_{\commit},\com'',r'_{\commit}, (\spk_{\widehat{i}},\sigma'_{\sig}))=1$, 
where $\sigma'_{\sig} = \sig.\sign(\pp_{\sig},\ssk_{\sig,\widehat{i}},(\spk_{\widehat{i}},\sigma'_{\mdvs});\prf.\eval(k_{\widehat{i}},(\spk_{\widehat{i}},\sigma'_{\mdvs},0)))$  and 
$r'_{\commit} = \prf.\eval(k_{\widehat{i}},(\spk_{\widehat{i}},\sigma'_{\mdvs},1))$.  Thus, the fact that $\com''\neq\com'$ implies that $\commit.\dec(\pp_{\commit},\com',r'_{\commit}, (\spk_{\widehat{i}},\{\vpk_{j}\}_{j\in S'},\sigma'_{\sig}))=0$, which further implies that $\oraclm$ will return $\bot$, i.e., $\textsf{Evt}$ will not occur. 

Hence, $\sigma'_{\mdvs}$ is not included in any response  generated by the challenger for $\A$'s $\orasign$-oracle query $(i',S',*)$   satisfying $i'=\widehat{i}$. In other words, the challenger has never computed $\prf.\eval(k_{\widehat{i}},(\spk_{\widehat{i}},\sigma'_{\mdvs},*))$.

Note that $\Pr[\textsf{Evt}]$ is non-negligible. When \textsf{Evt} occurs, during the computation of $\oraclm(\widehat{i},S',\sigma')$ (i.e., $\cla(\pp,\ssk_{\widehat{i}},\{\vpk_{j}\}_{j\in S'},\sigma')$), if we uniformly sample  $r'_{\sig}$ and $r'_{\commit}$ instead of computing them via $\prf.\eval$,  it still holds that $\commit.\dec(\pp_{\commit},\com',r'_{\commit},(\spk_{\widehat{i}},\{\vpk_{j}\}_{j\in S'},\sigma'_{\sig}))=1$ with overwhelming probability, because of the pseudorandomness of $\prf$. Note that the truly random   $r'_{\sig}$ and the original one generated with  $\prf.\eval$ will lead to two different $\sigma'_{\sig}$'s, contradicting the binding property of $\commit$. 

Therefore, the assumption that $\Pr[\textsf{Evt}]$ is incorrect.  We conclude that  $\Pr[\textsf{Evt}]\leq \negl$. \qed

\end{proof}

\begin{proof}[of Lemma \ref{lem:clm_3_2}]
	Consider the following two events. 
\begin{itemize}
	\item $\Evt{1}$: Let $\Evt{1}$ denote the event that $\A$ wins $\textup{Game}_3$, and $\sigma^*$ (from $\A$'s final output)  is not included in any response generated by the challenger for $\A$'s $\orasign$-oracle query $(i',*,*)$   satisfying $i'=\widehat{i}$. 
	\item $\Evt{2}$: Let $\Evt{2}$ denote the event that $\A$ wins $\textup{Game}_3$, and $\sigma^*$ (from $\A$'s final output)  is   included in some  response generated by the challenger for $\A$'s $\orasign$-oracle query $(i',*,*)$   satisfying $i'=\widehat{i}$. 
\end{itemize}
	
	It is straightforward to see that 
	\[\Pr[\textup{Game}_3 = 1] = \Pr[\Evt{1}] + \Pr[\Evt{2}].\]
	

If  both $\Pr[\Evt{1}]$ and  $\Pr[\Evt{2}]$ are  negligible, then we can conclude this proof. So what remains is to prove that both $\Pr[\Evt{1}]$ and  $\Pr[\Evt{2}]$ are negligible
	
	\noindent \underline{\emph{Case 1}: If $\Pr[\Evt{1}]$ is non-negligible:} 
	
	In this case,  we construct a PPT adversary $\B$ attacking the existential unforgeability of   $\sig$ as follows.

    On receiving $(\pp_{\sig},pk)$ from the challenger of the existential unforgeability game, $\B$ firstly samples $\widehat{i}\leftarrow[\poly]$ uniformly at random, and 
    generates $\pp_{\mdvs}$ and $\pp_{\commit}$ via invoking $\mdvs.\setup$ and $\commit.\setup$, respectively. $\B$ also runs  $\prf.\kg$  and $\mdvs.\skg$ to generate $k_{\widehat{i}}$ and  $(\ssk_{\mdvs,\widehat{i}},\spk_{\mdvs,\widehat{i}})$.  
    Thus, the public key of the signer $\widehat{i}$ is  $\spk_{\widehat{i}} = (\spk_{\sig,\widehat{i}} = pk, \spk_{\mdvs,\widehat{i}})$, and the secret key of the signer $\widehat{i}$ is $\ssk_{\widehat{i}} = (k_{\widehat{i}}, \ssk_{\sig,\widehat{i}} = \bot , \ssk_{\mdvs,\widehat{i}})$.
	Then, $\B$ initializes an empty set $Q'$, sends $\pp=(\pp_{\mdvs},\pp_{\sig},\pp_{\commit})$ to $\A$, and  answers $\A$'s oracle queries as follows:
	\begin{itemize}
		\item $\orask$, $\oravk$, $\oraspk$, $\oravpk$, $\oraver$: $\B$ answers these oracle queries as in $\textup{Game}_3$. 
		\item $\orasign(i',S',m')$: When  $i'\neq\widehat{i}$, $\B$ proceeds as in $\textup{Game}_3$. When  $i'=\widehat{i}$,  $\B$ checks whether  there is some tuple  $(S',m',\sigma',\pi') \in Q'$ or not. If so,  $\B$ return $\sigma'$ to $\A$ directly. Otherwise, $\B$ generates a signature as follows.
		$\B$ firstly obtains the public keys $\{\vpk_j\leftarrow\oravpk(j)\}_{j\in S'}$. Then, it computes  $\sigma'_{\mdvs}\leftarrow \mdvs.\sign(\pp_{\mdvs},\ssk_{\mdvs,\widehat{i}},\{\vpk_j\}_{j\in S'},m')$. After that, $\B$ queries its own signing oracle on $(\spk_{\widehat{i}}, \sigma'_{\mdvs})$, and then obtain a signature $\sigma'_{\sig}$ as a response. Subsequently, $\B$ computes $r'_{\commit} \leftarrow \prf.\eval(k_{\widehat{i}}, (\spk_{\widehat{i}}, \sigma'_{\mdvs}),1)$ and $\com' \leftarrow \commit.\com(\pp_{\commit}, (\spk_{\widehat{i}},\{\vpk_j\}_{j\in S'}, \sigma'_{\sig}); r'_{\commit})$. Finally,  $\B$ sets $\sigma' = (\sigma'_{\mdvs},\com')$ and $Q' \leftarrow Q'\cup \{(S',m',\sigma',\pi'=(\sigma'_{\sig},r'_{\commit}))\}$, returns $\sigma'$ to $\A$.
		\item $\oraclm(i',S',\sigma')$: 
		When  $i'\neq\widehat{i}$, $\B$ proceeds as in $\textup{Game}_3$. When  $i'=\widehat{i}$,  $\B$ checks whether  there is some tuple  $(S',m',\sigma',\pi') \in Q'$ or not. If so (i.e., $\sigma'$ is the response corresponding to  the query $(\widehat{i},S',m')$ to $\orasign$),  $\B$ return $\pi'$ to $\A$ directly. Otherwise, $\B$ returns $\bot$ to $\A$. 
	\end{itemize}
	
	Receiving $\A$'s final  output $(i^*,S^*,m^*,\sigma^*,\pi^*)$, $\B$ checks whether $i^*=\widehat{i}$. If not, $\B$ returns uniformly sampled $(m_{\text{rand}},\sigma_{\text{rand}})$ as its final output. Otherwise, $\B$  parses $\sigma^*=(\sigma^*_{\mdvs},\com^*)$ and  $\pi^*=(r^*_{\commit},\sigma^*_{\sig})$, and returns $(m_{\text{sig}} = (\spk_{\widehat{i}},\sigma^*_{\mdvs}),\sigma_{\text{sig}}=\sigma^*_{\sig})$  as its final output. 
	
	That is the construction of $\B$. 
	
	It is clear that $\B$ perfectly simulates $\textup{Game}_3$ for $\A$. Note that $\A$ wins only when $\claver(\pp,\spk_{\widehat{i}},\{\vpk_j\}_{j\in S^*},\sigma^*,\pi^*)=1$, which implies that 
    $\sig.\verify(\pp_{\sig},$ $\spk_{\sig,\widehat{i}},(\spk_{\widehat{i}},\sigma^*_{\mdvs}),\sigma^*_{\sig})=1$.	On the other hand, $\textsf{Evt}_1$ requires that $\sigma^*=(\sigma^*_{\mdvs},\com^*)$  is not included in any response  generated for $\A$'s $\orasign$-oracle query $(\widehat{i},*,*)$. 
    %
	%
	In other words, $m_{\text{sig}} = (\spk_{\widehat{i}},\sigma^*_{\mdvs})$ has never been queried to $\B$'s own signing oracle. 
	Hence, we have 
	\[\textbf{\textup{Adv}}_{\sig,\B}^{\textup{unforg}}(\lambda)\geq\Pr[\Evt{1}].\] 

	Note that we have assumed that $\Pr[\Evt{1}]$ is non-negligible,  so $\textbf{\textup{Adv}}_{\sig,\B}^{\textup{unf}}(\lambda)$ is also non-negligible, contradicting the existential unforgeability of   $\sig$. 


\noindent \underline{\emph{Case 2}: If $\Pr[\Evt{2}]$ is non-negligible:}

Let $q$ be the   number of queries made by $\A$ to the oracle $\orasign$. For each $\theta\in [q]$, let $\Evt{2}^{(\theta)}$ denote the event that $\Evt{2}$ occurs and $\sigma^*$ (from $\A$'s final output)  is   included in some  response generated by the challenger for $\A$'s $\theta^{th}$ $\orasign$-oracle query. 

Since $\Pr[\Evt{2}]$ is non-negligible, there must be some $\widehat{\theta}\in[q]$ such that $\Pr[\Evt{2}^{(\widehat{\theta})}]$ is non-negligible. 

Now,  we construct a PPT adversary $\B'$ attacking the existential unforgeability of   $\sig$. This adversary $\B'$ is identical to  the  above  adversary $\B$ (described in Case 1), except for the following modifications:

On $\A$'s $\widehat{\theta}^{th}$ query $(i',S',m')$ to $\orasign$, 
$\B'$ proceeds as the aforementioned $\B$ does, except that when $i'=\widehat{i}$, during the generation of $\sigma'$, $\B'$ uniformly samples $\sigma'_{\sig}$ at random, instead of querying its own signing oracle on $(\spk_{\widehat{i}}, \sigma'_{\mdvs})$ to obtain $\sigma'_{\sig}$. 

That is the construction of $\B'$. 

By the hiding property of $\commit$, as long as the oracle $\oraclm$ is not queried on $(\widehat{i},S',\sigma')$ by $\A$, where  $\sigma'$ is the response of $\A$'s $\widehat{\theta}^{th}$ query $(\widehat{i},S',m')$ to $\orasign$,  $\B'$'s simulated game and  $\textup{Game}_3$ are indistinguishable, from $\A$'s point of view. On the other hand, if this $(\widehat{i},S',\sigma')$ is queried on by $\A$ to $\oraclm$, $\Evt{2}^{(\widehat{\theta})}$ is impossible to occur. Hence, the probability that $\Evt{2}^{(\widehat{\theta})}$ occurs in $\B'$'s simulation is equivalent to the probability that $\Evt{2}^{(\widehat{\theta})}$ occurs in  $\textup{Game}_3$. Note that when $\Evt{2}^{(\widehat{\theta})}$ occurs, $\B'$ successfully finding a valid forgery, i.e., the message $m_{\text{sig}} = (\spk_{\widehat{i}},\sigma^*_{\mdvs})$ has never been queried to $\B'$'s own signing oracle. So we obtain 
	\[\textbf{\textup{Adv}}_{\sig,\B'}^{\textup{unforg}}(\lambda)\geq\Pr[\Evt{2}^{(\widehat{\theta})}].\]

Since  $\Pr[\Evt{2}^{(\widehat{\theta})}]$ is non-negligible,  so  $\textbf{\textup{Adv}}_{\sig,\B'}^{\textup{unforg}}(\lambda)$ is also non-negligible, contradicting the existential unforgeability of   $\sig$. \qed
\end{proof}
\qed
\end{proof}

\section{A more efficient DDW construction}
\label{sec:efficientddw}

In this section, we delve into a particular case within the MDDW framework, namely DDW, characterized by the existence of only one designated verifier (i.e, $|S|=1$). We propose a more efficient construction by leveraging a concrete DVS scheme.



\subsection{DVS construction}
The DVS scheme is shown in   Fig. \ref{fig:algo_pdvs_two}, which was firstly proposed in \cite{LV04} and extended from \cite{steinfeld2003universal}.

\begin{figure}[!h]
    \small
	\centering
		\fbox{
			\begin{minipage}
				[t]{0.98\textwidth}{

                    \begin{minipage}
                            [t]{0.51\textwidth}{
                                \underline{$\setup(1^{\lambda})$}:\\
                $(e,\G,\G_T,g,p) \leftarrow \genG(1^{\lambda})$\\
                Choose hash functions $\hashf_0$: $\{0,1\}^* \rightarrow \G$\\ $\hashf_1$: $\{0,1\}^* \rightarrow \{0,1\}^l$ \hfill$\sslash l=\poly$\\
                Return $\pp=(e,\G,\G_T,g,p,\hashf_0,\hashf_1)$\\

                                \underline{$\sign(\pp,\ssk,\vpk,m)$}:\\
                                $r\leftarrow \{0,1\}^*$, $h \leftarrow \hashf_0(m,r)$\\
                                $s \leftarrow \hashf_1(e(\vpk,h^{\ssk}))$\\
                                Return $\sigma := (r,s)$
                             }
                         \end{minipage}	
                         \hspace{0.01\textwidth}
                         \begin{minipage}
                                [t]{0.45\textwidth}{
                                    \underline{$\kg(\pp)$}:\\
                                $sk \leftarrow \Z$,
                                $pk := g^{sk}$\\
                                Return $(pk,sk)$\\

                                    \underline{$\verify(\pp,\spk,\vsk,m,\sigma)$}:\\
                                    $(r,s) \leftarrow \sigma$, $h \leftarrow \hashf_0(m,r)$\\
                                    If $(s \neq \hashf_1(e(\spk,h^{\vsk})))$: Return $0$\\
                                    Return $1$\\

                                 }
                             \end{minipage}	
			 	}
			 \end{minipage}


}	
	\captionsetup{justification=centering}
	\caption{Algorithms for a concrete and efficient pseudorandom designated-verifier signature $\dvs$}
	\label{fig:algo_pdvs_two}
\end{figure}

The correctness is trivial, so we omit the analysis of it here. Since there is only one designated verifier, we do not need to consider consistency and there is only one kind of off-the-record property.
For the analysis of the existential unforgeability and off-the-record property, please refer to  \cite{steinfeld2003universal,LV04}.

In addition, we will prove that the scheme supports pseudorandomness (please refer to the definition in Sec. \ref{sec:proof_concrete_pdvs}).
In other words, without the secret key of the designated verifier, it is difficult to distinguish the output of the DVS scheme with a random string. 

Then, we have the following theorem.

\begin{theorem}
    If $\hashf_0$ and $\hashf_1$ are collision resistant hash functions, the DVS scheme $\dvs$ above is existentially unforgeable, off-the-record and pseudorandom in the random oracle model, under the CBDH, DBDH and GBDH assumptions.
    \label{theo:concrete_pdvs}
\end{theorem}

We postpone the proof of Theorem \ref{theo:concrete_pdvs} to Appendix \ref{sec:proof_concrete_pdvs}. In fact, the DVS scheme in Fig. \ref{fig:algo_pdvs_two} supports pseudorandomness for free. Then, when constructing DDW using the framework in Sec. \ref{sec:mddw_const}, we can eliminate   some steps,  which we will elaborate on later.

\subsection{Some concrete DDWs and their comparison}
\label{sec:ddw_cons}
\noindent \textbf{Schemes.} Here we will introduce some concrete DDWs.

\vspace{1mm}
\noindent \underline{\emph{Scheme I.}} In Sec. \ref{sec:generic_construction}, we introduce a framework that shows how to construct a MDDW from a MDVS. When adopting the MDVS \cite{au2014strong}, which is recalled in Appendix. \ref{sec:concrete_scheme}, and letting $|S|=1$, we can obtain a DDW.

\vspace{1mm}
\noindent \underline{\emph{Scheme II.}} On the other hands, incorporating the DVS scheme shown in   Fig. \ref{fig:algo_pdvs_two} into the MDDW framework in Sec. \ref{sec:generic_construction}, we can obtain  another DDW scheme.

\vspace{1mm}
\noindent \underline{\emph{Scheme III.}}
Scheme III can be viewed as an improvement of Scheme II.
 We could further reduce the size of the DVS, when plugging it into our MDDW framework in Sec. \ref{sec:generic_construction}. The details are as follows. 
 
In Fig. \ref{fig:algo_pdvs_two}, the DVS signature contains  a random string $r$ and   a hash value $s$. When plugging the DVS into our framework, we can adopt the hash of previously signed message (i.e., in DDW, the previously signed message could be the prompt or the some tokens just output by the model) to generate the random string $r$ instead. In this way, the   signature could only contain a hash value $s$.

In addition, since the DVS supports pseudorandomness, we do not need to perform XOR computation of the signature and the tokens before sampling new tokens. More exactly, codes in Line \ref{code:one} of Algorithm \ref{alg:Watmdvs} and codes in Line \ref{code:two} of Algorithm \ref{alg:Detdvs} are removed. Then, the algorithms use the signature output by the DVS directly to proceed the computation.

\begin{algorithm}
	\caption{$\watm(\pp,\wsk_i,\dpk_j,\bm{p})$ \label{alg:Watmdvs} }
	\begin{algorithmic}[1]
	\State $\ssk_i\leftarrow \wsk_i$
	\State $\vpk_j\leftarrow\dpk_j$
	
	\State $\bm{t}\leftarrow\epsilon$
	
\While{$|\bm{t}|+\ell+\ell\cdot\textsf{len}_{\text{sig}}<n$}

		\State $\bm{t}\leftarrow\bm{t}||\genmodel_\ell(\bm{p},\bm{t})$
		
		\State \sout{$\hat{\bm{\sigma}}$}~~~~$\bm{\sigma}\leftarrow \dvs.\sign(\pp,\ssk_i,\vpk_j, H_1(\bm{t}[-\ell]))$

		\State \sout{$\bm{\sigma}\leftarrow \hat{\bm{\sigma}} \oplus H_2(\bm{t}[-\ell])$} \label{code:one}
		
		\State $\bm{\sigma}_{\textup{prev}}\leftarrow\epsilon$, 		
		$~\bm{m}\leftarrow\epsilon$
		
		\While{$\bm{\sigma}\neq\epsilon$}		
			
			\State $\sigma\leftarrow \bm{\sigma}[1]$, $\bm{\sigma}\leftarrow\bm{\sigma}[2:]$
			
			\State $\bm{x}\leftarrow\genmodel_\ell(\bm{p},\bm{t})$
			
			\While{$H_3(\bm{m}||\bm{x}||\bm{\sigma}_{\textup{prev}})\neq\sigma$}
			\State $\bm{x}\leftarrow\genmodel_\ell(\bm{p},\bm{t})$
			\EndWhile
		\State $\bm{m}\leftarrow\bm{m}||\bm{x}$, $~\bm{t}\leftarrow\bm{t}||\bm{x}$

		\State $\bm{\sigma}_{\textup{prev}}\leftarrow\bm{\sigma}_{\textup{prev}}||\sigma$
		\EndWhile
	\EndWhile
    \If{$|\bm{t}|<n$}
    \State $\bm{t}\leftarrow\bm{t}||\genmodel_{n-|\bm{t}|}(\bm{p},\bm{t})$
    \EndIf
	\State \Return $\bm{t}$
	\end{algorithmic}
\end{algorithm}

\begin{algorithm}
	\caption{$\detect(\pp,\wpk_i,\dsk_j,\bm{t})$ \label{alg:Detdvs} }
	\begin{algorithmic}[1]
	\State $\spk_i\leftarrow \wpk_i$
	
	\State $\vsk_j\leftarrow\dsk_j$
	
	
	\For{$\mu\in\{1,2,\cdots,n-(\ell+\ell\cdot\textsf{len}_{\text{sig}})\}$}
\State $\widetilde{m}\leftarrow H_1(\bm{t}[\mu:\mu+\ell-1])$,  $\bm{\sigma}\leftarrow\epsilon$, 
$\bm{m}\leftarrow\epsilon$
	
	\For{$\varphi\in[\textsf{len}_{\text{sig}}]$}
	\State $\bm{m}\leftarrow\bm{m}||\bm{t}[\mu+\varphi\cdot\ell:\mu+\varphi\cdot\ell+\ell-1]$
	
	\State $\bm{\sigma}\leftarrow\bm{\sigma}||H_3(\bm{m}||\bm{\sigma})$
     \EndFor
	
	\State \sout{$\hat{\bm{\sigma}}\leftarrow \bm{\sigma} \oplus H_2(\bm{t}[\mu:\mu+\ell-1])$} \label{code:two}
	
	\If{$\dvs.\verify(\pp,\spk_i,\vsk_j,\widetilde{m},$~~ \sout{$\hat{\bm{\sigma}}$}~~~$\bm{\sigma})=1$}
    \State \Return 1
    \EndIf
    \EndFor
	\State \Return 0
	\end{algorithmic}
\end{algorithm}

\vspace{1mm}
\noindent \textbf{Security analysis.} It is clear that DDW of Scheme I and Scheme II are secure (i.e., satisfying completeness, consistency, soundness, distortion-freeness, robustness and off-the-record), since the two schemes adopt the framework in Sec. \ref{sec:generic_construction}. 

As for Scheme III, it is very similar to Scheme II, except that the XOR computation is removed. The changes on framework would have no influence on these security properties, except distortion-freeness. In addition, the $\dvs$ in Fig. \ref{fig:algo_pdvs_two} supports pseudorandomness. Since the hash function $H_1$ in the framework of Sec. \ref{sec:generic_construction} serves as a random oracle.  Thus, the result of XOR computation is also pseudorandom, which is essentially equivalent to the output of $\dvs$. Thus, Scheme III also has distortion-freeness.
\vspace{1mm}

\noindent \textbf{Comparison.} We compare the above three DDWs in   terms of the size of a DVS signature, the time of inserting a DVS signature and the time of verifying a DVS signature.

\begin{table}[htbp]
\centering
	\small
	\caption{Comparison of some  DDWs (from the view of one DVS signature)}
	\label{table:comp}
	\begin{tabular}{lccc}
		\hline
		Schemes &Scheme I &Scheme II &Scheme III\\
		\hline
		Size  &$4 \times |\Z|$ & $|\Z|+l$ &$l$\\
		 Insertion &$1\times \textsf{Sca}+3 \times \textsf{Exp}$ ~~& $1 \times \textsf{Exp}+ 1 \times \textsf{Bil}$ ~~&$1 \times \textsf{Exp}+ 1 \times \textsf{Bil}$\\ 
		Verification  &$2\times \textsf{Sca}+4 \times \textsf{Exp}$ ~~& $1 \times \textsf{Exp}+ 1 \times \textsf{Bil}$ ~~&$1 \times \textsf{Exp}+ 1 \times \textsf{Bil}$\\
		\hline
	\end{tabular}
	\begin{tablenotes}
	\item When considering the same security level, it should hold that $l \approx |\Z|$. For simplicity, we assume that Scheme I works over $\G$, and Scheme II and III works over a bilinear map $e : \G \times \G \rightarrow \G_T$. Then, $\textsf{Sca}$ means one time of the scalar computation over $\G$, $\textsf{Exp}$ means one time of the exponential computation over $\G$, and $\textsf{bil}$ means one time of the bilinear map computation.
 		\end{tablenotes}
\end{table}

In Table \ref{table:comp}, it holds that $l \approx |\Z|$ when we consider the same security level. So, from the view of size of one DVS signature, it is clear that Scheme III outperforms the others. It means that given a fixed length of the output of the model, Scheme III can inserts more DVS signatures into the output of the model. Therefore, it achieves a stronger soundness.

Note that $\textsf{Sca}$ means one time of the scalar computation over $\G$, $\textsf{Exp}$ means one time of the exponential computation over $\G$, and $\textsf{bil}$ means one time of the bilinear map computation. From Table \ref{table:comp}, it is hard to determine which scheme would run in the least time, since it depends more on the actual implementation (a empirical comparison can be found in Sec. \ref{sec:evaluation}). However, if we only focus on Scheme II and Scheme III, we can conclude that Scheme III is more efficient than Scheme II, in terms of inserting one DVS signature and verifying one DVS signature, since Scheme III does not requires the XOR computation.

\subsection{Proof of Theorem \ref{theo:concrete_pdvs}}
\label{sec:proof_concrete_pdvs}

As analysis in the aforementioned section, we omit the analysis of the existential unforgeability and off-the-record property, since you can find it in \cite{steinfeld2003universal,LV04}. Therefore, we only focus on the pseudorandomness here.

\noindent \underline{\emph{Definition of Pseudurandomness.}} The definition of pseudorandomness is as follows.

\begin{definition}{\textbf{\textup{(Pseudorandomness).}}}
	We say that $\dvs$ is pseudorandom, if for all PPT adversary $\A$,  
	\[\advpseudodvs=|\Pr[\gamepseudodvs=1]-\frac{1}{2}| \leq \negl \] 
	where $\gamepseudodvs$ defined in Fig. \ref{fig:dvs_pseudorandom}.  
	\label{def:dvs_pseudo}
\end{definition}

\begin{figure}[t!]\small
	\centering
	\fbox{\shortstack[l]{
			\underline{$\gamepseudodvs$:}\\
			$\pp \leftarrow \setup(1^{\lambda})$, $b\leftarrow \{0,1\}$\\ 
			$b' \leftarrow \A^{\orask,\oravk,\oraspk,\oravpk,\mathcal{O}_{S\text{-}chl}^{(b)},\oraver}(\pp)$\\
			$\qquad$where $\orask,\oravk,\oraspk,\oravpk,\oraver$ are defined as in Fig. \ref{fig:MDVS_oracle} with  $|S|=1$, \\ $\qquad$$\mathcal{O}_{S\text{-}chl}^{(0)}(i,j,m)$ outputs $\orasign(i,j,m)$ ($\orasign$ is also defined in Fig. \ref{fig:MDVS_oracle} with input $S=\{j\}$),\\
   $\qquad\qquad$$\mathcal{O}_{S\text{-}chl}^{(1)}(i,j,m)$ outputs a uniformly sampled $\sigma\leftarrow\mathcal{SG}$,\\
				$\qquad$let  $Q$ denote the set of query-response (i.e., $(i,j,m,\sigma)$) of  $\mathcal{O}_{S\text{-}chl}^{(b)}$,\\
            $\qquad$all queries $j$ to $\oravk$ should satisfy  $(*,j,*,*)\notin Q$, \\
			$\qquad$all queries $(i,j,m,\sigma)$ to $\oraver$ should satisfy  $(i,j,m,\sigma)\notin Q$, \\
            $\qquad$and all queries $(i,j,m)$ to $\mathcal{O}_{S\text{-}chl}^{(b)}$ should satisfy  $j$ is not queried to $\oravk$. \\
			If $b'=b$, then return $1$,\\
			Return $0$
	}		}
	\captionsetup{justification=centering}
	\caption{Game for defining pseudorandomness of DVS}
	\label{fig:dvs_pseudorandom}
\end{figure}  

\noindent \underline{\emph{Proof for pseudorandomness.}} We prove the pseudorandomness in the random oracle model, so we add a random oracle $\mathcal{O}_{RO}$, which inputs a string and outputs a random string. We assume that every time calling signing algorithm $\sign$, $\sign$ would query $str=e(\vpk,h^{\ssk})$ to the random oracle $\mathcal{O}_{RO}$, where $h = \hashf_0(m,r)$.
Then, we prove the pseudorandomness of the DVS scheme in Fig. \ref{fig:algo_pdvs_two} in a sequence of games.

\noindent \underline{$\textbf{G}_0$}: On receiving the security parameter $\lambda$, the game initializes the public parameter $\pp \leftarrow \setup(1^{\lambda})$ and a bit $b\leftarrow \{0,1\}$. 

Then, the challenger answers the queries as follows.
\begin{itemize}
    \item $\mathcal{O}_{RO}(str)$: On receiving a string $str$, it finds $str$ in its table $L_{\textup{ro}}$. If $(str,y)$ in the table, then it returns $y$. Otherwise, it samples a random string $y$, inserts $(str,y)$ into $L_{\textup{ro}}$, and returns $y$.
    \item $\mathcal{O}_{SK}(i)$: On receiving an index $i$, 
    it calls $\skg(\pp)$ to generate a key pair $(\spk_i,\ssk_i)$, and store the key pair. Finally, it returns $(\spk_i,\ssk_i)$.
    \item $\mathcal{O}_{VK}(j)$: On receiving an index $j$, it proceeds as follows. 
    If $(*,j,*,*) \in Q$, where $Q$ is the query-response set of $\mathcal{O}_{S-chl}^{(b)}$, then it aborts. Otherwise, the oracle would return a key pair in the following way. If $j$ has been queried, then search $(\vpk_j,\vsk_j)$. Else, it calls $\vkg(\pp)$ to generate a key pair $(\vpk_j,\vsk_j)$, and store the key pair. Finally, it returns $(\vpk_j,\vsk_j)$.
    \item $\mathcal{O}_{SPK}(i)$: On receiving a index $i$, it calls $\mathcal{O}_{SK}(i)$ to obtain the key pair $(\spk_i,\ssk_i)$, and then returns $\spk_i$.
    \item $\mathcal{O}_{VPK}(j)$: On receiving a index $j$, it calls $\mathcal{O}_{VK}(j)$ to obtain the key pair $(\vpk_j,\vsk_j)$, and then returns $\vpk_j$.
    \item  $\mathcal{O}^{(b)}_{S-chl}(i,j,m)$: On receiving $(i,j,m)$, it proceeds as follows. If 
    $j$ has been queried to $\mathcal{O}_{VK}$, then it aborts. Otherwise,
    it calls $\sign(\pp,\ssk_i,\vpk_j,m)$ to generate the sigma $\sigma$, where $\ssk_i$ is output by $\mathcal{O}_{SK}(i)$ and $\vpk_j$ is output by $\mathcal{O}_{VPK}(j)$. Then it sets $Q \leftarrow \{(i,j,m,\sigma)\} \cup Q $ and returns $\sigma$.
    \item $\mathcal{O}_{V}(i,j,m,\sigma)$: On receiving $(i,j,m,\sigma)$, it proceeds as follows. If $(i,j,m,\sigma) \in Q$, then it aborts. Otherwise, it outputs $b \leftarrow \verify(\pp,\spk_i,\vsk_j,m,\sigma)$, where $\spk_i$ is output by $\mathcal{O}_{SPK}(i)$, $\vsk_{j}$ is output by $\mathcal{O}_{VK}(j)$.
\end{itemize}

Finally, $\A$ outputs a bit $b'$.

It is clear that $\textbf{G}_0$ is identical to $\gamepseudodvs$ when $b=0$ in Fig. \ref{fig:dvs_pseudorandom}. Thus we have \[\Pr[\gamepseudodvs=1|b=0]=\Pr[\textbf{G}_0=1].\]

\noindent \underline{$\textbf{G}_1$}: $\textbf{G}_1$ is similar to $\textbf{G}_0$, except when answering the $\mathcal{O}_{V}$ queries on $(i,j,m,\sigma)$, 
$\textbf{G}_1$ proceeds as follows.

\begin{itemize}
    \item $\mathcal{O}_{V}(i,j,m,\sigma)$: On receiving $(i,j,m,\sigma)$, it proceeds as follows. If $(i,j,m,\sigma) \in Q$, then it aborts. Otherwise, the oracle proceeds:
    \begin{itemize}
        \item If $i$ has not been queried to $\mathcal{O}_{SK}$ and $j$ has not been queried to $\mathcal{O}_{VK}$, then return $0$.
        \item If $i$ has been queried to $\mathcal{O}_{SK}$, then parse $(r,s) \leftarrow \sigma$, and compute $h=\hashf_0(m,r)$, $str=e(\vpk_j,h^{\ssk})$. If $str$ has not been queried to $\mathcal{O}_{RO}$, then return $0$. Otherwise, check if $s=L_{\textup(ro)}(str)$. If equal, return $1$, otherwise $0$.
        \item If $j$ has been queried to $\mathcal{O}_{VK}$, then parse $(r,s) \leftarrow \sigma$, and compute $h=\hashf_0(m,r)$, $str=e(\spk_k,h^{\vsk})$. If $str$ has not been queried to $\mathcal{O}_{RO}$, then return $0$. Otherwise, check if $s=L_{\textup(ro)}(str)$. If equal, return $1$, otherwise $0$.
    \end{itemize}
\end{itemize}

Let $\textsf{evt}$ denote the event, that for a query $(i,j,m,\sigma)$ to $\mathcal{O}_{V}$, it holds that $(i,j,m,\sigma) \not \in Q$ and $\verify(\pp,\spk_i,\vsk_k,m,\sigma)=1$, where $i$ has not been queried to $\mathcal{O}_{SK}$ by the adversary and $j$ has not been queried to $\mathcal{O}_{VK}$ by the adversary.
 
Then, we have \[|\Pr[\textbf{G}_1=1]-\Pr[\textbf{G}_0=1]|=\Pr[\textsf{evt}].\] If $\Pr[\textsf{evt}] \leq \negl$, then $|\Pr[\textbf{G}_1=1]-\Pr[\textbf{G}_0=1]| \leq \negl$. In the following, we show that $\Pr[\textsf{evt}] \leq \negl$.

\begin{itemize}
    \item[Case $1$]: If for all $(i',j',m')$ that have been queried to $\mathcal{O}^{(b)}_{S-chl}$, it holds $(i,j,m) \neq (i',j',m')$, then we say it breaks the existential unforgeability game, of which the successful probability is negligible.
    \item[Case $2$]: If for all $(i',j',m')$ that have been queried to $\mathcal{O}^{(b)}_{S-chl}$, there exists one (denoted as $(\tilde{i},\tilde{j},\tilde{m})$) such that $(i,j,m) = (\tilde{i},\tilde{j},\tilde{m})$, then it implies that $\sigma \neq \tilde{\sigma}$, where $\tilde{\sigma}$ is output by  $\mathcal{O}^{(b)}_{S-chl}(\tilde{m})$. We parse $\sigma=(r,s)$ and $\tilde{\sigma}=(\tilde{r},\tilde{s})$.
    \begin{itemize}
        \item If $r=\tilde{r}$, then it holds that $s \neq \tilde{s}$. However, since $(\tilde{i},\tilde{j},\tilde{m})$ is queried to $\mathcal{O}^{(b)}_{S-chl}$, which invokes $\sign$ to generate the sigma. It obtains $\tilde{s}$, by querying $\tilde{str}=e(\vpk_{\tilde{j}},(\tilde{h})^{\ssk_{\tilde{i}}})$ to random oracle $\mathcal{O}_{RO}$, where $\tilde{h}=\hashf_0(\tilde{m},\tilde{r})=\hashf_0(m,r)=h$. Thus, when the verification algorithm $\verify$ queries $str=e(\spk_i,h^{\vsk_{j}})=e(\vpk_j,h^{\ssk_i})=e(\vpk_{\tilde{j}},(\tilde{h})^{\ssk_{\tilde{i}}})=\tilde{str}$ to the random oracle $\mathcal{O}_{RO}$, it would get $s=\tilde{s}$, which is contradictory to $s \neq \tilde{s}$. Thus, the assumption is not held.
        \item If $r \neq \tilde{r}$, then it with overwhelming probability that $h=\hashf_0(m,r)$ is not equal to any $h'=\hashf_0(m',r')$, where $m'$ is the any message queried to $\mathcal{O}^{(b)}_{S-chl}$ and $r'$ is the corresponding output, since $\hashf_0$ is a collision-resistant hash function. Thus, $(\spk_i,\vpk_j,h) \neq (\spk_{i'},\vpk_{j'},h')$ where $(i',j',m')$ is queried to $\mathcal{O}^{(b)}_{S-chl}$. Note that $str=e(\vpk_j,h^{\ssk_i})$, so $(\spk_i,\vpk_j,h,str)$ is a DBDH tuple. If $str$ has not been queried to $\mathcal{O}_{RO}$, then it with negligible probability that $\hashf_1(str)=s$ is held, since the random oracle $\mathcal{O}_{RO}$ would return a random string for $str$. 
        Then, we analyze the probability that $A$ succeeds in querying $str$ to $\mathcal{O}_{RO}$, given $(\spk_i,\vpk_j,h)$.
        Supposing that $A$ can succeeds in query $str$ with non-negligible probability, we construct another adversary $\B$ to break the CBDH assumption.
        
        Here is the construction of $\B$.
        After receiving $(X,Y,Z)$ from the challenger of CBDH, the adversary $\B$ firstly answers the oracle queries from $\A$, mostly just as the challenger in $\textbf{G}_1$ does. 
        There are some minor difference
        \begin{itemize}
            \item Assuming that $\A$ issues at most $q_0$ queries to $\mathcal{O}_{S-chl}^{(b)}$, $\B$ chooses one of them $(\tilde{i},\tilde{j},\tilde{m})$, and sets $\spk_{\tilde{i}}=X$, $\vpk_{\tilde{j}}=Y$. To answer this query, $\B$ randomly chooses a random string $\tilde{r}$ and sets $\hashf_0(\tilde{r},\tilde{m}) = \tilde{h}=g^{\alpha}$ (here, $\hashf_0$ also serves as a random oracle). Then, $\B$ computes $\tilde{str}=e(\spk_{\tilde{i}},\vpk_{\tilde{j}})^{\alpha}$ and queries $\tilde{str}$ to the oracle $\mathcal{O}_{RO}$, obtaining $\tilde{s}$. Finally, $\B$ returns $\tilde{\sigma}=(\tilde{r},\tilde{s})$.
            \item When $\A$ issues queries to $\mathcal{O}_{RO}$ for $\hashf_0$, if there are some  queries (assume that they are $q_1$ such queries) in the form of $(*,\tilde{m})$, then $\B$ chooses one of them (e.g., $(r',\tilde{m})$) and sets $\hashf_0(r',\tilde{m})=Z$.
        \end{itemize}
        After $\A$ issues a query $(i,j,m,\sigma)$ to $\mathcal{O}_{V}$, which makes case 2 happens, if $(i,j,m)=(\tilde{i},\tilde{j},\tilde{m})$ and $\hashf_0(r,\tilde{m})=Z$ where $\sigma=(r,s)$, then $\B$ finds the query in $L_{\textup{ro}}$ such that $\hashf_1(str)=s$. Finally, $\B$ outputs $str$ as its output.
        That is construction of $\B$.

        Then, we analyze the probability of $\B$.
        \[\Pr[\B~\textup{succeeds}] \geq \frac{1}{q_0\cdot q_1} \Pr[\textup{Case 2 happens}].\]
        If $\Pr[\textup{Case 2 happens}]$ is non-negligible, then $\Pr[\B~\textup{succeeds}]$ is also non-negligible, which is contradictory to that CBDH is thought a hard problem. Thus, the assumption is not held, so $\Pr[\textup{Case 2 happens}]$ is negligible.

        
    \end{itemize}
    In all, the probability of Case $2$ is negligible.
\end{itemize}
Therefore, $\Pr[\textsf{evt}] \leq \negl$, which implies that \[|\Pr[\textbf{G}_1=1]-\Pr[\textbf{G}_0=1]| \leq \negl.\]


\noindent \underline{$\textbf{G}_2$}: $\textbf{G}_2$ is similar to $\textbf{G}_1$, except when answering the  $\mathcal{O}^{(b)}_{S-chl}$ queries on $m$, it proceeds as follows.

\begin{itemize}
    \item $\mathcal{O}^{(b)}_{S-chl}(i,j,m)$: On receiving $(i,j,m)$, it proceeds as follows. If 
    $j$ has been queried to $\mathcal{O}_{VK}$, then it aborts. Otherwise, it randomly chooses $\sigma := (r,s) \leftarrow \{0,1\}^* \times \{0,1\}^l$, sets $Q \leftarrow \{(i,j,m,\sigma)\} \cup Q $, and returns $\sigma$.
\end{itemize}

Suppose that the adversary makes $\ell$ queries to oracle $\mathcal{O}^{(b)}_{S-chl}$. Denote $\textbf{G}_{2,k}$ ($k \in \{0,1,\cdots,\ell\}$) as the game, where for the first $k$ queries to oracle $\mathcal{O}^{(b)}_{S-chl}$, the oracle $\mathcal{O}^{(b)}_{S-chl}$ proceeds these queries as it does in $\textbf{G}_2$, and then for the left $(\ell-k)$ queries to oracle $\mathcal{O}^{(b)}_{S-chl}$, the oracle $\mathcal{O}^{(b)}_{S-chl}$ proceeds these queries as it does in $\textbf{G}_1$. Thus, $\textbf{G}_{2,0} = \textbf{G}_1$ and $\textbf{G}_{2,\ell} = \textbf{G}_2$. We have the following lemma.

\begin{lemma}
    For every $k \in [\ell]$, it holds that $|\Pr[\textbf{\textup{G}}_{2,k-1}=1]-\Pr[\textbf{\textup{G}}_{2,k}=1]| \leq \negl$.
    \label{lemma:gtwoj}
\end{lemma}

\begin{proof}[of Lemma \ref{lemma:gtwoj}]
We denote the $\ell$ queries as $((i_1,j_1,m_1),\ldots,$ $(i_{\ell},j_{\ell},m_{\ell}))$. Then, we show that distinguishing between $\textbf{G}_{2,k-1}$ and $\textbf{G}_{2,k}$ is difficult.

In $\textbf{G}_{2,k-1}$, when querying $(i_k,j_k,m_k)$ to $\mathcal{O}^{(b)}_{S-chl}$, the oracle samples a random string $r_k \leftarrow \{0,1\}^*$, computes $h_k \leftarrow \hashf_0(m_k,r_k)$ and $str_k = e(\vpk_{j_k},(h_k)^{\ssk_{i_k}})$, queries $str_k$ to the random oracle $\mathcal{O}_{RO}$, obtaining $s_k$ and setting $\sigma_k = (r_k,s_k)$, sets $Q \leftarrow \{(i_k,j_k,m_k,\sigma_k)\} \cup Q $ and returns $\sigma_k$.

Then, we define another game $\textbf{G}'$, which is similar to $\textbf{G}_{2,k-1}$, except that $\textbf{G}'$ queries a random string $str'$ to the random oracle $\mathcal{O}_{RO}$, obtaining $s_k$ (i.e., the random oracle assigns the same $s_k$ to the random oracle query).

It is clear the only difference is that $str_k=e(\vpk_{j_k},(h_k)^{\ssk_{i_k}})$ and $str' \leftarrow \{0,1\}^{|str_k|}$. Given $(g,\spk_{i_k},\vpk_{j_k})$, distinguishing $str_k$ and $str'$ can be reduced to DBDH problem, which is thought a hard problem. Thus, we have $|\Pr[\textbf{G}_{2,k-1}=1]-\Pr[\textbf{G}'=1]| \leq \negl$. 

Note that, the distribution of a string output by the random oracle $\mathcal{O}_{RO}$ on a random query $str'$, is equivalent to the uniform distribution. Thus, $\textbf{G}'$ is identical to $\textbf{G}_{2,k}$. Therefore, we have $\Pr[\textbf{G}'=1]=\Pr[\textbf{G}_{2,k}=1]$.

Therefore, it holds that $|\Pr[\textbf{G}_{2,k-1}=1]-\Pr[\textbf{G}_{2,k}=1]| \leq \negl$.
\end{proof}

Applying Lemma \ref{lemma:gtwoj}, we have $|\Pr[\textbf{G}_{1}=1]-\Pr[\textbf{G}_{2}=1]|=|\Pr[\textbf{G}_{2,0}=1]-\Pr[\textbf{G}_{\ell}=1]| \leq \negl$.

\noindent \underline{$\textbf{G}_3$}: $\textbf{G}_3$ is similar to $\textbf{G}_2$, except when answering the queries $(i,j,m,\sigma)$ to $\mathcal{O}_{V}$, it proceeds as follows.

\begin{itemize}
    \item $\mathcal{O}_{V}(i,j,m,\sigma)$: On receiving $(i,j,m,\sigma)$, if $(i,j,m,\sigma) \in Q$, it aborts. Otherwise, it 
    returns $b \leftarrow \verify(\pp,\spk_i,\vsk_j,m,\sigma)$.
\end{itemize}

Note that the analysis of the indistinguishability between $\textbf{G}_2$ and $\textbf{G}_3$ is similar to that between $\textbf{G}_0$ and $\textbf{G}_1$. Thus, we have \[|\Pr[\textbf{G}_{3}=1]-\Pr[\textbf{G}_{2}=1]| \leq \negl.\]

In fact, $\textbf{G}_3$ is identical to the game $\gamepseudodvs$ when $b=1$.

Therefore, we have $|\Pr[\gamepseudodvs=1|b=0]-\Pr[\gamepseudodvs=1|b=1]|=|\Pr[\textbf{G}_0=1]-\Pr[\textbf{G}_3=1]| \leq \negl$, which implies that the DVS scheme $\dvs$ constructed in Fig. \ref{fig:algo_pdvs_two} is pseudorandom.

\end{document}